\documentclass{statsoc}
\usepackage{etoolbox}
\makeatletter
\patchcmd{\@makecaption}
  {\parbox}
  {\advance\@tempdima-\fontdimen2} 
  {}{}
\makeatother 

\usepackage[a4paper,top=1.25in, bottom=1.25in, left=1in, right=1in]{geometry}

\usepackage{graphicx}
\usepackage{bm}
\usepackage{dsfont}
\usepackage{amsmath}
\usepackage{amssymb}
\usepackage{graphicx}
\usepackage{natbib}
\usepackage{enumitem}
\usepackage{mathtools}
\usepackage{mathrsfs}
\mathtoolsset{showonlyrefs}
\usepackage{tikz}
\usepackage{anyfontsize}
\usepackage{algorithm}
\usepackage{algorithmic}
\usepackage{hyperref}
\hypersetup{
    colorlinks,%
    citecolor=blue,%
    filecolor=blue,%
    linkcolor=black,%
    urlcolor=blue
}

\newtheorem{definition}{Definition}
\newtheorem{theorem}{Theorem}
\newtheorem{lemma}{Lemma}

\newtheorem{remark}{Remark}

\newcommand{\ind}{\mathds{1}}

\newcommand{\indep}{\rotatebox[origin=c]{90}{$\models$}}
\newcommand{\sign}{\mathrm{sign}}
\newcommand{\fdr}{\textnormal{FDR}}

\newcommand{\eqd}{\stackrel{\textnormal{d}}{=}}

\def \defn {\,:=\,}

\newcommand{\kl}{\textnormal{KL}}
\newcommand{\eps}{\varepsilon}

\usepackage{xspace}

\newcommand{\stepa}[1]{\overset{\rm (a)}{#1}}

\newcommand{\given}{{\,|\,}}
\newcommand{\biggiven}{\,\big{|}\,}
\newcommand{\Biggiven}{\,\Big{|}\,}
\newcommand{\bigggiven}{\,\bigg{|}\,}

\def\@#1\@{\begin{align}#1\end{align}}
\def\$#1\${\begin{align*}#1\end{align*}}

\definecolor{myblue}{rgb}{.8, .8, 1}
\definecolor{mathblue}{rgb}{0.2472, 0.24, 0.6} 
\definecolor{mathred}{rgb}{0.6, 0.24, 0.442893}
\definecolor{mathyellow}{rgb}{0.6, 0.547014, 0.24}

\newcommand{\tX}{{\tilde{X}}}

\newcommand{\calH}{{\mathcal{H}}}

\newcommand{\calN}{{\mathcal{N}}}

\newcommand{\calS}{{\mathcal{S}}}

\newcommand{\calV}{{\mathcal{V}}}
\newcommand{\calW}{{\mathcal{W}}}

\newcommand{\EE}{\mathbb{E}}
\newcommand{\PP}{\mathbb{P}}

\newcommand{\RR}{\mathbb{R}}

\newcommand{\bX}{\mathbf{X}}
\newcommand{\bY}{\mathbf{Y}}

\graphicspath{{./figs/}}

\newcommand{\kn}{\textnormal{kn}}
\newcommand{\ebh}{\textnormal{ebh}}
\newcommand{\ci}{\textnormal{ci}}
\newcommand{\cst}{\textnormal{cst}}

\newcommand{\bW}{{\bm W}}

\let\emptyset\varnothing
\let\hat\widehat
\let\tilde\widetilde
\title[Derandomized knockoffs]{
Derandomized knockoffs: leveraging e-values for\\
false discovery rate control}
\author{Zhimei Ren}
\address{Department of Statistics and 
Data Science, Wharton School, 
University of Pennsylvania, PA, USA.}
\email{zren@wharton.upenn.edu}

\author[Zhimei Ren and Rina Foygel Barber]{Rina Foygel Barber}
\address{Department of Statistics, University of Chicago, IL, USA.}
\email{rina@uchicago.edu}

\begin{document}
\begin{abstract}
Model-X knockoffs is a flexible wrapper method for high-dimensional
regression algorithms, which provides guaranteed control of the 
false discovery rate (FDR).
Due to the randomness inherent to the method, different runs of model-X knockoffs on the same dataset
often result in different sets of selected variables, which is undesirable
in practice. In this paper, we introduce a methodology for derandomizing 
model-X knockoffs with provable FDR control. The key insight of our proposed method lies
in the discovery that the knockoffs procedure is 
in essence an e-BH procedure. We make use of 
this connection, and derandomize model-X knockoffs by
aggregating the e-values resulting from 
multiple knockoff realizations.
We prove
that the derandomized procedure controls the FDR
at the desired level, without any additional
conditions (in contrast, previously proposed methods for
derandomization are not able to guarantee FDR control).
The proposed method is evaluated with
numerical experiments, where we find that the
derandomized procedure achieves comparable power and
dramatically decreased selection variability when compared
with model-X knockoffs. 
\end{abstract}

\section{Introduction}
\label{sec:intro}
In high-dimensional datasets, it is common to have measurements of a large number of potential explanatory features,
of which relatively few are informative for predicting the target response. The problem of identifying these few relevant features among
the many candidates, also known as the {\em variable selection problem}, is often framed as a test of conditional independence:
for which features $X_j$ is it true that $Y\!\not\!\!\!\indep X_j \mid X_{-j}$? Here, $Y$ denotes the response variable and $X_{-j}$ 
denotes all measured features aside from $X_j$---effectively, this question asks whether $X_j$ carries information for predicting $Y$,
beyond what is already contained in the set of remaining features $X_{-j}$. 

The {\em knockoff filter} \citep{barber2015controlling,candes2018panning} is a framework for selecting a set of $X_j$'s 
that are likely relevant, with guaranteed control of the false discovery rate (FDR). 
It operates by constructing a  knockoff copy $\tilde{X}_j$ of each candidate feature $X_j$, after which the response $Y$
and (original and knockoff) features $X_1,\dots,X_p,\tilde{X}_1,\dots,\tilde{X}_p$ are given as input to an arbitrary variable selection procedure.
By examining whether the procedure chooses substantially more original variables (the $X_j$'s) than knockoffs (the $\tilde{X}_j$'s), 
we may infer whether the procedure is successfully controlling the FDR.

Since the $\tilde{X}_j$'s are drawn randomly, the resulting output $\calS_{\kn}$ of the method is random as well---that is, multiple
runs of knockoffs on  the same observed dataset can in general lead to different selected sets $\calS_{\kn}$. Empirically, it has been observed
that the output can be highly variable from one run to another, which is potentially an undesirable property. 
To address this, \cite{ren2021derandomizing} propose the ``derandomized knockoffs''. After running knockoffs $M$
times on the given dataset, their procedure computes
\begin{align}\label{eqn:old_prob}
\Pi_j = \frac{1}{M} \sum_{m=1}^M \ind \big\{j \in \calS_{\kn}^{(m)}\big\}
\end{align}
for each feature $X_j$, where $\calS_{\kn}^{(m)}$ is the selected set of features on the $m$-th run of the knockoffs method. 
The final selected set is then given by $\calS_{\kn\textnormal{--derand}} = \{j : \Pi_j\geq \eta\}$, all features $X_j$ exceeding some threshold probability of selection for a random run of knockoffs.
 For this method, \cite{ren2021derandomizing} establish a guaranteed bound on the expected number of false discoveries,
but it appears impossible to prove a bound directly on the FDR.

\paragraph{Our contributions} In this work, we find that, with a simple and natural modification in the construction of the derandomized knockoffs procedure, we can restore the guarantee of
FDR control. Specifically, we will consider a {\em weighted} probability of selection, replacing $\Pi_j$ with
\begin{equation}\label{eqn:informal_intro}\frac{1}{M} \sum_{m=1}^M 
\textnormal{weight}_j^{(m)} \cdot \ind \big\{j \in \calS_{\kn}^{(m)}\big\},\end{equation}
where, informally, we choose a lower weight if
many features were selected in the $m$-th run.
In the case of a single run of knockoffs ($M=1$), our new procedure reduces to the original knockoff filter, while for large $M$,
the output is derandomized (i.e., as $M\rightarrow\infty$, the set $\calS_{\kn\textnormal{--derand}}$ becomes a deterministic function of the observed data).
Our method builds on the recent e-BH procedure of \cite{wang2022false} (a generalization of the Benjamini--Hochberg algorithm for FDR control \citep{benjamini1995controlling}), which allows us to find an adaptive, FDR-controlling threshold for the weighted selection probabilities
computed in~\eqref{eqn:informal_intro}. 
 By leveraging these tools, we are able to provide 
a version of the knockoff filter that offers derandomization without losing the benefit of a rigorous guarantee on the FDR.

\subsection{Problem setup}
Suppose there are $p$ explanatory variables $X = (X_1,X_2,\ldots,X_p)$
and a response $Y$, where $(X,Y)$ are jointly sampled from some distribution $P_{XY}$.
For each $j\in  [p] \defn \{1,2,\ldots,p\}$, we wish to test the hypothesis
\@\label{eq:ci_test}
H_j: Y~\indep~X_j \given X_{-j}.
\@
We call a feature $j$ a {\em null} if $H_j$ 
is true, and a {\em non-null} otherwise. We will write $\calH_0 = \{j :H_j \mbox{ is true}\}$
and $\calH_1 = [p]\backslash\calH_0$ to denote the set of nulls and the set of non-nulls, respectively.

Imagine now we have $n$ samples $(X_i,Y_i)\stackrel{\textnormal{iid}}{\sim}P_{XY}$,
 and we assemble the 
covariates into a matrix $\bX \in \RR^{n \times p}$
and the responses into a vector $\bY \in \RR^n$. 
A multiple testing procedure applied to $(\bX,\bY)$
then produces a set $\calS$ of selected variables. The goal here 
is to include in this set as many non-nulls as possible
while controlling the false discovery rate
\[
\fdr \defn \EE \Big[\frac{|\calS \cap \calH_0|}{|\calS| \vee 1}\Big],
\]
where $a\vee b = \max(a,b)$. In this work, we consider the {\em model-X framework},
where we (approximately) know the marginal distribution $P_X$ of the covariates $X$, but do not assume any knowledge of the model of 
$Y \given X$. For example, in many applications, we may have ample unlabeled data (i.e., 
observations of $X=(X_1,\dots,X_p)$) which may be used for estimating $P_X$, but relatively little labeled data (i.e., observations of labeled
pairs $(X,Y)$); see \citet{candes2018panning} for more discussion of this framework.

\subsection{Background: model-X knockoffs}
The model-X knockoff filter, or knockoffs for short, is a multiple 
testing procedure that provably controls the FDR under the model-X 
framework. Given the dataset $(\bX,\bY)$ as well as knowledge of $P_X$,
the knockoffs procedure starts by generating a knockoff copy $\tilde{\bX}$
satisfying
\begin{equation}\label{eqn:pairwise_exch}
\big(\bX_j, \tilde{\bX}_j, \bX_{-j},\tilde{\bX}_{-j}\big)
 \eqd
\big(\tilde{\bX}_j, \bX_j, \bX_{-j},\tilde{\bX}_{-j}\big)
\end{equation}
for each $j$,
where $\eqd$ denotes equality in distribution. This condition requires $\tilde{\bX}$ 
to depend on $\bX$ (so that any dependence between features $\bX_j$ and $\bX_k$ is mimicked
by dependence between, e.g., $\tilde{\bX}_j$ and $\bX_k$), but $\tilde{\bX}$ is constructed
independently of $\bY$, i.e., $\bY ~\indep~ \tilde{\bX} \given \bX$.
(See~\citet{candes2018panning,
sesia2019gene,romano2020deep,bates2021metropolized,spector2022powerful}
for discussion and examples of knockoff generation.)
Having sampled $\tilde \bX$, the knockoffs procedure proceeds
to compute feature importance statistics $W \in \RR^p$ 
using the augmented dataset $([\bX,\tilde{\bX}],\bY)$:
\[
W = \calW\big( [\bX, \tilde{\bX}], \bY \big),
\]
where $\calW(\cdot)$ is an algorithm evaluating the 
importance of the features, with the property that 
swapping $\bX_j$ and $\tilde{\bX}_j$ flips the sign 
of $W_j$, and a larger value of $W_j$ suggests  
evidence against the null. For instance, to compute
the Lasso coefficient-difference (LCD) statistic proposed by \citet[Eqn.~(3.7)]{candes2018panning},
we can run
the cross-validated Lasso on $\bY \sim [\bX,\tilde{\bX}]$;
let $\beta_j$ denote the resulting coefficient of $\bX_j$ 
and $\tilde{\beta}_j$ that of $\tilde{\bX}_j$, and define
$W_j = |\beta_j| - |\tilde{\beta}_j|$ for each $j\in [p]$. 
Since the knockoffs act as a control group for the real features,
if $\bX_j$ is a null then the Lasso is equally likely to select $\bX_j$ or $\tilde{\bX}_j$,
and moreover, $W_j$'s distribution is symmetric around 0.

The final selected set of features is then given by 
\begin{equation}\label{eqn:define_knockoff}
\calS_{\kn} \defn \{j: W_j \ge T\}\textnormal{ where }T \defn \inf\bigg\{t >0: \frac{1 + \sum_{j\in[p]} \ind \{W_j \le -t\}}
{\sum_{j\in[p]} \ind\{W_j \ge t\}} \le \alpha\bigg\}.
\end{equation}
\citet{barber2015controlling,candes2018panning} prove that this selected set satisfies $\fdr\leq \alpha$.

\subsection{Background: e-values and the e-BH procedure}
The concept of e-values~\citep{vovk2021values} is another 
useful tool for statistical inference and multiple testing
in general. Given a null hypothesis, we call a non-negative
random variable $E$ an ``e-value'' if $\EE[E]\le 1$ 
under the null (in contrast, a p-value is a random variable 
$P \in [0,1]$ such that $\PP(P \le t) \le t$ for any $t\in[0,1]$
under the null). For the e-value, a large value shows evidence 
against the null, and hence the null hypothesis is rejected when 
the e-value exceeds a threshold. For example, if the goal is
to test a hypothesis at level $\alpha$, we can reject the null
hypothesis when $E \ge 1/\alpha$. The probability of making a 
type-I error is then 
\[
\PP(E \ge 1/\alpha) \le
\alpha \cdot \EE[E] \le \alpha,
\]
where the first step applies Markov's inequality and the second
 follows from the definition of
e-values.

In the context of multiple hypothesis testing, let $e_j$ be an e-value
associated with a null hypothesis $H_j$.
With e-values $e_1,e_2,\ldots,e_p$,~\citet{wang2022false}
propose the e-BH procedure that achieves FDR control when testing
$H_1,H_2,\ldots,H_p$ simultaneously. The e-BH procedure operates
in a similar way to the BH procedure~\citep{benjamini1995controlling}: the 
rejection set (i.e., the selected set of discoveries) is given by 
\begin{equation}\label{eqn:define_ebh}\calS_{\ebh} = \Big\{j : e_j \geq \frac{p}{\alpha\hat{k}}\Big\}\textnormal{ where }
\hat{k} = \max\Big\{k \in[p]: e_{(k)} \ge \frac{p}{\alpha k}\Big\},\end{equation}
where $e_{(1)}\ge \dots\ge e_{(p)}$ denotes the order statistics of the $e_j$'s, and where 
we take the convention
that if this latter set is empty then we set $\hat{k}=0$ (and so $\calS_{\ebh} = \emptyset$).
\citet{wang2022false} prove that the e-BH procedure satisfies
 $\fdr \leq \alpha \cdot |\calH_0|/p\leq \alpha$. Importantly, this result allows for arbitrary dependence among the $e_j$'s.

\subsection{Additional related work}
Previously, many attempts have been made to derandomize
knockoffs. As mentioned earlier,~\citet{ren2021derandomizing}
propose a derandomizing scheme controlling the expected 
number of false discoveries;~\citet{nguyen2020aggregation}
introduce a aggregation method aiming at FDR control, but
under strong assumptions such as that
the null feature importance statistics are i.i.d.
In a parallel line of work,~\citet{emery2019controlling,gimenez2019improving} consider constructing multiple 
{\em simultaneous} knockoffs to improve the stability of 
knockoffs---this is in contrast to our attempt to
aggregate multiple {\em independent} knockoff copies.
The idea of using e-values for derandomization can also
be found in~\citet{vovk2020note,wasserman2020universal}.

Broadly speaking, the process of derandomization via computing each 
feature's selection probability over random runs
of the knockoffs method, can be viewed as a special case 
of stability selection \citep{meinshausen2010stability} (see also
\citet{liu2010stability,shah2013variable} for related methods). 
More recently,~\citet{dai2022false,dai2023scale} also consider derandomization
(over multiple sample splits) with FDR control via computing each feature's
weighted selection frequency (this is termed the ``inclusion rate'' in 
their paper). Their procedure also takes the form of~\eqref{eqn:informal_intro}
but with a different choice of the weight terms, $\textnormal{weight}_j^{(m)}$,
and with asymptotic rather than finite-sample FDR control guarantees (we will compare
the definitions of the procedures, and the different results, in more detail in Section~\ref{sec:procedure} below).

\section{Knockoffs as an e-BH procedure}
\label{sec:equivalence}
Our first main result shows that 
the two multiple testing procedures introduced
in Section~\ref{sec:intro} can be unified through a certain
perspective---the knockoffs procedure is in fact {\em equivalent}
to a (relaxed) e-BH procedure with a class of properly
defined e-values.

\subsection{The equivalence between the knockoffs and the e-BH}
To see why knockoffs is equivalent to an e-BH 
procedure, we first define a set of relaxed
e-values. Recall that in the knockoffs procedure, 
we have the feature importance statistics $W$
and the stopping time $T$. For each $j \in [p]$, 
define
\@\label{eq:defn_e_vals}
e_j \defn p \cdot \frac{ \ind \{W_j \ge T\}}
{1 + \sum_{k \in [p]} \ind \{W_k \le -T\}}.
\@
We will now see that running e-BH on the $e_j$'s is exactly equivalent to running knockoffs.
\begin{theorem}\label{thm:equiv}
Let $\calS_{\kn}$ be the set of selected features for the knockoff procedure~\eqref{eqn:define_knockoff},
and let $\calS_{\ebh}$ be the set of selected features for the e-BH procedure~\eqref{eqn:define_ebh}
applied to $e_1,\dots,e_p$, where $e_j$ is defined in~\eqref{eq:defn_e_vals}. Then $\calS_{\kn}=\calS_{\ebh}$.
\end{theorem}
\begin{proof}
Taking $K = |\calS_{\kn}|=\sum_{k\in[p]}\ind\{W_k \ge T\}$, we see that $\frac{1 + \sum_{k\in[p]}\ind\{W_k \le -T\}}{K}\leq \alpha$
by definition of the knockoffs threshold $T$~\eqref{eqn:define_knockoff},
and so for all $j\in\calS_{\kn}$, $e_j$ takes the value
$\frac{p}{1 + \sum_{k\in[p]}\ind\{W_k \le -T\}} \geq p/(\alpha K)$. Therefore,
 $e_{(K)}  \geq p/(\alpha K)$, and so we must have $\hat{k} \geq K$
  when we run the e-BH procedure~\eqref{eqn:define_ebh}. This proves that $j\in\calS_{\ebh}$ for all
$j\in\calS_{\kn}$. Conversely, if $j\not\in\calS_{\kn}$, then $W_j < T$ and so $e_j=0$, which means that $j$ cannot be selected by the e-BH procedure, i.e., $j\not\in\calS_{\ebh}$.
\end{proof}

\subsection{A relaxation of the e-BH procedure}
While the theorem above shows that knockoffs gives identical output to the e-BH procedure,
we note that it does not yet give an alternative explanation for why knockoffs controls the FDR---this is because
the quantities $e_j$ defined in~\eqref{eq:defn_e_vals} have not been shown to be e-values. Indeed, while $e_j\geq 0$ holds by definition,
it may not be the case that $\EE[e_j]\leq 1$ for all $j\in\calH_0$. Instead, the $e_j$'s satisfy a more relaxed criterion.
A key step in the proof of the FDR control property of the knockoff filter is the bound \citep{barber2015controlling,candes2018panning}
\begin{equation}\label{eqn:kn_ratio}
\EE\left[\frac{\sum_{j\in\calH_0} \ind\{W_j\ge T\}}{1 + \sum_{j\in\calH_0} \ind\{W_j\le -T\}}\right] \leq 1.\end{equation}
Intuitively, this arises from the fact that
each null $j\in\calH_0$ is equally likely to have $W_j\ge T$ as to have $W_j \le -T$,
because the knockoff copy $\tilde{\bX}_j$ acts as a control group for the null feature $\bX_j$. As a result, we can immediately see that
the $e_j$'s defined in~\eqref{eq:defn_e_vals} satisfy
\begin{equation}\label{eqn:evals_relaxed}
\sum_{j\in\calH_0} \EE[e_j] \leq p.\end{equation}
Clearly, this condition is strictly weaker than requiring $\EE[e_j]\leq 1$ for all $j\in\calH_0$.
Nonetheless, as pointed out by~\citet{wang2022false}, this condition is sufficient
to bound FDR in the e-BH procedure: 
the $e_j$'s satisfying~\eqref{eqn:evals_relaxed} can be thought of as weighted e-values; that is, we 
can write $e_j = \EE[e_j]\cdot \bar{e}_j$,
where  $\bar{e}_j$'s are strict e-values with 
$\EE[\bar{e}_j] = 1$ and $\EE[e_j]$'s are weights.
This is summarized in the following theorem.
\begin{theorem}[{\citealt{wang2022false}}]
\label{thm:relaxed_ebh}
Suppose the values $e_1,e_2,\ldots,e_p$ satisfy condition~\eqref{eqn:evals_relaxed}.
Then the selected set  $\calS_{\ebh}$ of the e-BH procedure~\eqref{eqn:define_ebh} satisfies $\fdr\leq \alpha$.
\end{theorem}
From this point on, then, we will refer to any $e_1,\dots,e_p$ satisfying~\eqref{eqn:evals_relaxed} as e-values,
even though the original condition (i.e., $\EE[e_j]\leq 1$ for $j\in\calH_0$) is strictly stronger and may not be satisfied.

Before proceeding, we remark that the knockoff 
e-values are close to sharp in the following sense.

\begin{remark}[Sharpness of the knockoff e-values]
Given e-values $e_1,e_2,\ldots,e_p$, let $\calS_{\ebh}$
denote the set of features selected by the e-BH procedure, 
and its FDR is given by 
\$
\fdr = \EE\Big[ 
\frac{\sum_{j \in \calH_0} 
\ind\{e_j \ge \frac{p}{\alpha |\calS_{\ebh}|}\}}
{|\calS_{\ebh}| \vee 1}\Big].
\$
The proof of FDR control for the e-BH procedure 
makes use of the inequality
\@\label{eq:ebh_ineq}
\ind\Big\{e_j \ge \frac{p}{\alpha|\calS_{\ebh}|}\Big\}
\le \frac{e_j \alpha |\calS_{\ebh}|}{p},
\@
which is tight when $e_j \in \{p/(\alpha |\calS_{\ebh}|),0\}$.
Meanwhile, the knockoff e-values defined in~\eqref{eq:defn_e_vals} are 
either $0$ or
\@\label{eq:sharpness_1}
\frac{p}{1+\sum_{j\in[p]} \ind\{W_j \le -T\}} \ge 
\frac{p}{\alpha|\calS_{\kn}|} = 
\frac{p}{\alpha|\calS_{\ebh}|},
\@
where the first inequality is by the definition of $T$.
For any $t < T$, again by the definition of $T$, we have that
\@\label{eq:sharpness_2}
\frac{p}{1+\sum_{j\in[p]} \ind\{W_j \le -t\}} <
\frac{p}{\alpha\sum_{j\in[p]} \ind \{W_j \ge t\}} \le 
\frac{p}{\alpha|\calS_{\kn}|} = 
\frac{p}{\alpha|\calS_{\ebh}|}.
\@
Combining~\eqref{eq:sharpness_1} and~\eqref{eq:sharpness_2}, we can see 
when $\sum_{j \in [p]}\ind\{W_j \le -T\}$ is reasonably large, 
the nonzero e-values will be close to $p/(\alpha|\calS_{\ebh}|)$, and 
therefore the inequality in~\eqref{eq:ebh_ineq} is close to tight.

\end{remark}
\section{Derandomizing knockoffs}
\label{sec:method}
One major advantage of using e-values for multiple testing
is that validity depends only on the expected values $\EE[e_j]$---in particular,
no assumptions are needed on the dependence structure among the $e_j$'s \citep{wang2022false}.
As observed by \cite{vovk2021values}, 
the fact that the average of multiple e-values is still an e-value is a very favorable property,
as it allows us to pool results from multiple runs or multiple analyses of an experiment.
Since Theorem~\ref{thm:equiv} finds an equivalent e-value based formulation
of the knockoffs method, we can therefore combine results from 
different knockoff realizations by averaging the corresponding
 e-values, to achieve a derandomized procedure.

\subsection{The procedure}\label{sec:procedure}

Suppose we have used our knowledge of $P_X$ to construct
a valid distribution for $\tilde{\bX} \given \bX$, i.e., so that~\eqref{eqn:pairwise_exch} is satisfied for each $j$.
We then
generate $M$ copies of the knockoff matrix, drawing $\tilde{\bX}^{(1)},
\tilde{\bX}^{(2)},\ldots,\tilde{\bX}^{(M)}$ from this distribution (drawn i.i.d.~conditional on the observed data $(\bX,\bY)$). 
Let $W^{(m)}$ denote the feature importance statistics
computed with the $m$-th knockoff matrix,
$
W^{{(m)}} = \calW\big([\bX, \tilde{\bX}^{(m)}],\bY\big)
$.
Choosing some $\alpha_{\kn}\in(0,1)$, we compute the threshold
\@
\label{eq:defn_stopping_time}
T^{(m)} = \inf\bigg\{t>0: \frac{1+\sum_{j\in [p]} \ind\{W_j^{(m)}\le -t\}}
{\sum_{k\in[p]} \ind\{W_k^{(m)} \ge t\}} \le \alpha_{\kn}\bigg\}
\@
for each $m$, so that $\calS^{(m)}_{\kn} = \{j : W_j^{(m)}\ge T^{(m)}\}$ is the selected set for the knockoff filter
when run with the $m$-th copy $\tilde{\bX}^{(m)}$ of the knockoff matrix. Let
\@\label{eq:e_val_m}
e_j^{(m)} = p \cdot\frac{ \ind\{W_j^{(m)} \ge T^{(m)}\}}
{1+\sum_{k \in [p]}\ind\{W_k^{(m)} \le -T^{(m)}\}}
\@
be the corresponding e-value, so that as proved in Theorem~\ref{thm:equiv}, the $m$-th selected set $\calS^{(m)}_{\kn}$ is equivalent to
running e-BH on the e-values $e_1^{(m)},\dots,e_p^{(m)}$.
For each $j\in[p]$, we aggregate the e-values
obtained from these $M$ knockoff copies by taking 
the average 
\$
e_j^{\textnormal{avg}} = \frac{1}{M} \sum^M_{m=1} e_j^{(m)}.
\$
Finally, we obtain the selected set of discoveries, $\calS_{\kn\textnormal{--derand}}$, by applying the 
e-BH procedure at level $\alpha_{\ebh}$ to  the e-values $e_1^{\textnormal{avg}},\dots,e_p^{\textnormal{avg}}$. 
Note that the parameters $\alpha_{\kn}$ and $\alpha_{\ebh}$ may be different---we will discuss
this more below.

The following theorem proves that this derandomized  
procedure
controls the FDR at level $\alpha_{\ebh}$.
\begin{theorem}\label{thm:fdr}
For any $\alpha_{\kn},\alpha_{\ebh} \in (0,1)$, and any number of knockoff copies $M\geq 1$,
the selected set $\calS_{\kn\textnormal{--derand}}$ computed in Algorithm~\ref{alg:aggregate_knockoff} satisfies $\fdr\leq\alpha_{\ebh}$.
\end{theorem}
\begin{proof}
Applying the bound~\eqref{eqn:kn_ratio} with the $m$-th knockoff copy $\tilde{\bX}^{(m)}$ in place of $\tilde{\bX}$, and with the threshold 
$\alpha_{\kn}$ in place of $\alpha$, we see that
$\EE\left[\frac{\sum_{j\in\calH_0} \ind\{W_j^{(m)}\ge T^{(m)}\}}{1 + \sum_{j\in\calH_0} \ind\{W_j^{(m)}\le -T^{(m)}\}}\right] \leq 1$.
This implies that the $m$-th set of e-values satisfies $\sum_{j\in\calH_0} \EE[e_j^{(m)}] \leq p$, for each $m$. Taking the average
over $m=1,\dots,M$, we have $\sum_{j\in\calH_0}\EE[e_j^{\textnormal{avg}}]\leq p$ as well.
The result then follows from Theorem~\ref{thm:relaxed_ebh}.
\end{proof}

\begin{remark}
A special case of Algorithm~\ref{alg:aggregate_knockoff} is when
$\alpha_{\kn}=\alpha_{\ebh} = \alpha$ and $M=1$. In this case, the derandomized procedure
reduces to the original knockoffs procedure at level $\alpha$, and 
$\alpha_{\ebh}$ is also the optimal choice for $\alpha_{\kn}$.

To see the optimality, let $T(\alpha)$ and $\calS(\alpha)$ denote 
the selection threshold and the selection set of knockoffs with 
target FDR level $\alpha$; let $e_j(\alpha)$ 
be the e-value by plugging $T(\alpha)$ in~\eqref{eq:defn_e_vals}.
When $\alpha_{\kn} < \alpha_{\ebh}$, we have by construction that
$|\{j:e_j(\alpha_{\kn})>0\}| = |\calS(\alpha_{\kn})|
\le |\calS(\alpha_{\ebh})|$. Since only positive e-values can 
possibly be selected, the number of selections 
made by e-BH applied to $\{e_j(\alpha_{kn})\}_{j\in[p]}$
is no larger than $|\calS(\alpha_{\ebh})|$. On the other hand,
when $\alpha_{\kn} > \alpha_{\ebh}$, we assume $|\calS(\alpha_{\kn})|
> |\calS(\alpha_{\ebh})|$, since otherwise $e_j(\alpha_{\kn})
\le e_j(\alpha_{\ebh})$ for any $j\in[p]$ and the result is immediate.
By the definition of $T(\alpha_{\kn})$, 
\$ 
\frac{1+\sum_{j\in[p]}\ind\{W_j \le -T(\alpha_{\kn})\}}
{\sum_{j\in[p]} \ind\{W_j \ge T(\alpha_{\kn})\}}
> \alpha_{\ebh}.
\$
As an implication, for any integer $K$ such that 
$|\calS(\alpha_{\ebh})| + 1 \le K \le |\calS(\alpha_{\kn})|$,
the $K$-th largest element in $\{e_j(\alpha_{\kn})\}_{j\in[p]}$ is
\$ 
\frac{p}{1+\sum_{j \in [p]} \ind 
\{W_j \le -T(\alpha_{\kn})\}} < \frac{p}{\alpha_{\ebh}
\cdot |\calS(\alpha_{\kn})|}
\le \frac{p}{\alpha_{\ebh}\cdot K}.
\$
By definition, the e-BH procedure cannot make more than 
$|\calS(\alpha_{\ebh})|$ selections. Collectively, we 
show that $\alpha_{\kn} = \alpha_{\ebh}$ is optimal.
\end{remark}

\begin{remark}
In fact, we are allowed to use different methods for constructing the knockoffs $\tilde{\bX}^{(m)}$
 and/or different functions $\calW$ for defining the feature importance statistics $W^{(m)}$, for each run $m=1,\dots,M$;
the FDR control result will still hold as long as the conditions for validity of the knockoff procedure are satisfied
for each $m$. This could be the case when two different labs are 
using different knockoffs generating mechanisms and/or different $\calW$ 
for their data analysis and wish to combine their results. More generally, 
combining different knockoff configurations across
multiple runs can potentially improve the robustness of our procedure, and 
this is an interesting direction for future research.
\end{remark}

\begin{remark}[Generalization of $T^{(m)}$]
\label{remark:generalized_stopping_time}
More generally, we are free to choose the $m$-th threshold $T^{(m)}$ in a different way
as long as it is still a stopping time with respect to the filtration generated
by a masked version of $W^{(m)}$, and the FDR guarantee for derandomized knockoffs will still hold.

To be specific, for each run $m$ of knockoffs,
the values $e_j^{(m)}$ defined in~\eqref{eq:e_val_m}
are valid e-values (according to the relaxed definition~\eqref{eqn:evals_relaxed}),
as long as $T^{(m)}$ is a stopping time with respect to the filtration generated
by a masked version of $W^{(m)}$, meaning that for each $t>0$,
the event $\ind\{T^{(m)}\ge t\}$
is determined by (1) the magnitudes $|W^{(m)}_j|$ for all $j$, (2) the values $W^{(m)}_j$
for all $j$ with $|W_j^{(m)}|< t$, and (3) $\sum_{j:|W_j^{(m)}|\ge t} \ind\{W_j^{(m)} >0\}$. 
For example, for any $c \ge 0$, one can define 
\$
T^{(m)} \defn
\inf\bigg\{t>0: \frac{c + \sum_{j\in[p]}\ind\{W_j^{(m)} \le -t\}}
{\sum_{k\in[p]} \ind\{W_k^{(m)} \ge t\}} \le \alpha_{\kn}\bigg\}.
\$
\end{remark}

\begin{remark}
In practice, we recommend letting $M$ be 
the largest acceptable number given the computational 
constraint of the user.
\end{remark}
\paragraph{A uniform improvement on $T^{(m)}$ in~\eqref{eq:defn_stopping_time}}
For any choice of $\alpha_{\kn}$, we can uniformly 
improve the power of our procedure by slightly 
modifying the stopping time defined in~\eqref{eq:defn_stopping_time}.
This improvement is inspired by~\citet{luo2022improving}, and 
specifically, we define an alternative stopping time as
\@\label{eq:early_stopping_time}
T^{(m)} \defn  \inf\bigg\{t>0: 
\frac{1+\sum_{j\in[p]}\ind\{W_j^{(m)} \le -t\}}
{\sum_{k \in [p]}\ind\{W_k^{(m)} \ge t\}} \le \alpha_{\kn}
\textnormal{ or } \sum_{j\in[p]} \ind\{W_j \ge t\}<1/\alpha_{\kn}  \Big\}.
\@
Note that \eqref{eq:early_stopping_time} differs 
from~\eqref{eq:defn_stopping_time} 
only when 
\$\inf\bigg\{t>0: 
\frac{1+\sum_{j\in[p]}\ind\{W_j \le -t\}}
{\sum_{k \in [p]}\ind\{W_k \ge t\}} \le \alpha_{\kn}\bigg\}
> \inf\Big\{t>0: \sum_{j\in[p]} \ind\{W_j^{(m)} \ge t\}<1/\alpha_{\kn}  \Big\},
\$
that is, the ``hopeless'' case where~\eqref{eq:defn_stopping_time}
gives an infinite stopping time and all the corresponding 
e-values are $0$'s. The alternative definition
of $T^{(m)}$ only \emph{increases} the value of e-values by stopping earlier in the
powerless case, therefore leading to a uniformly more powerful procedure.
Plugging~\eqref{eq:early_stopping_time} in~\eqref{eq:defn_e_vals}
still yields valid e-values as a result of Remark~\ref{remark:generalized_stopping_time}. The improvement is particularly significant
when $\alpha_{\kn}$ and/or the fraction of non-nulls is small.
The complete procedure with the improved stopping time 
is described in Algorithm~\ref{alg:aggregate_knockoff}.

\subsubsection{Comparison to~\citet{ren2021derandomizing}}
Recall that the procedure of~\citet{ren2021derandomizing} 
computes the (unweighted) selection probability defined 
in~\eqref{eqn:old_prob}, then picking the features whose 
selection probability passes a constant threshold.
In contrast, our procedure instead computes a weighted 
selection probability, and then 
thresholding it with a data-dependent cutoff---here, 
the weight corresponding to the $m$-th run is given by 
\$ 
\frac{1/(1+\sum_{k \in [p]} \ind\{W_k^{(m)} \le -T^{(m)}\})}
{\sum_{l \in [M]} 1/(1+\sum_{k \in [p]}\ind\{W_k^{(l)} \le -T^{(l)}\} )}.
\$

\subsubsection{Comparison to~\citet{dai2022false,dai2023scale}}
Our procedure uses e-values of the form $e_j = \frac{1}{M} \sum_{m=1}^M 
\textnormal{weight}_j^{(m)} \cdot \ind\{j \in \calS_{\kn}^{(m)}\}$,
where we have now seen that the weights are chosen as
\[ \textnormal{weight}_j^{(m)}= p \cdot\frac{1}
{1+\sum_{k \in [p]}\ind\{W_k^{(m)} \le -T^{(m)}\}}.\]
In contrast, the method proposed in~\citet{dai2022false,dai2023scale}
can be rewritten in our notation by taking
\begin{equation}\label{eqn:weights_Dai} \textnormal{weight}_j^{(m)}= p \cdot\frac{1}
{\alpha_{\kn}\cdot \sum_{k\in[p]} \ind\{W_k^{(m)} \ge t\}}.\end{equation}
By definition of the knockoff procedure, comparing the denominators, the weights used by~\citet{dai2022false,dai2023scale} 
are always less than or equal to the weights in our proposed method; consequently,
our method is strictly more powerful. In addition, 
the results of~\citet{dai2022false,dai2023scale} establish asymptotic control of the FDR
and require additional strong conditions (e.g., an assumption that the ranking of the baseline 
algorithm is consistent asymptotically).
In contrast, our result in Theorem~\ref{thm:fdr} offers finite-sample FDR control without any additional
assumptions.
(Indeed, since the key
to establish finite-sample FDR control is showing that the weighted selection 
frequencies we construct in~\eqref{eqn:informal_intro} is a class of (relaxed) e-values,
and since this is also true when the weights $\textnormal{weight}_j^{(m)}$ are chosen
according to~\citet{dai2022false,dai2023scale}'s more conservative rule~\eqref{eqn:weights_Dai},
our analysis can also be applied to show a finite-sample FDR control result for their method.)

\begin{algorithm}[t]
\caption{Derandomized knockoffs with e-values}\label{alg:aggregate_knockoff}
\begin{algorithmic}[1]
\REQUIRE Data $(\bX,\bY)$; parameters $\alpha_{\kn},\alpha_{\ebh} \in (0,1)$, 
$M \in \mathbb{N}_+$.\;\\ 
\vspace{0.05in}
\FOR{$m = 1,\ldots,M$}
\STATE Sample the knockoff copy $\tilde{\bX}^{(m)}$.\;
\STATE Compute the feature importance statistics: $W^{(m)} = \calW\big([\bX,\tilde{\bX}^{(m)}],\bY)$.\;
\STATE Compute the knockoffs threshold $T^{(m)}$ according to~\eqref{eq:early_stopping_time}.\; 
\STATE Compute the e-values $e_j^{(m)}$ according to~\eqref{eq:e_val_m} for each $j\in[p]$.\; 
\ENDFOR
\STATE Compute the averaged e-values $e_j^{\textnormal{avg}} = \frac{1}{M}\sum^M_{m=1} e_j^{(m)}$ for each $j\in[p]$.\;
\STATE Compute $\hat{k} = \max\{k: e^{\textnormal{avg}}_{(k)} \ge p/(\alpha_{\ebh} k)\}$, or $\hat{k} = 0$ if this set is empty.\;
\vspace{0.05in}
\ENSURE The selected set of discoveries $\calS_{\kn\textnormal{--derand}} \defn \{j \in[p]: e^{\textnormal{avg}}_j \ge p/(\alpha_{\ebh}\hat{k})\}$.
\end{algorithmic}
\end{algorithm}

\subsection{Choices of parameters}\label{sec:choices_of_parameters}
In Algorithm~\ref{alg:aggregate_knockoff}, we have the freedom 
to choose two potentially different parameters $\alpha_{\ebh}$ and $\alpha_{\kn}$.
While $\alpha_{\ebh}$ determines the ultimate FDR guarantee of the procedure (as established in Theorem~\ref{thm:fdr}),
the parameter $\alpha_{\kn}$ does not affect this FDR guarantee, but may instead affect the power of the method. 
Should we simply choose $\alpha_{\kn} = \alpha_{\ebh}$ (particularly given that
this choice reduces to the original knockoffs method when $M=1$)? 

In fact, it turns out that choosing  $\alpha_{\kn} < \alpha_{\ebh}$ is preferable when $M>1$. To understand why, we will sketch a simple scenario 
where choosing  $\alpha_{\kn} = \alpha_{\ebh}$ might lead to zero power, while choosing 
a smaller value of $\alpha_{\kn}$
can lead to high power.
Suppose that the data contains $s$ many non-nulls, which each have extremely strong signals, so that 
any single run of the knockoff filter is highly likely to select all $s$ non-nulls. 
For each run $m$, we will expect a false discovery proportion of approximately $\alpha_{\kn}$, meaning that we have
$|\calS_{\kn}^{(m)}\cap\calH_0| \approx \frac{\alpha_{\kn}}{1 - \alpha_{\kn}} s$ false discoveries together
with the $\approx s$ true discoveries. In particular, we should then have
\[1 + \sum_{j\in\calH_0}\ind\big\{W_j^{(m)} \le -T^{(m)}\big\} \approx \sum_{j\in\calH_0}\ind\big\{W_j^{(m)} \ge T^{(m)}\big\}  = 
\big|\calS_{\kn}^{(m)}\cap\calH_0\big| \approx \frac{\alpha_{\kn}}{1 - \alpha_{\kn}} s,\]
where the first step holds since null features $j$ have statistics $W_j^{(m)}$ that are symmetric around zero. 
Since each non-null $j$ is selected in (nearly) every run of knockoffs, 
while for a null $j\in\calH_0$
it may be the case that $j$ is selected for only a small proportion of the runs,
we therefore expect to see
\[e_j^{\textnormal{avg}}\approx \frac{p}{\frac{\alpha_{\kn}}{1 - \alpha_{\kn}} s}\textnormal{ for }  j\in\calH_1, \quad e_j^{\textnormal{avg}}\approx 0 \textnormal{ for } j\in\calH_0.\]
Applying the e-BH procedure at level $\alpha_{\ebh}$ to these e-values, we see that power can be high only
 if $p /(\frac{\alpha_{\kn}}{1 - \alpha_{\kn}} s)  \gtrapprox p / (\alpha_{\ebh} s)$, while otherwise power will be zero as we will not
 be able to make any discoveries. In other words, $\alpha_{\kn}=\alpha_{\ebh}$ will likely lead to zero power but 
 $\alpha_{\kn} \lessapprox \frac{\alpha_{\ebh}}{1+\alpha_{\ebh}}$ will lead to (nearly) perfect recovery.

The above derivation provides some intuition 
for choosing $\alpha_{\kn} < \alpha_{\ebh}$ in order to obtain high power.
As pointed out by Remark~\ref{remark:generalized_stopping_time},
we also have the freedom to tune the offset parameter $c$ in defining
the knockoff threshold. To provide a guidance on parameter selection, 
we numerically investigate the power of our proposed procedure with different choices of 
$\alpha_{\kn}$ and $c$ in the simulated experiments (see Section~\ref{sec:simulation} and 
Section~\ref{sec:add_simulations} of the appendix). Empirically, we find that the combination of 
$c = 1$ and $\alpha_{\kn}=\alpha_{\ebh}/2$ consistently works well in different 
scenarios;
whether these parameters are optimal, or can be chosen in a better way,
remains an open question for future work.

\subsection{Derandomization: a closer look}
In Algorithm~\ref{alg:aggregate_knockoff}, conditional on the data $(\bX,\bY)$, for each feature $j$ the $e_j^{(m)}$'s are i.i.d.~for $m=1,\dots,M$.
By the strong law of large numbers, we then have
\begin{equation}\label{eqn:LLN}e_j^{\textnormal{avg}}=
\frac{1}{M}\sum_{m=1}^M e_j^{(m)} {\longrightarrow}
~e_j^{\infty} \defn 
\EE\big[e_j^{(1)} \biggiven \bX,\bY\big],
\end{equation}
almost surely as $M\rightarrow\infty$. 
Let $\calS_{\kn\textnormal{--}\infty}$ denote the selected set obtained by applying e-BH
to these limiting e-values $e_1^{\infty},\dots,e_p^\infty$. Note 
that $\calS_{\kn\textnormal{--}\infty}$ is now a deterministic function of the data $(\bX,\bY)$.
 By~\eqref{eqn:LLN}, we would then expect that the selected set $\calS_{\kn\textnormal{--derand}}$
obtained by averaging over $M$ runs of knockoffs, should eventually stabilize to this completely derandomized 
selected set $\calS_{\kn\textnormal{--}\infty}$. 

To quantify this intuition, let 
$
\Delta = \min_{k\in[p]} |e_{(k)}^{\infty} - p/(\alpha k)|
$.
Applying Hoeffding's inequality, we have with probability at least
$1 - 2p \cdot \exp(-2\Delta^2 M/p^2)$ that 
$\max_{j\in[p]}|e_j^{\textnormal{avg}} - e_j^{\infty}|< \Delta$, which implies $\max_{k\in[p]}|e_{(k)}^{\textnormal{avg}} - e_{(k)}^{\infty}|< \Delta$.
On this event, then, the definition of the e-BH procedure implies that
 $\calS_{\kn\textnormal{--derand}} =\calS_{\kn\textnormal{--}\infty}$.

\section{Simulations}
\label{sec:simulation}
We next illustrate the performance of our proposed derandomized knockoffs method,
and compare to the original knockoffs procedure,
with numerical experiments.\footnote{Code to reproduce all experiments on simulated and real data is available at 
\url{https://github.com/zhimeir/derandomized_knockoffs_fdr}. 
The knockoffs method is implemented via the R-package \texttt{knockoff} \citep{patterson2018knockoff}.}

\paragraph{Settings} We generate data for two different models: the Gaussian linear model
and the logistic model. 
In both settings, we choose sample size 
$n=1000$. In the linear case, the feature dimension
$p = 800$ with $|\calH_1| = 80$ non-nulls, and 
in the logistic case, $p = 600$ with $|\calH_1| = 50$ non-nulls.
The marginal distribution $P_X$ of the feature vector $X\in\RR^p$
is given by $X\sim\calN(0,\Sigma)$, where the covariance matrix has entries $\Sigma_{jk}=0.5^{|j-k|}$.
The distribution of $Y \given X$ is defined as follows for our two models:
\[\textnormal{Gaussian linear model: }Y\given X \sim\calN(X^\top\beta , 1), \quad
\textnormal{Logistic model: }Y\given X \sim\textnormal{Bernoulli}\bigg(\frac{e^{X^\top\beta}}{1+e^{X^\top\beta}}\bigg),\]
where $\beta \in \RR^p$ is the coefficient vector. 
To determine $\beta$, we first construct $\bar{\beta} \in \RR^{|\calH_1|}$,
where each entry of $\bar{\beta}$ is independently generated from 
$\calN(A,1)$, with $A$ denoting the overall signal amplitude. Once generated, 
$\bar{\beta}$ is fixed across the experiments. We then determine the 
coefficient vector $\beta$ by letting 
\$
& \beta = 
\Big(\underbrace{0,\ldots,0}_{z},\frac{\bar{\beta}_1}{\sqrt{n}},
\underbrace{0,\ldots,0}_{z },-\frac{\bar{\beta}_2}{\sqrt{n}},
\ldots\Big),
\$
where $z = 9$ in the linear case and $z=11$ in the logistic case.
The signal amplitude $A\in\{4, 5,6,7,8\}$ for the Gaussian linear model or 
$A\in\{10, 15, 20, 25,30,35\}$ for the logistic model. 

\paragraph{Methods} 
For the original model-X knockoff filter, the knockoffs $\tilde{\bX}$ are constructed using
the multivariate Gaussian distribution $P_X$ of the features, as in \citet[Eqn.~(3.2),~(3.13)]{candes2018panning},
and the feature importance statistics $W\in\RR^p$ are computed via the Lasso coefficient-difference (LCD)
statistic as in \citet[Eqn.~(3.7)]{candes2018panning}. Finally, the selection
set $\calS_{\kn}$ is computed via the knockoff filter run with
 target FDR level $\alpha = 0.1$.

For the derandomized knockoffs method (Algorithm~\ref{alg:aggregate_knockoff}),
we draw $M=50$ copies of the knockoffs, and compute the $m$-th selected set $\calS_{\kn}^{(m)}$
using the selection threshold defined in~\eqref{eq:early_stopping_time}
with $\alpha_{\kn} \in \{0.01,0.02,\ldots,0.2\}$, 
and then apply e-BH with target FDR level $\alpha_{\ebh} = 0.1$.
When compared with the original knockoffs, $\alpha_{\kn}$ is 
taken to be $0.05$.

\paragraph{Measures of performance}
For each model and each signal amplitude setting, the data
set $(\bX,\bY)$ is generated independently 100 times. For each draw of the dataset $(\bX,\bY)$,
we also replicate the knockoffs procedure (i.e., original or derandomized) 20 times.
(When investigating the effect of different parameters, the derandomized knockoffs procedure is implemented once for each draw 
of the dataset.)
 For each model, each signal amplitude, and each of the two methods,
we measure performance in several different ways.
First, we estimate the power and FDR
as
\[\textnormal{Power} = \frac{1}{100\cdot 20}\sum_{d = 1}^{100} \sum_{k = 1}^{20} \frac{| \calS_{d,k} \cap \calH_1|}{|\calH_1|}\textnormal{\ \ and \ }
\fdr = \frac{1}{100\cdot 20}\sum_{d = 1}^{100} \sum_{k = 1}^{20} \frac{| \calS_{d,k} \cap\calH_0|}{|\calS_{d,k} | \vee 1},\]
where $\calS_{d,k}$ denotes the selected set for the $d$-th draw of the dataset $(\bX,\bY)$
and the $k$-th run of (original or derandomized) knockoffs for this dataset. 

Next, we would like to quantify the ``randomness'' of each selection
procedure.
Note 
that there are two sources of variability:
the randomness conditional on the dataset (i.e., arising from the knockoffs construction), and 
that from the dataset itself. 
We define two different measures, to capture the overall (``marginal'') variability,
and to capture the ``conditional'' variability that arises from the knockoffs construction
(note that it is this latter measure that derandomized knockoffs aims to reduce).
Define
\[\hat{p}_j = \frac{1}{100\cdot 20}\sum_{d=1}^{100}\sum_{k=1}^{20} \ind\{j \in \calS_{d,k}\}\textnormal{ \ \ and \ }\hat{p}_{j,d} = \frac{1}{20}\sum_{k=1}^{20} \ind\{j \in \calS_{d,k}\}\]
for each $j\in[p]$, where $\hat{p}_j$ is the empirical probability of selection on average over a random draw of the datasets and the knockoffs,
while $\hat{p}_{j,d}$ is
the empirical probability of selection conditional
on the $d$-th draw of the dataset. Let
$\hat{s}$ be the average number of selections over all $d,k$, and let $\hat{s}_d$ be the average number over all runs $k$ of knockoffs for a fixed dataset $d$.
We then define
\[\textnormal{Marginal selection variability} = \frac{\sum_{j\in[p]} \hat{p}_j(1-\hat{p}_j)}{ p \cdot \frac{\hat{s}}{p}(1-\frac{\hat{s}}{p})}.\]
To motivate this definition, if the selected set $\calS_{d,k}$ contains the same $\hat{s}$ features for each $d,k$ 
then we will have $\hat{p}_j \in\{0,1\}$ for all $j$ and so $\text{Var}_{\textnormal{marg}} =0$,
while if $\calS_{d,k}$ contains a random selection of $\hat{s}$ features for each $d,k$, then we will have $\hat{p}_j \approx \hat{s}/p$
for all $j$ and so $\text{Var}_{\textnormal{marg}} \approx1$.
Similarly, we define
\[\textnormal{Conditional selection variability} = \frac{\sum_{d=1}^{100}\sum_{j\in[p]} \hat{p}_{j,d}(1-\hat{p}_{j,d})}{ \sum_{d=1}^{100}p \cdot \frac{\hat{s}_d}{p}(1-\frac{\hat{s}_d}{p})}.\]

\begin{figure}[h]
\centering 
\rotatebox{90}{Gaussian}
\begin{minipage}{0.45\textwidth}
\centering
\includegraphics[width = 0.8\textwidth]{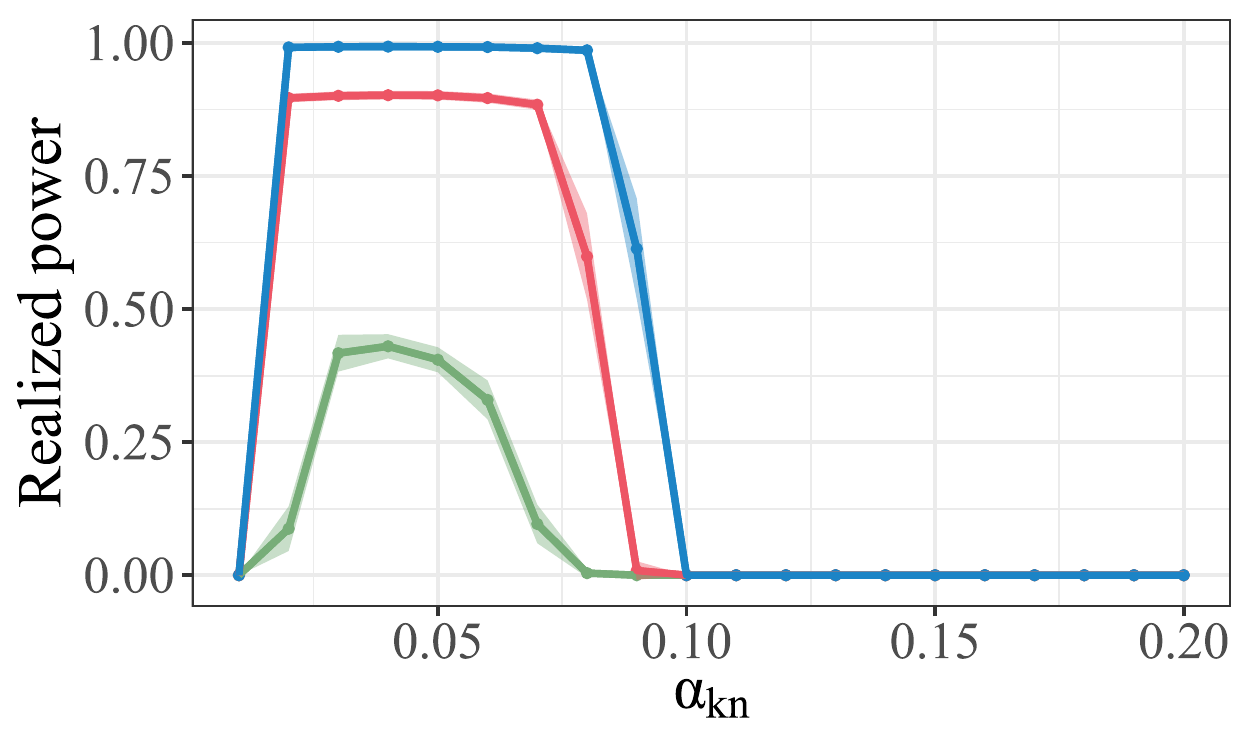}
\end{minipage}
\begin{minipage}{0.45\textwidth}
\centering
\includegraphics[width = 0.8\textwidth]{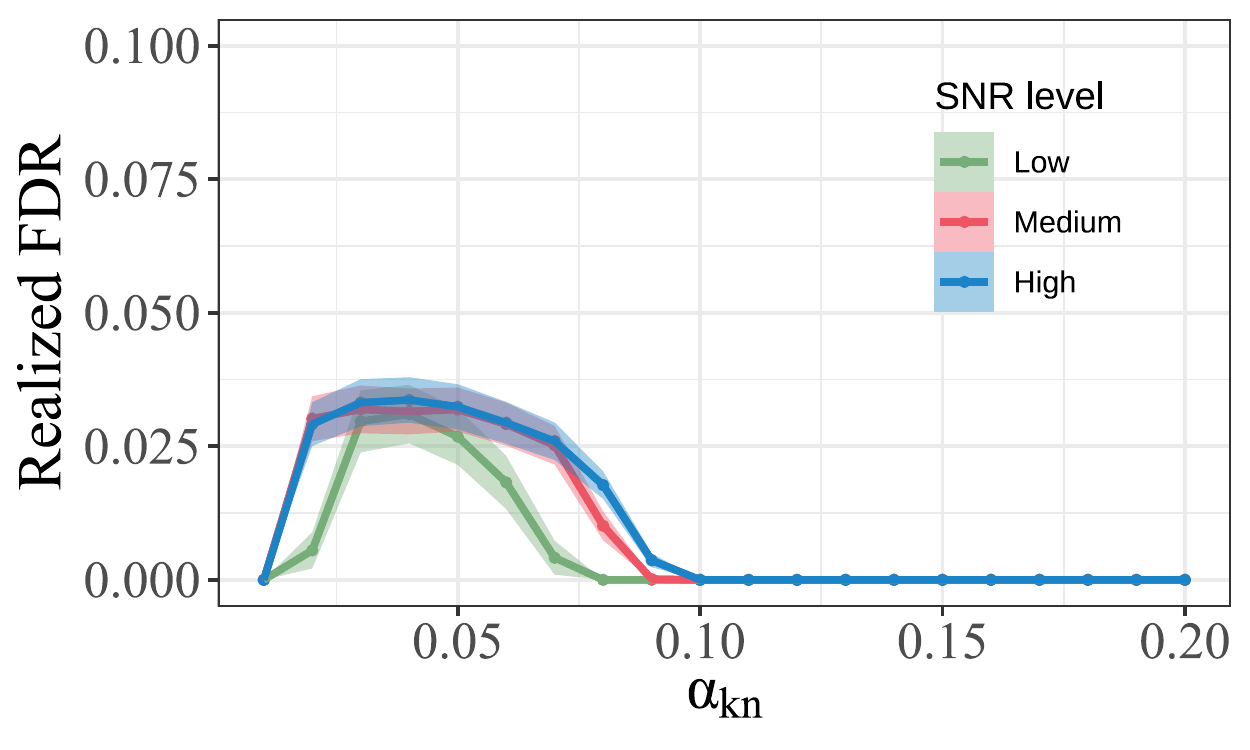}
\end{minipage}\\
\rotatebox{90}{Logistic}
\begin{minipage}{0.45\textwidth}
\centering
\includegraphics[width = 0.8\textwidth]{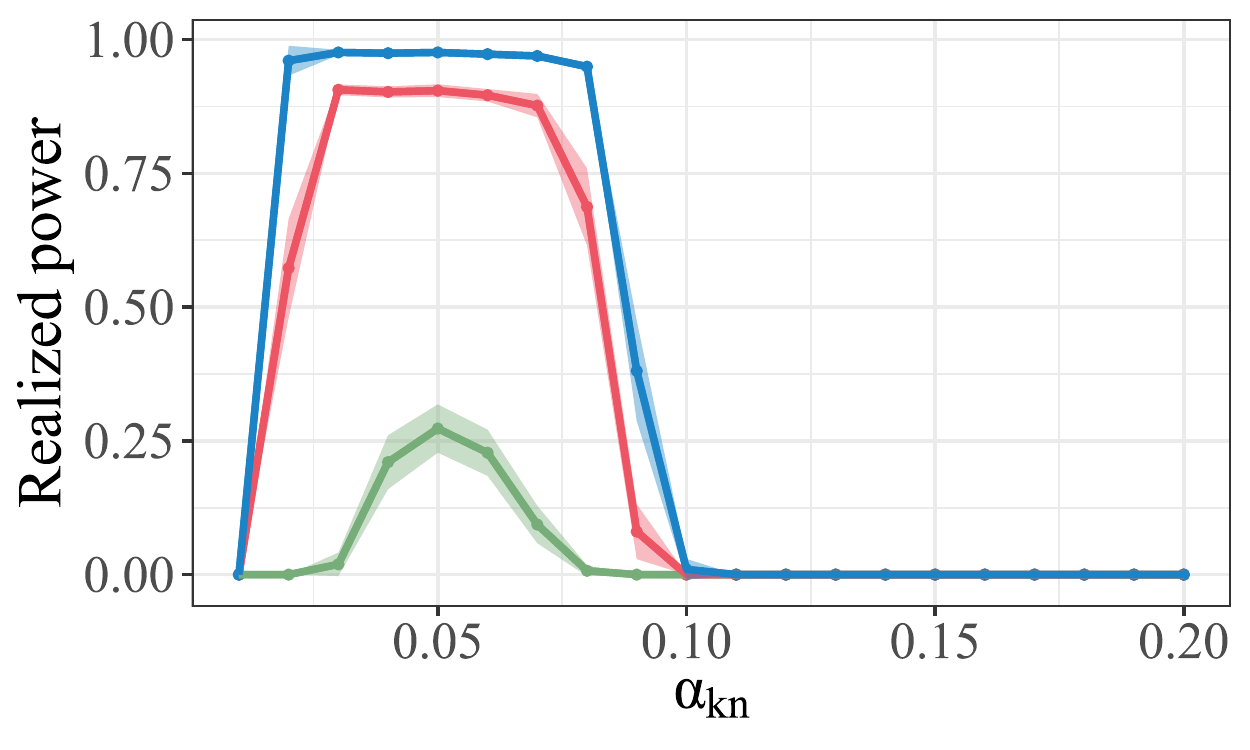}
\end{minipage}
\begin{minipage}{0.45\textwidth}
\centering
\includegraphics[width = 0.8\textwidth]{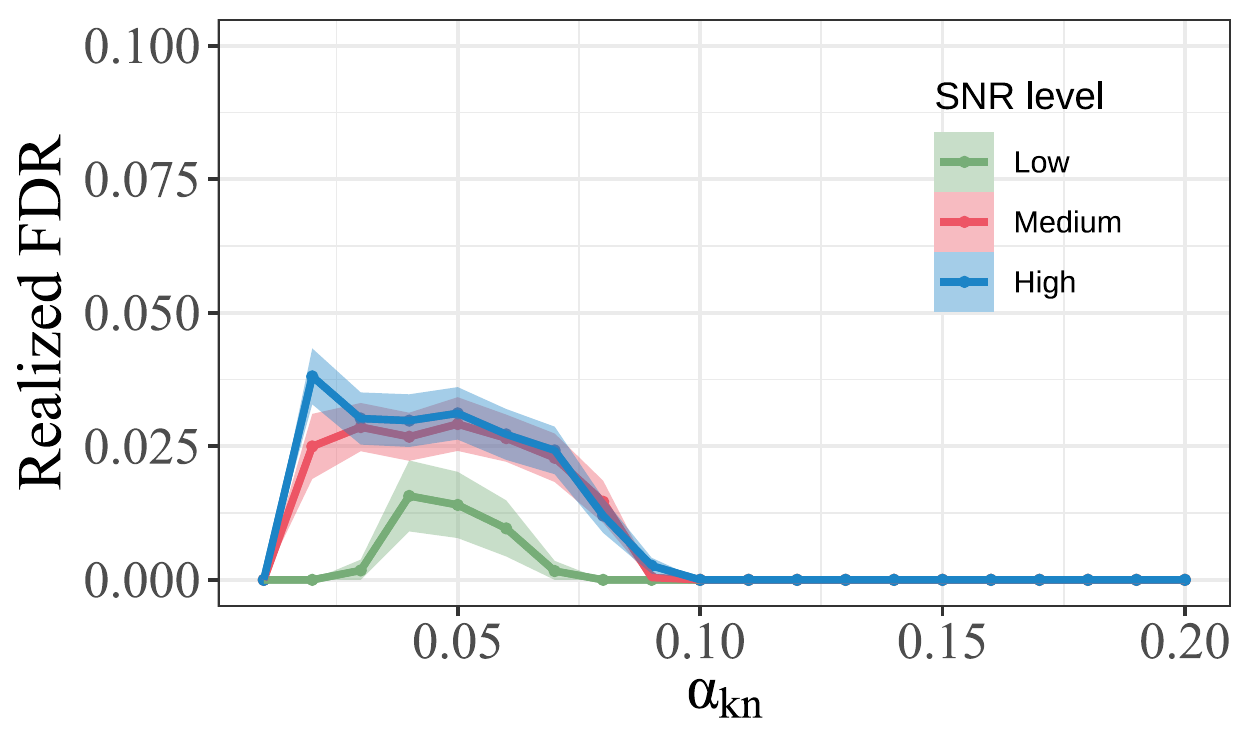}
\end{minipage}
\caption{Realized power (left) and FDR (right) of derandomized 
knockoffs as a function of the parameter $\alpha_{\kn}$ for the simulation 
data experiments. The offset parameter $c=1$. 
Shading for the power and FDR plots indicates error bars.
The target FDR level $\alpha_{\ebh} = 0.1$.
Results are averaged over $100$ independent trials.}
\label{fig:alpha}
\end{figure}

\begin{figure}[h]
\centering
\rotatebox{90}{Gaussian}
~
\begin{minipage}{0.3\textwidth}
\centering
\includegraphics[width = \textwidth]{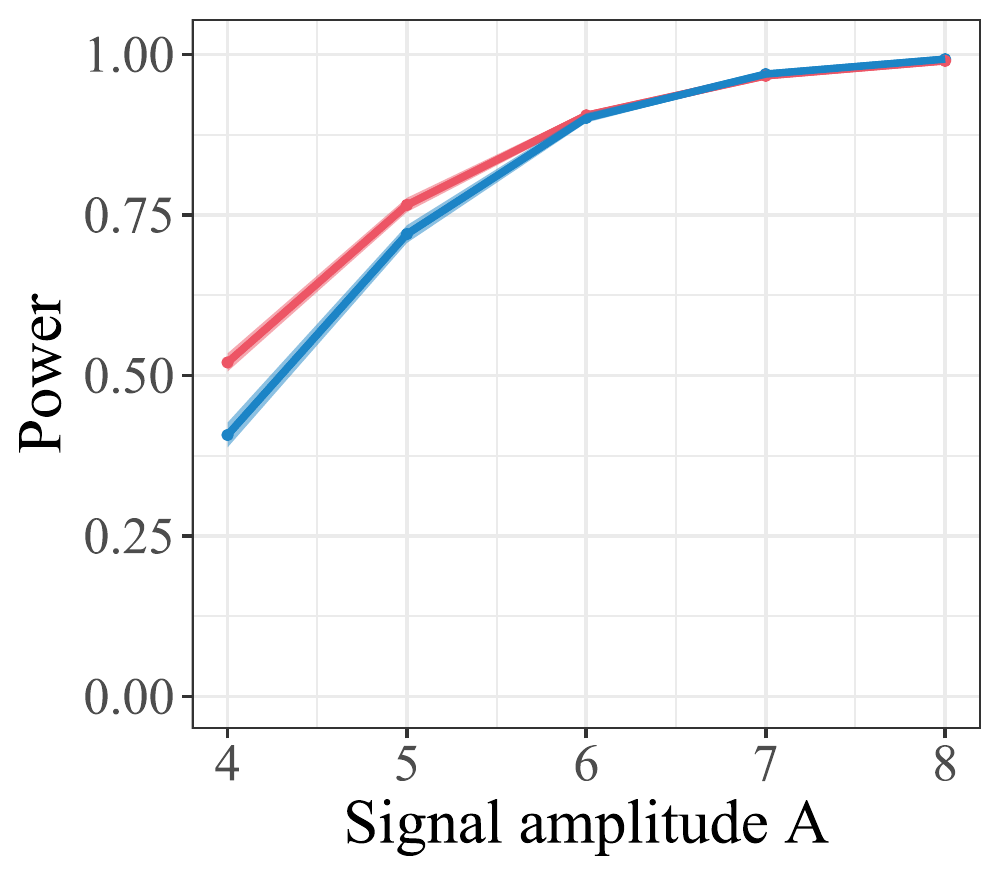}
\end{minipage}
\begin{minipage}{0.3\textwidth}
\centering
\includegraphics[width = \textwidth]{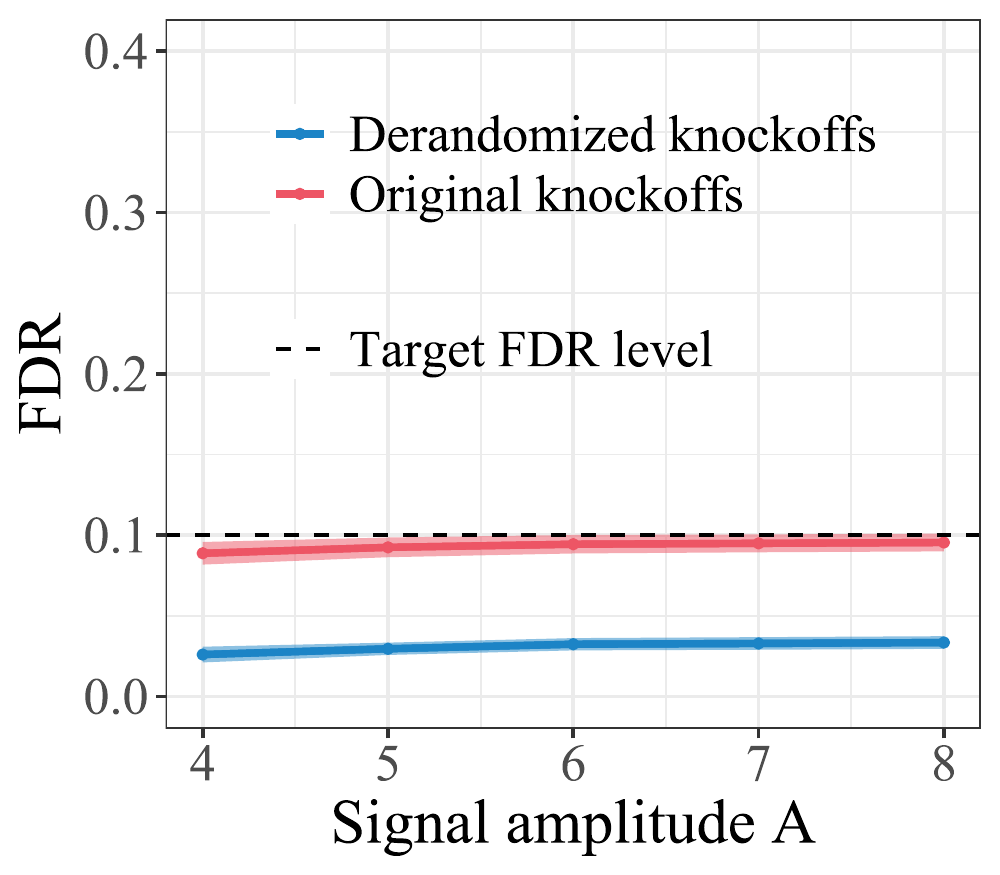}
\end{minipage}
\begin{minipage}{0.3\textwidth}
\centering
\includegraphics[width = \textwidth]{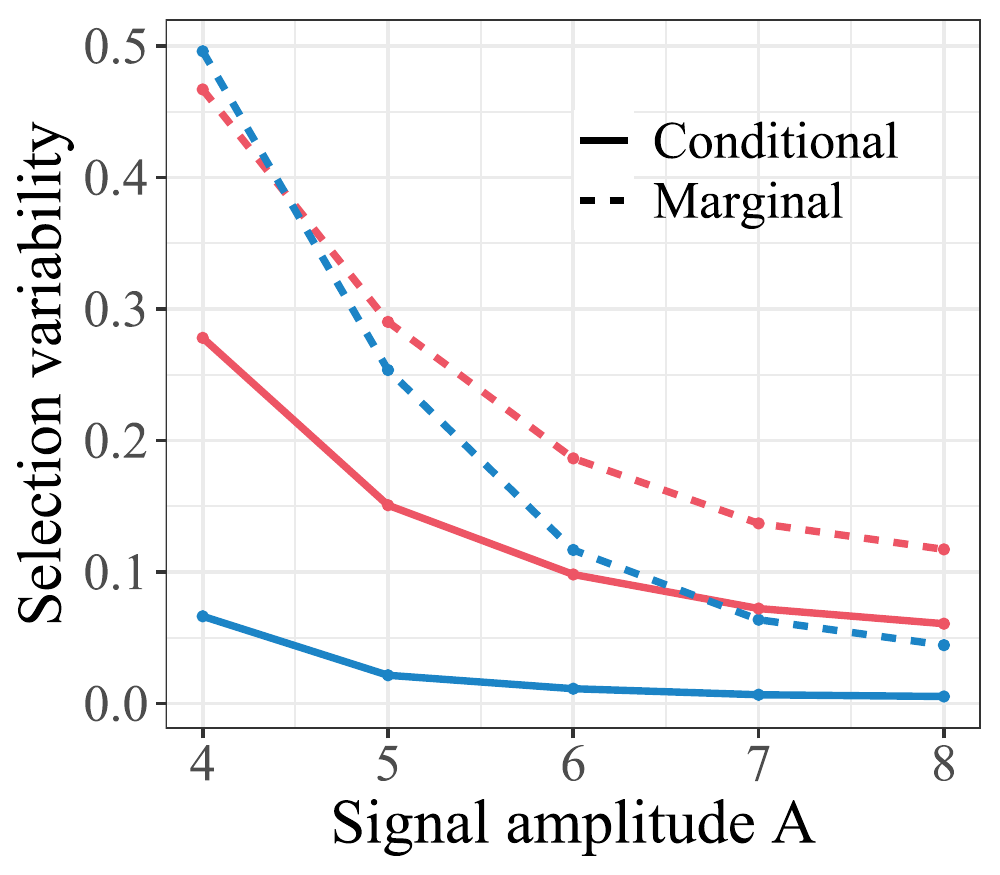}
\end{minipage}\\
\centering
\rotatebox{90}{Logistic}
~
\begin{minipage}{0.3\textwidth}
\centering
\includegraphics[width = \textwidth]{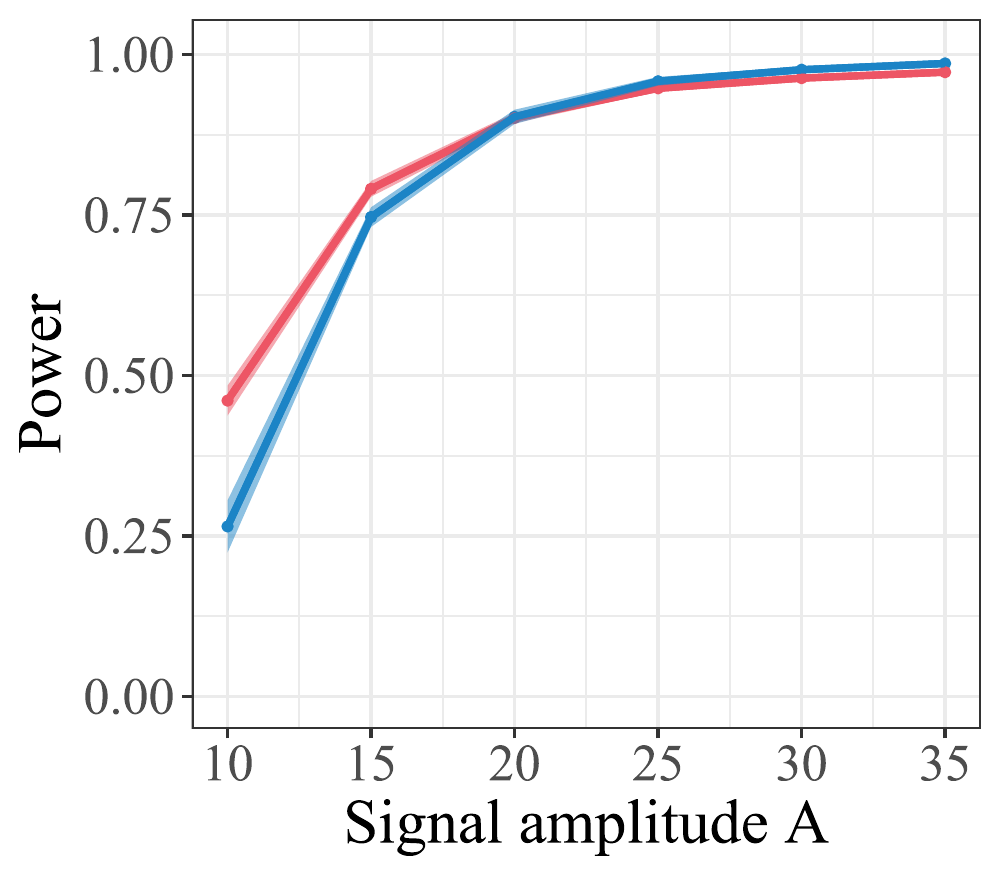}
\end{minipage}
\begin{minipage}{0.3\textwidth}
\centering
\includegraphics[width = \textwidth]{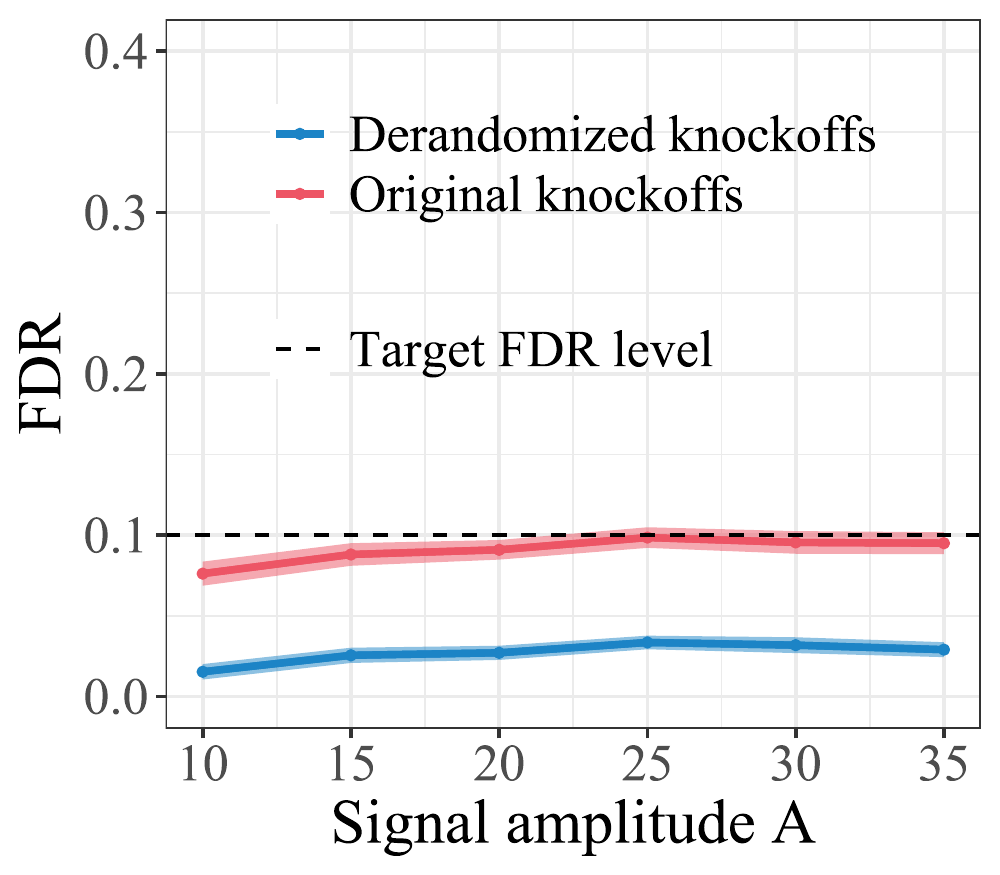}
\end{minipage}
\begin{minipage}{0.3\textwidth}
\centering
\includegraphics[width = \textwidth]{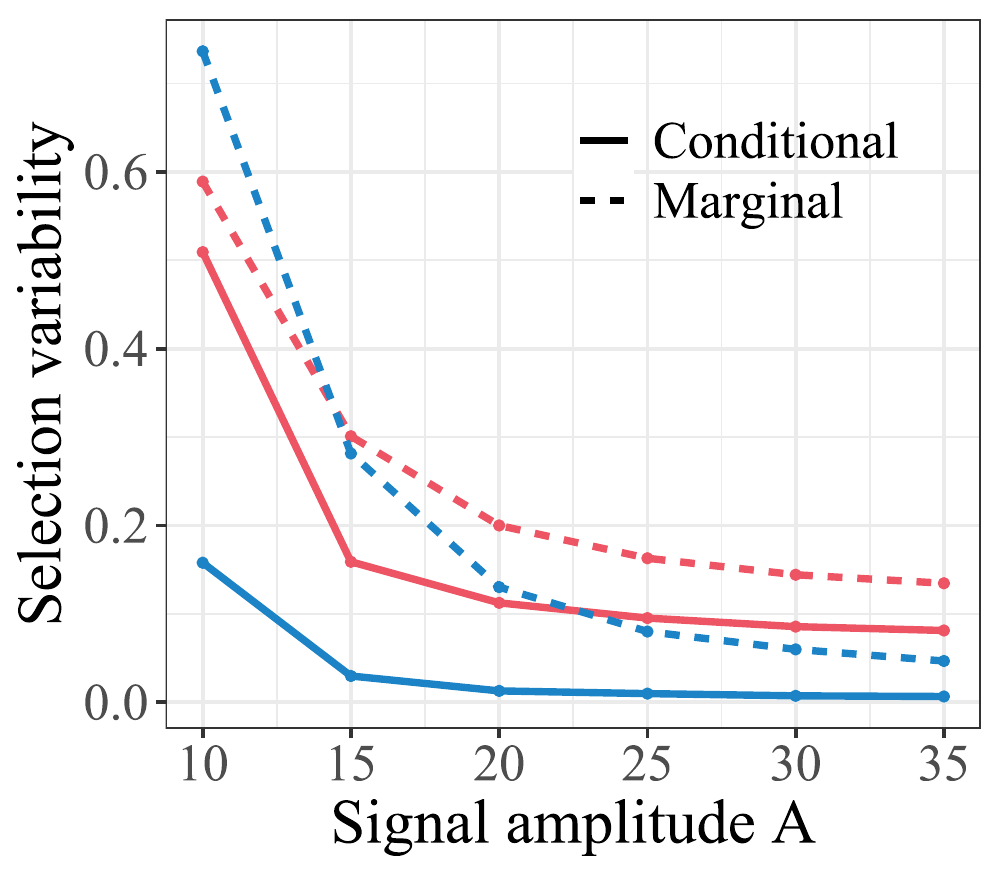}
\end{minipage}
\caption{Power, FDR,
and selection variability, for the simulated data experiments.
Shading for the power and FDR plots indicates error bars.
Results are averaged over $100$ independent trials.}
\label{fig:simulation}
\end{figure}

\begin{figure}[h]
\centering 
\includegraphics[width = \textwidth]{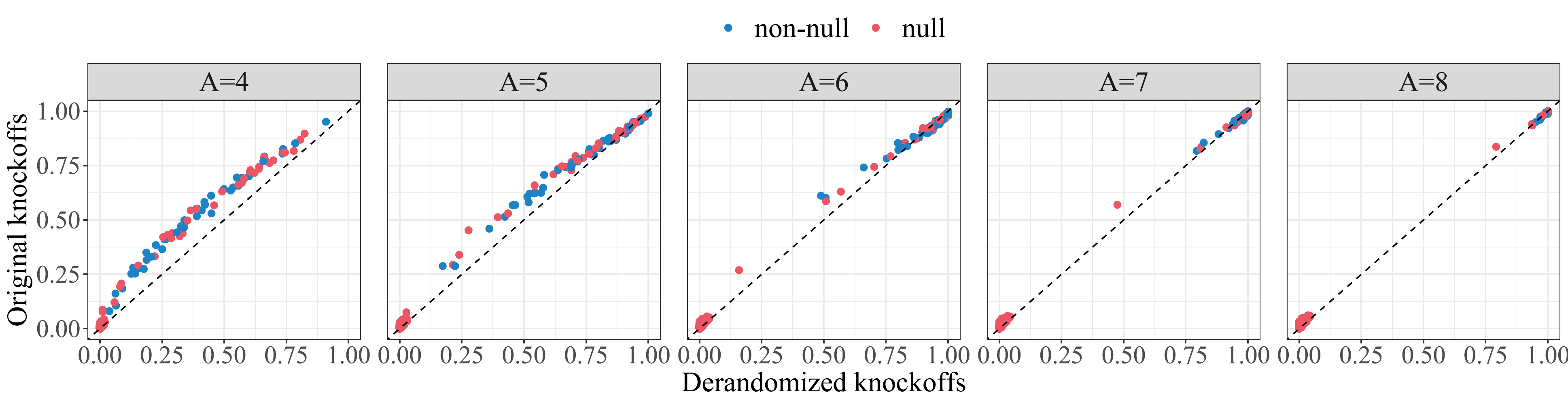}
\includegraphics[width = \textwidth]{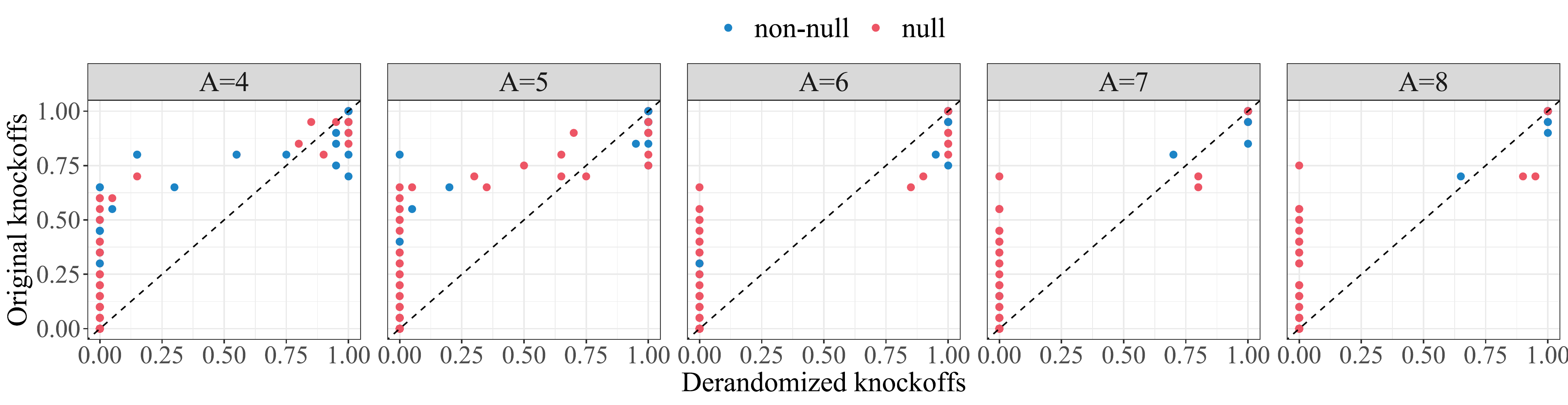}
\caption{Top: the marginal selection probability $\hat{p}_j$
by original knockoffs versus that by derandomized knockoffs.
Bottom: the conditional selection probability $\hat{p}_{j,1}$
by original knockoffs versus that by derandomized knockoffs. 
The results are from simulations under the Gaussian linear model.
Each point corresponds to a feature, where the blue ones are 
non-nulls and the red ones are nulls. For the given dataset, 
many null features never selected by derandomized knockoffs have
large selection probability by the original knockoffs.}
\label{fig:sel_linear}
\end{figure}

\begin{figure}[h]
\centering 
\includegraphics[width = \textwidth]{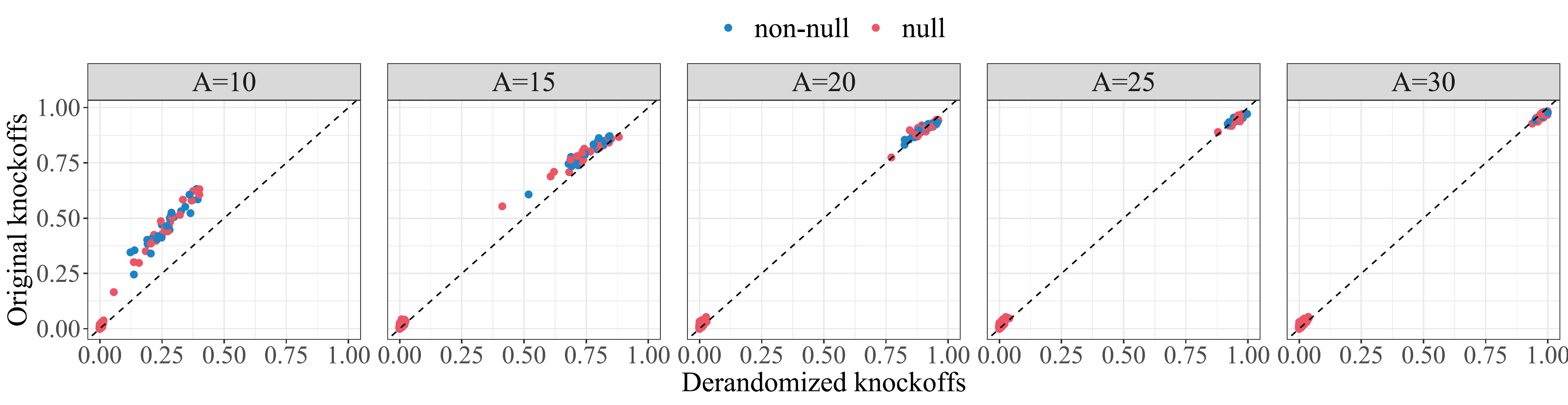}
\includegraphics[width = \textwidth]{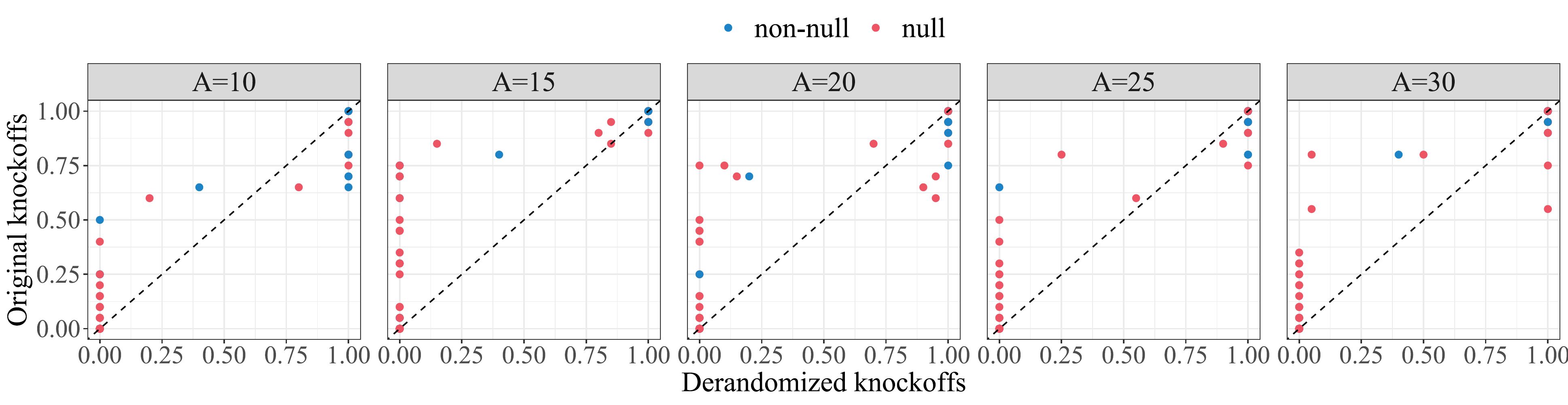}
\caption{ 
The results are from simulations under the logistic model.
The other details are the same as in Figure~\ref{fig:sel_logistic}.}
\label{fig:sel_logistic}
\end{figure}

\paragraph{Results} 
Figure~\ref{fig:alpha} plots the power and FDR of 
derandomized knockoffs as a function of $\alpha_{\kn}$ 
under both models with different signal amplitude. 
For the Gaussian linear model, the low, medium and high
signal amplitude correspond to $A = 4$, $6$ and $8$, 
respectively; for the logistic model,
the three cases correspond to $A=10$, $20$ and $30$, 
respectively. In all the settings, we observe that 
derandomized knockoffs consistently performs well with
$\alpha_{\kn} = \alpha_{\ebh}/2 = 0.05$.

Figure~\ref{fig:simulation} plots the power, FDR, and selection variability
as functions of the signal amplitude, for both settings.
We see that derandomized knockoffs
achieves comparable power as the original knockoffs method when the signal 
strength is reasonably large, while there is some power loss 
when the signal strength is weak (as we shall see in 
Figure~\ref{fig:sel_linear} and~\ref{fig:sel_logistic}, 
the original knockoffs achieves higher power in this scenario with an increased 
chance of selecting null features). For logistic regression, 
derandomized knockoffs achieves a slightly higher
power than the original knockoffs method in the 
strong-signal setting. In all cases, 
both methods have FDR bounded by the target level $\alpha=0.1$,
whereas derandomized knockoffs consistently achieves an FDR substantially lower than 
the original knockoffs. Finally, derandomized knockoffs in general exhibits
lower selection variability, both marginally and conditionally.
In particular, the conditional variability (that is, the variability arising from the randomness in the knockoffs
method, while conditioning on a fixed dataset) is 
substantially decreased for the derandomized method, as expected.

In Figure~\ref{fig:sel_linear} and~\ref{fig:sel_logistic},
we take a closer look at the marginal and conditional selection 
probability of individual features by the two methods. Especially in weak-signal settings, 
we can see that the original knockoffs often has a high selection 
probability for the nulls that are rarely selected by derandomized knockoffs 
(i.e., many red points are above the 45-degree line); this phenomenon 
is more pronounced when we focus on the conditional selection probability.

\section{Application to HIV mutation data}
\label{sec:application}
We next compare original and derandomized knockoffs on a study of detecting mutations
associated with drug resistance in Human Immunodeficiency
Virus Type 1 (HIV-1)~\citep{rhee2006genotypic}. 

\paragraph{Data} The 
dataset\footnote{Available at 
\url{https://hivdb.stanford.edu/download/GenoPhenodatasets/PI_dataset.txt}}
 documents  drug resistance measurements (for several
classes of drugs) and genotype information on samples of HIV-1.
Our analysis focuses on a particular  protease inhibitor (PI)
named lopinavir. For sample $i$ and mutation $j$, 
$Y_i$ denotes the log-fold increase in resistance to the drug,
and $X_{ij}$ is a binary variable indicating whether or not 
mutation $j$ is present in sample $i$.

\paragraph{Methods}
After cleaning the data (removing rows with missing values 
and mutations that appear less than three times in the data), we have
a dataset of $n = 1840$ samples and $p=219$ mutations.
The data matrix $\bX$ is normalized such that
each column has zero mean and unit variance.
We use the second-order method~\citep[Sec.~3.4.2]{candes2018panning}
to generate knockoffs via the first and second moments of $\bX$
(as discussed in~\citet{romano2020deep} and the references therein,
it is empirically observed that this approximate knockoffs construction often affects the
power but leaves the FDR intact). The original and derandomized knockoffs methods
are implemented with the LCD as the feature importance statistic, as for our simulations. 
We set $\alpha = 0.1$ for knockoffs, and $\alpha_{\kn}=0.05$ and $\alpha_{\ebh} =0.1$ for derandomized
knockoffs, so that the target FDR level is $0.1$ for both methods.
The derandomized method is run with $M=100$ copies of the knockoff matrix.
For both original and derandomized knockoffs, we repeat the experiment $50$ times
in order to examine the variability of the selected set across different generations of the knockoffs.

\begin{figure}[t]
\centering\hspace{-.2in}
\begin{minipage}{0.4\textwidth}
\centering
\includegraphics[width = \textwidth]{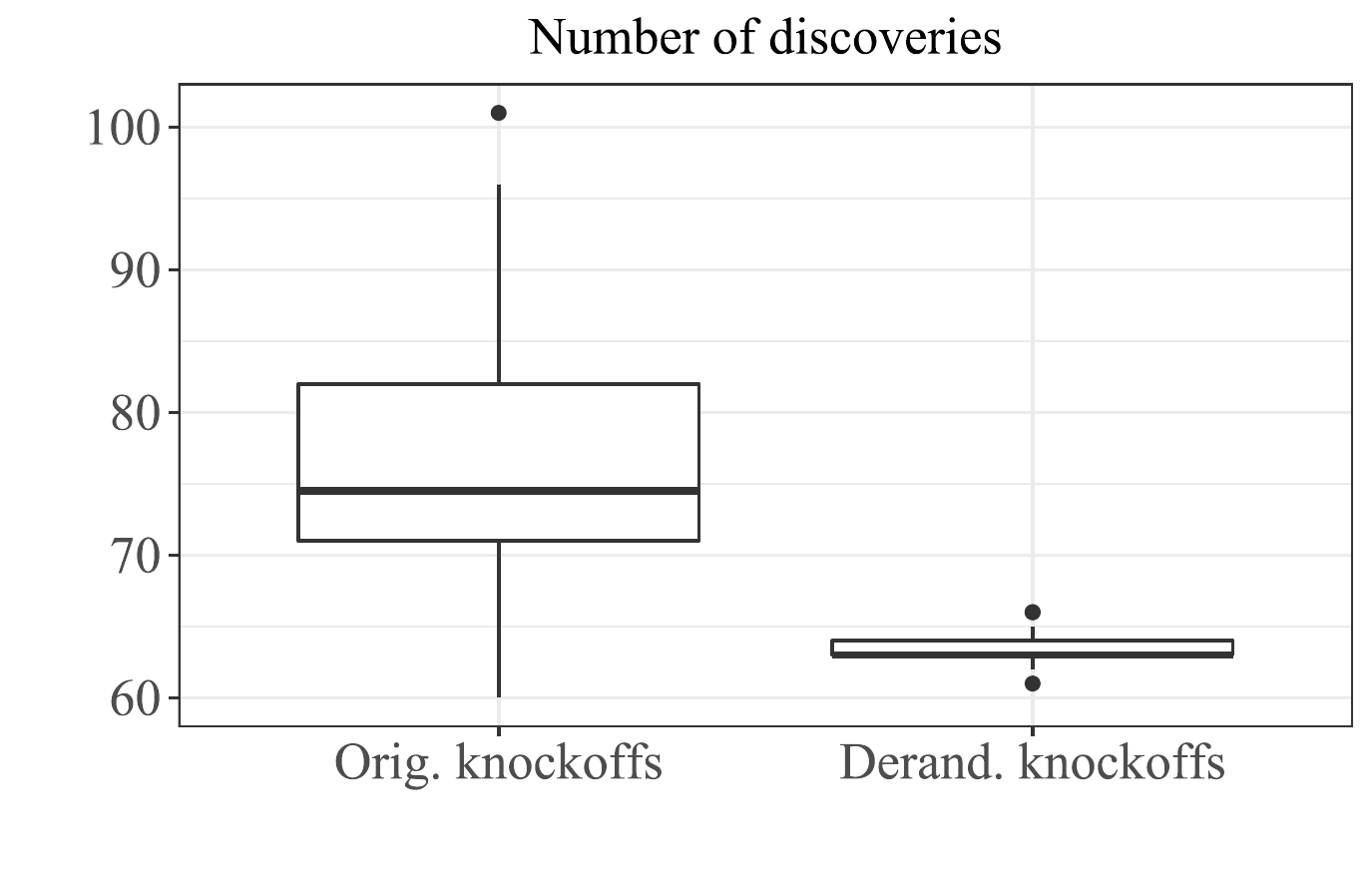}\\
\end{minipage}\hspace{.2in}
\begin{minipage}{0.58\textwidth}
\centering
\includegraphics[width = \textwidth]{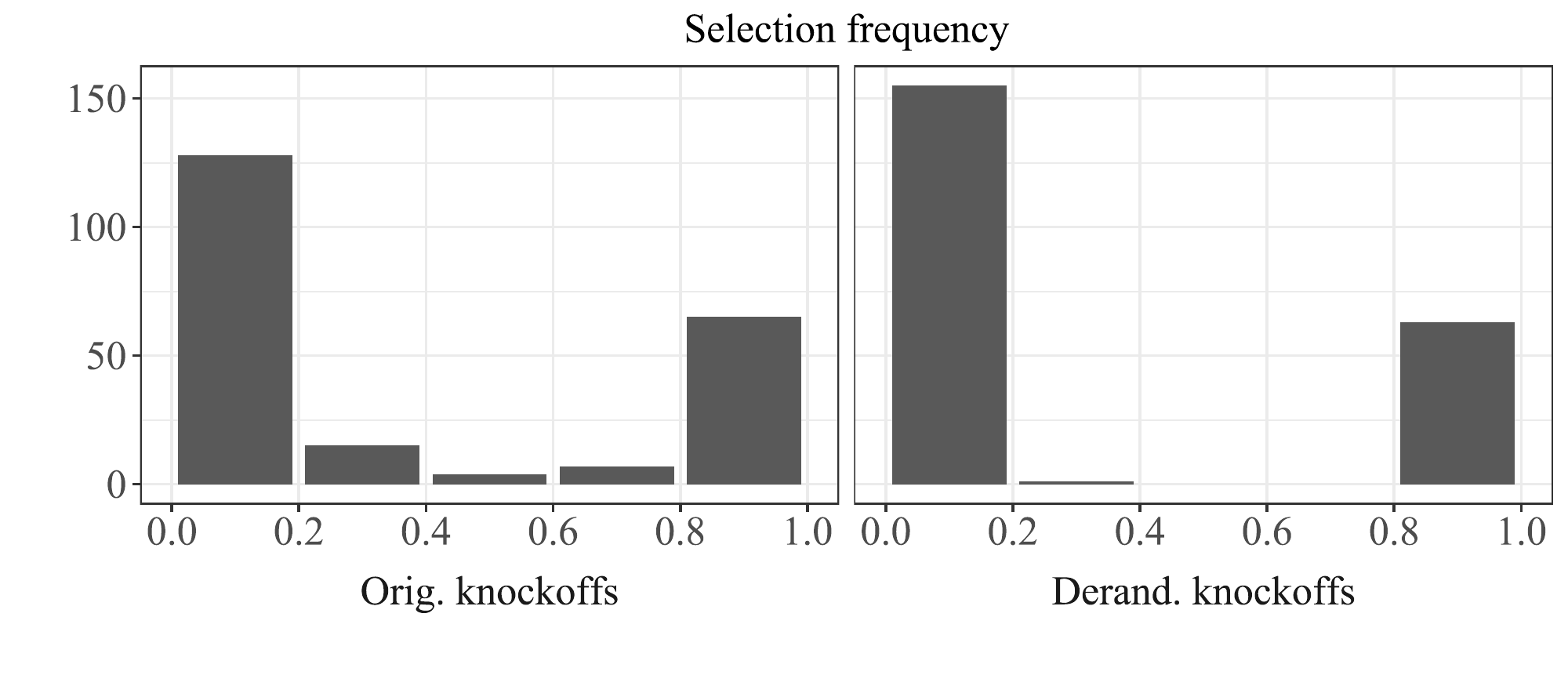}
\end{minipage}
\caption{Results for the HIV data experiment. Left: boxplot of the number of discoveries for original and derandomized knockoffs,
over 50 independent trials. Right: histogram of the selection probability for each feature $j\in[p]$, over 50 independent
trials, for original and derandomized knockoffs.}
\label{fig:realdata}
\end{figure}

\paragraph{Results} 
The results of this experiment are shown in Figure~\ref{fig:realdata}, summarizing the results
from 50 independent runs of both methods.
In the histograms, we display the selection frequency (across the 50 trials) for each mutation $j\in[p]$, for both methods.
For the original knockoffs method, we see that a substantial number of mutations have selection probability
bounded away from 0 and from 1, exhibiting some instability. For derandomized knockoffs,
on the other hand, nearly all the mass is at 0 and 1, meaning that each mutation is either (almost) always
selected or (almost) never selected.
In the boxplot, we see that the original knockoffs method makes more discoveries on average
than derandomized knockoffs, but the number of discoveries is quite variable.
When the discoveries from each method are cross-referenced
with the verified discoveries reported at~\url{https://hivdb.stanford.edu/dr-summary/comments/PI/},
we find evidence of a higher FDR for the original knockoff method, indicating that derandomized knockoffs
may retain the same power despite the lower average number of discoveries (see Appendix~\ref{sec:discovery_list}
for a more detailed analysis, including which specific mutations were selected).

\section{Discussion}
\label{sec:discussion}
In this work, we discover a connection between
knockoffs and e-values, showing that the knockoff filter can be reframed as an e-BH
procedure. This equivalence in turn allows us to pool information across multiple
runs of knockoffs to achieve a derandomized procedure without compromising the guarantee
of FDR control.
Our method shows desirable performance in practice, retaining the power of the original method
while reducing variability. Additionally, derandomized knockoffs
achieves an empirical FDR that is lower than the target level in our experiments;
as an open question for future work, it would be interesting
to see under what mild conditions can we derive a sharper
FDR bound (for example, as in \citet{meinshausen2010stability,shah2013variable}
for stability selection, or as in \citet{ren2021derandomizing} for the previously proposed version of
derandomized knockoffs).

Many variants and extensions of the knockoffs method have appeared in the literature,
including the ``fixed-X'' setting \citep{barber2015controlling}, robustness results
for the model-X setting when $P_X$ is not known exactly \citep{barber2020robust}, 
a multienvironment knockoff filter \citep{li2021searching}, and knockoffs that incorporate side information \citep{ren2020knockoffs}.
In Appendix~\ref{sec:extensions_appendix}, we give derandomized versions of each of these,
together with accompanying theoretical guarantees. 
In addition to these extensions, our proposed derandomization scheme 
could potentially also be combined with other existing extensions of the knockoffs 
methodology---for example, the {\em group knockoff filter}~\citep{dai2016knockoff} for discovering 
a group-sparse structure, or the {\em missing value knockoffs method}~\citep{koyuncu2022missing} 
to handle the presence of missing data.

\subsection*{Data availability statement}
The code for reproducing the numerical results 
in this paper can be found at \url{https://github.com/zhimeir/derandomized_knockoffs_fdr}.

\subsection*{Acknowledgment}
Z.R.~and R.F.B~were supported by the Office of Naval Research via grant N00014-20-1-2337.
R.F.B.~was additionally supported by the National Science Foundation 
via grants DMS-1654076 and DMS-2023109.

\appendix

\numberwithin{equation}{section}
\numberwithin{theorem}{section}
\numberwithin{lemma}{section}
\numberwithin{definition}{section}
\numberwithin{figure}{section}
\numberwithin{table}{section}
\numberwithin{algorithm}{section}

\section{Extensions}\label{sec:extensions_appendix}

\subsection{Fixed-X knockoffs}
The derandomization procedure can also
be applied to aggregating multiple runs of
the fixed-X knockoffs procedure~\citep{barber2015controlling}. 
The fixed-X knockoffs procedure considers
a setting where the design matrix $\bX$ is fixed
and the response is generated from a Gaussian linear model:
\$
\bY = \bX\beta + \varepsilon,
\$
where $\beta \in \RR^p$ is an unknown coefficient
vector and $\varepsilon \sim \calN(0,I_n)$.
The test of conditional independence (defined in~\eqref{eq:ci_test}) 
is equivalent to testing if $\beta_j = 0$.
The (fixed-X) knockoffs $\tilde{X}$ for $\bX$ is any 
matrix $\tilde{\bX} \in \RR^{n\times p}$
satisfying  
\$
\tilde{\bX}^\top \tilde{\bX} = \mathbf{\Sigma},
\qquad
\bX^\top \tilde{\bX} = \mathbf{\Sigma} - \textnormal{diag}\{s\},
\$
where $\mathbf{\Sigma} = \bX^\top \bX$ is the gram matrix,
and $s$ is a $p$-dimensional nonnegative vector; given $\tilde{\bX}$,
one can compute the feature importance statistics
with $([\bX,\tilde{\bX}],\bY)$ and apply the knockoffs
filter as in the model-X case.
Specifically, $\tilde{\bX}$ can be generated as
\[\tilde{\bX} = \bX(I_p - \mathbf{\Sigma}^{-1}\textnormal{diag}\{s\}) - \mathbf{U} (2\textnormal{diag}\{s\} - \textnormal{diag}\{s\}\Sigma^{-1}\textnormal{diag}\{s\})^{1/2},\]
where $\mathbf{U} \in\RR^{n\times p}$ is an orthonormal matrix orthogonal to $\bX$ (i.e., $\mathbf{U}^\top\mathbf{U} = I_p$ and $\mathbf{U}^\top\bX=0$).
When generating $\tilde{\bX}$, the choice of $\mathbf{U}$ is arbitrary, and may be constructed
deterministically or randomly (c.f.~\citet[Section 2]{barber2015controlling})---in either case, this arbitrary choice will generally affect the set of discoveries,
and thus it might be desirable to derandomize in order to obtain a more stable and meaningful selected set.
Our derandomization 
scheme can be applied here to aggregating results from fixed-X knockoffs
with different methods of implementation, and the resulting set of 
discoveries has guaranteed FDR control (the proof is exactly the 
same as in the model-X case).
The complete procedure for derandomized fixed-X
knockoffs is described in Algorithm~\ref{alg:aggregate_fix_x},
and its validity proved in the following theorem.
\begin{theorem}
\label{thm:fix_kn}
For any $\alpha_{\kn} ,\alpha_{\ebh} \in (0, 1)$, and any number of knockoff copies $M \ge1$, the selected set
 offered by 
Algorithm~\ref{alg:aggregate_fix_x}
satisfies $\fdr\leq \alpha_{\ebh}$.
\end{theorem}   
\begin{proof}
By~\citet[Lemma A1]{barber2015controlling}, 
$\sum_{j\in\calH_0} \EE[e_j^{(m)}]\le p$, and
by the linearity of expectation, 
$\sum_{j\in\calH_0}\EE[e^{\textnormal{avg}}_j]\le p$.
The result follows immediately from 
Theorem~\ref{thm:relaxed_ebh}.
\end{proof}

\begin{algorithm}[htbp]
\caption{Derandomized fixed-X knockoffs}\label{alg:aggregate_fix_x}
\begin{algorithmic}[1]
  \REQUIRE  Data $(\bX,\bY)$; nonnegative $s$ satisfying $2\textnormal{diag}\{s\} - \textnormal{diag}\{s\}\Sigma^{-1}\textnormal{diag}\{s\}\succeq 0$;
  parameters $\alpha_{\ebh},\alpha_{\kn} \in (0,1)$ and 
  $M \in \mathbb{N}_+$. \;\\ 
\vspace{0.05in}
\FOR{$m = 1,\ldots,M$}
\STATE Construct the knockoff copy \$\tilde{\bX}^{(m)} = \bX(I_p - \mathbf{\Sigma}^{-1}\textnormal{diag}\{s\}) - \mathbf{U}^{(m)} (2\textnormal{diag}\{s\} - \textnormal{diag}\{s\}\Sigma^{-1}\textnormal{diag}\{s\})^{1/2},\$ where $\mathbf{U}^{(m)}$ is drawn uniformly at random from the set of orthonormal $n\times p$ matrices
orthogonal to $\bX$. \;
\STATE Compute the feature importance statistics: 
$W^{(m)} = \calW^{(m)}\big([\bX,\tilde{\bX}^{(m)}],\bY)$.\;
\STATE Compute the stopping time $T^{(m)}$
according to~\eqref{eq:early_stopping_time}.
\STATE Compute 
the e-value $e_j$ according
to~\eqref{eq:e_val_m}, for all $j\in[p]$.\\
\ENDFOR
\STATE Compute the average e-value $e_j^{\textnormal{avg}} 
= \frac{1}{M}\sum^M_{m=1} e_j^{(m)}$ for each $j\in[p]$.\;
\STATE Compute $\hat{k} = \max\big\{k: e^{\textnormal{avg}}_{(k)} \ge {p}/(\alpha_{\ebh} k)\big\}$,
or $\hat k = 0$ if this set is empty.\;
\vspace{0.05in}
\ENSURE The selected set of discoveries 
$\calS_{\kn\textnormal{--derand}} 
\defn \big\{j \in [p]: e_j^{\textnormal{avg}} \ge p / (\alpha_{\ebh} \hat{k})\big\}$.
\end{algorithmic}
\end{algorithm}

\subsection{Robust knockoffs}
The validity of the model-X knockoffs frame relies 
crucially on the knowledge of $P_X$, but in practice, this distribution might be estimated
rather than known exactly. 
In this section, we investigate the robustness of the 
derandomized knockoffs procedure when the knockoff
copies are sampled without exact knowledge of $P_X$. 
In this section, we will show that when the knockoffs are constructed
using approximate knowledge of $P_X$, Algorithm~\ref{alg:aggregate_knockoff}
achieves approximate FDR control.

The setting of our discussion follows that 
of~\citet{barber2020robust}. Formally, suppose for 
each $j \in [p]$, the researcher only has access to
an estimated conditional distribution for $X_j \given X_{-j}$, 
denoted by $Q_j(\cdot \given X_{-j})$ (instead of 
the ground truth $P_j(\cdot \given X_{-j}) = P_{X_j\given X_{-j}}$). 
Based on these $Q_j$'s, the researcher samples the knockoffs
with $\tilde{X} \given X \sim P_{\tX \given X}$, such that
for any $j\in[p]$, $P_{\tX \given X}$ is pairwise exchangeable 
w.r.t.~$Q_j$. The definition of pairwise exchangeability 
w.r.t.~$Q_j$ is as follows.
\begin{definition}\label{defn:pairwise_exchangeability}
$P_{\tX \given X}$ is pairwise exchangeable with respect
$Q_j$ if for any distribution $D^{(j)}$ on $\RR^p$
with the $j$-th conditional distribution $Q_j$, and 
$(X,\tX) \sim D^{(j)} \times P_{\tX \given X}$,
it holds that
\$
\big(X_j, \tX_j ,X_{-j}, \tX_{-j}\big) 
\eqd
\big(\tX_j, X_j, X_{-j}, \tX_{-j}\big).
\$
\end{definition}
As remarked in~\citet{barber2020robust}, when the conditional
distributions $Q_j$'s are mutually compatible, i.e.~there exists
a joint distribution $Q_X$ on $\RR^p$ whose $j$-th conditional 
distribution coincides with $Q_j$, then the knockoff construction
algorithm~\citep[Algorithm 1]{candes2018panning} 
with $Q_X$ as the input produces knockoffs satisfying
Definition~\ref{defn:pairwise_exchangeability}. When the 
$Q_j$'s are not mutually compatible,~\citet[Section 4]{barber2020robust}
provide examples in which knockoffs satisfying Definition~\ref{defn:pairwise_exchangeability}
can be produced. Moving on, we work under the premise that we 
can generate knockoffs obeying Definition~\ref{defn:pairwise_exchangeability}.

To measure the deviation of $Q_j$ from the ground truth,
we follow~\citet{barber2020robust} and
define the ``empirical KL divergence'' quantity
for any $j\in[p]$ and for each run $m$ of knockoffs:
\$
 \hat{\kl}^{(m)}_j = \sum_i \log 
\bigg(\frac{P_j(\bX_{ij} \given \bX_{i,-j})}
{Q_j(\bX_{ij} \given \bX_{i,-j})}
\cdot \frac{Q_j(\tilde{\bX}_{ij}^{(m)} \given \bX_{i,-j})}
{P_j(\tilde{\bX}_{ij}^{(m)} \given \bX_{i,-j})}
\bigg),~ m\in[M].\$
We then take the maximum over the $M$ runs of knockoffs,
\$ \hat{\kl}^{\max}_j = \max_{m \in [M]} \hat{\kl}^{(m)}_j.
\$
Clearly, the smaller $\hat{\kl}^{\max}_j$ is, the closer
$Q_j$ is to $P_j$, and $\hat{\kl}^{\max}_j = 0$ 
when $P_j = Q_j$. Intuitively, the better
$Q_j$ approximated $P_j$, the more likely
it is to control
the FDR (when $P_j = Q_j$ for all $j\in[p]$,
the FDR is controlled at the desired level).
Such an intuition is formalized
by the following theorem which shows that Algorithm~\ref{alg:aggregate_knockoff}
controls the FDR among the features with small values of $\hat{\kl}^{\max}_j$.

\begin{theorem}\label{thm:robustness}
Fix any $\alpha_{\kn} ,\alpha_{\ebh} \in (0, 1)$, and any number of knockoff copies $M \ge1$.
For any $\varepsilon\ge 0$, 
consider the null variables for which 
$\hat{\kl}^{\max}_j \le \varepsilon$. The fraction of the 
rejections made by Algorithm~\ref{alg:aggregate_knockoff} 
that correspond to such nulls obeys
\$
\EE\bigg[\frac{\big|\big\{j: j \in \calS \cap \calH_0, \hat{\kl}^{\max}_j \le \eps \big\}\big|}
{|\calS|\vee 1}\bigg] \le e^{\eps} \cdot \alpha.
\$
As a result, the FDR can be controlled as
\$
\fdr \le \min_{\eps >0} \bigg\{e^{\eps} \cdot \alpha 
+ \PP\Big(\max_j \hat{\kl}^{\max}_j \ge \eps\Big)\bigg\}.
\$
\end{theorem}
The proof of Theorem~\ref{thm:robustness} can be
found in Appendix~\ref{appx:proof_robustness}. 
Note that the FDR inflation depends directly on the
$\hat{\kl}^{\max}_j$'s, which take the form of the maximum
of $M$ conditional independent objects. Consequently,
when $\hat{\kl}_j^{(m)}$ is bounded almost surely, 
the FDR inflation is bounded for any choices of $M$. 
In the numerical experiments to be presented in the next section,
we will see that we generally observe controlled FDR even with estimated 
$P_X$.

\subsection{Multienvironment knockoffs}
\label{sec:multi}
For the purpose of making scientific conclusions, it
is sometimes not sufficient to find associations that
are significant in 
{\em only} one environment---here, environments can
refer to subpopulations, experimental settings or
data sources (see more discussion in~\citet{li2021searching}).
Intuitively, if an association is found to be significant in
all environments (e.g., a genetic
variant found to be associated with a certain disease
in all subpopulations), it is more likely to be the 
``true'' driving factor and less prone to the 
impact of unobserved confounders.

Formally, suppose there are $E$ environments, 
and in each environment $e\in[E]$, the data 
vector $(X,Y)\sim P_{XY}^e$ (we work again in the model-X framework, and so $P_X^e$ is assumed known).
We write the conditional independence hypothesis in environment $e$ as
\@\label{eq:e_null}
H_j^{\ci,e}: Y^e~\indep~X_j^e \given X^e_{-j}.
\@
To find the association that is 
{\em consistently} true across all the environments (i.e., 
robust association) is equivalent to testing the following 
consistent conditional independence hypothesis:
\@\label{eq:cst}
H_j^{\cst}:~\exists~e\in[E]\mbox{ such
that the null }H_{j}^{\ci,e} \mbox{ defined in}~\eqref{eq:e_null}
\mbox{ is true.}
\@
By definition, rejecting $H_j^{\cst}$ suggests an association 
between $X_j$ and $Y$ across all the environments. In 
practice,~\eqref{eq:cst} may be hard to reject especially 
as the number of environments increases.~\citet{li2021searching}
also propose a relaxed version of ~\eqref{eq:cst} that 
tests for partial consistency. For any fixed number 
$r \in [E]$, the partial consistency null hypothesis 
is defined as 
\@\label{eq:pcst}
H_j^{\text{pcst},r}: \Big|\Big\{ e\in[E]: 
H_{j}^{\ci,e}\mbox{ is true} \Big\}\Big| > E-r.
\@
Rejecting~\eqref{eq:pcst} then means that the association 
is true in at least $r$ environments.

To proceed, 
we let the set of nulls $\calH_0 = \{j: H_j^{\cst} \mbox{ is true}\}$ when 
testing $H_j^{\cst}$'s or $\calH_0 = \{j: H_j^{\text{psct}} \mbox{ is true}\}$ when 
testing $H_j^{\text{pcst}}$'s.
Previously,~\citet{li2021searching} 
proposed a variant of knockoffs called
the {\em multi-environment knockoff filter (MEKF)}
to simultaneously test $H_j^{\cst}$'s or $H_j^{\text{pcst}}$'s 
with guaranteed FDR control. To be specific, suppose we have independent
datasets $(\bX^1,\bY^1),\ldots,(\bX^E,\bY^E)$ where $(\bX^e,\bY^e)$
is obtained from environment $e$.
For each $e\in[E]$, the MEKF generates
a knockoff copy $\tilde{\bX}^e$ for $\bX^e$;
it then takes $([\bX^e,\tilde{\bX}^e],Y^e)_{e\in[E]}$
as input and computes the multi-environment 
feature importance statistic
${\bm W}\in \RR^{E\times p}$, whose
definition 
is given as
follows.
\begin{definition}[Definition 1 of \citet{li2021searching}]
$\bW \in \RR^{E\times p}$ are valid multi-environment knockoff
statistics if they satisfy $\bW \eqd \bW \odot {\bm \varepsilon} $
and ${\bm \varepsilon} \in \{\pm 1\}^{E\times p}$ is a random 
matrix with independent entries and rows $\varepsilon^e$  
such that $\varepsilon^e_j = \pm 1$ with probability 1/2 if
$H_j^{\ci,e}$ is true and $\varepsilon^e_j = +1$ otherwise,
for all $j \in[p]$ and $e\in[E]$.
\end{definition}
To construct such $\bW$, one can simply compute
$W^e = \calW([\bX^e,\tilde{\bX}^e], \bY^e)$ for
all $e\in[E]$ and assemble the $W^e$'s into a 
$E\times p$ matrix. More  methods for constructing
$\bW$ can be found in~\citet[Section 4.2]{li2021searching}.
With $\bW$, we define for each $j\in[p]$ the
(normal) feature importance statistic for testing $H_j^{\cst}$ to be
\@\label{eq:mekf}
W_j = \min_{e\in[E]}\big(\sign(\bW_j^e)\big) \cdot 
\prod_{e\in[E]} |\bW_j^e|,
\@
and that for testing $H_j^{\text{psct},r}$ to be 
\@\label{eq:mekf_pcst}
W_j = \sign\Big(\frac{1}{2} - p_j\Big) \cdot 
\prod_{e=1}^r |\bW_j|^{(E-e+1)}.
\@
Above, $|W_j|^{(k)}$ denotes the $k$-th largest 
element among $\{|W_j^e|\}_{e\in[E]}$, and 
\$
p_j = \Psi\Big(n_j^--1,(E-r+1-n_j^0) \vee 0, \frac{1}{2}\Big)
+ U_j \cdot \psi\Big(n_j^-, (E-r+1-n_j^0) \vee 0, \frac{1}{2}\Big),
\$
where $\Psi(x,m,\pi)$ is the binomial cumulative distribution function 
evaluated at $x$, $\psi(x,m,\pi)$ is the corresponding probability mass,
and $U_j \sim \text{Uniform}[0,1]$ that is independent of everything else;
$n_j^-$ and $n_j^0$ denote the number of negative signs and 
zeros in the $j$-the column of $\bW$, respectively.
Finally, applying the knockoff filter to the $W_j$'s
produces a selected set controlling FDR.

The MEKF procedure is random since it depends on
the one-time construction of knockoffs (and the auxiliary $U_j$'s); we can again 
derandomize it with the assistance of e-values. 
Specifically, for each $m\in[M]$, we construct  
the feature importance statistics $W^{(m)}$
as in~\eqref{eq:mekf} or~\eqref{eq:mekf_pcst} 
(based on the $m$-th draw of knockoffs, $\tilde{\bX}^{(m)}$); define then the stopping time
and e-values according to~\eqref{eq:early_stopping_time}
and~\eqref{eq:e_val_m} respectively.
For each $j \in [p]$, we aggregate the e-values from 
different runs by taking their average. 
The selected set is obtained by applying e-BH 
to the $e_j$'s.  We summarize the complete 
derandomized MEKF in Algorithm~\ref{alg:derandom_multienv};
the following theorem verifies the validity
of the derandomized MEKF.
\begin{theorem}
\label{thm:dmekf}
Consider testing consistent conditional independence hypotheses 
defined in~\eqref{eq:cst} or partial consistent conditional independent 
hypotheses defined in~\eqref{eq:pcst} for any $r\in[E]$.
For any $\alpha_{\kn} ,\alpha_{\ebh} \in (0, 1)$, any number of knockoff copies $M \ge1$, 
and any number of environments $E\geq 1$, 
the selected set offered by 
Algorithm~\ref{alg:derandom_multienv}
satisfies $\fdr\leq \alpha_{\ebh}$.
\end{theorem}   
\begin{proof}
When testing $H_j^\cst$'s, by Theorem 1 of~\citet{li2021searching}, we
have for any $m\in[M]$ that
$\sum_{j\in\calH_0}\EE[e_j^{(m)}] \le p$.
When testing $H_j^{\text{psct}}$'s, by Proposition 7 of~\citet{li2021searching},
we
have for any $m\in[M]$ that
$\sum_{j\in\calH_0}\EE[e_j^{(m)}] \le p$.
In either case, by the linearity of expectation, $\sum_{j\in\calH_0}\EE[e_j^{\textnormal{avg}}]\le p$.
Applying Theorem~\ref{thm:relaxed_ebh} completes
the proof.
\end{proof}

\begin{algorithm}[htbp]
\caption{Derandomized MEKF}\label{alg:derandom_multienv}
\begin{algorithmic}[1]
  \REQUIRE  Data $(\bX,\bY)$; 
  parameters $\alpha_{\ebh},\alpha_{\kn} \in (0,1)$ and 
  $M \in \mathbb{N}_+$\;\\ 
\vspace{0.05in}
\FOR{$m = 1,\ldots,M$}
\STATE Construct the knockoff copy $\tilde{\bX}^{(e,m)}$
for all $e\in[E]$.\;
\STATE Compute the multi-environment feature importance statistics: 
$\bW^{(m)} = \calW\big([\bX,\tilde{\bX}^{(m)}],\bY)$.\;
\STATE Transform the multi-environment feature importance
statistics $\bW^{(m)}$ into the feature importance statistics $W^{(m)}$
according to~\eqref{eq:mekf} when testing $H_j^{\cst}$'s or 
to~\eqref{eq:mekf_pcst} when testing $H_j^{\text{pcst}}$'s.\;
\STATE Compute the stopping time $T^{(m)}$
according to~\eqref{eq:early_stopping_time}.\;
\STATE Compute 
the e-value $e_j$ according
to~\eqref{eq:e_val_m}, for all $j\in[p]$.\\
\ENDFOR
\STATE Compute the average e-value $e_j^{\textnormal{avg}} 
= \frac{1}{M}\sum^M_{m=1} e_j^{(m)}$ for each $j\in[p]$.\;
\STATE Compute $\hat{k} = \max\big\{k: e^{\textnormal{avg}}_{(k)} \ge {p}/
(\alpha_{\ebh} k)\big\}$,
or $\hat k = 0$ if this set is empty.\;
\vspace{0.05in}
\ENSURE The selected set of discoveries 
$\calS_{\kn\textnormal{--derand}} 
\defn \big\{j \in [p]: e_j^{\textnormal{avg}} \ge p / (\alpha_{\ebh} \hat{k})\big\}$.
\end{algorithmic}
\end{algorithm}

\subsection{Knockoffs with side information}
Imagine a situation where each hypothesis $H_j$
is associated with some prior/side information 
$u_j\in\RR^r$. For example, in a genome-wide association
study (GWAS), $H_j$ corresponds to the association
between genetic variant $j$ and the response;
$u_j$ can be the prior knowledge of this genetic variant
a scientist has from previous related work.
For structured hypothesis testing, $u_j$ can be 
the physical location corresponding to the hypothesis,
and scientists know a priori that signals are sparse/clustered. 
In these examples, incorporating prior information
in the multiple testing procedure can potentially
boost the detecting power.  

Consider first a special case where $u_j \in \RR_+$
reflects the possibility of $j$ being a 
non-null (a higher value suggests an increased 
chance). We explicitly make use of $u_j$ when constructing
the e-values: for each $j \in [p]$ and
$m \in [M]$, define
\@
\label{eq:weighted_eval}
e_j^{(m)} \defn \frac{ pu_j \cdot \ind\big\{W_j^{(m)} \ge T^{(m)}\big\}}
{u_j + \sum_{k\in[p]} u_k \cdot \ind \big\{W_k^{(m)} \le -T^{(m)}\big\}},
\@
and aggregate the e-values across $m\in[M]$ 
by taking their average.
We then apply e-BH to the newly defined 
e-values and obtain a selected set of discoveries.
Intuitively, a large value of $u_j$ can boost the 
e-value of hypothesis $j$, thereby increasing the 
chance of selecting $j$ and improving the power
of our procedure. We refer to this weight-assisted 
multiple testing procedure the {\em weighted version}
of derandomized knockoffs with side information.

Alternatively,  we may have side information $u_j$ that
is of higher dimensions, or  it may not be explicit 
how the side information helps selection (it needs to 
be learned from the data). We can 
apply the adaptive knockoffs procedure~\citep{ren2020knockoffs}
and similarly aggregate the 
e-values associated with different knockoff copies.
Concretely, the adaptive knockoffs procedure generates
knockoffs and computes feature importance statistics
as in the original knockoffs procedure; it then
sequentially screens the hypotheses in an order determined
by the side information and the partially masked feature importance
statistics. At each step $k \in \{0,1,\ldots,p\}$,
let $P(k)$ denote the set of unscreened features (there
are in all $p-k$ of them) with positive feature importance
statistics and $N(k)$ those with negative feature importance
statistics. Define
\@
\label{eq:adaptive_stopping_time}
T = \inf\bigg\{k\in\{0,1,\ldots,p\}: 
  \frac{1 + |N(k)| }{|P(k)| \vee 1}
\le \alpha_\kn{}\bigg\},
\@
and the selected set of adaptive knockoffs 
is $P(T)$. When there are $M$ knockoff runs,
we define for any $m\in[M]$ correspondingly 
the set of unscreened features with positive $W_j$'s 
to be $P_m(k)$, the set of those with negative $W_j$'s
to be $N_m(k)$, and the stopping time to be $T^{(m)}$. 
The e-value for this $m$ and $j$ can be constructed as 
\@
\label{eq:adaptive_eval}
e^{(m)}_j = \frac{p\cdot \ind\big\{j\in P_m(T^{(m)})\big\}}
{1 + |N_m(T^{(m)})|}.
\@
Averaging the e-values over $m\in[M]$ and
applying e-BH to the $e_j$'s gives the set
of discoveries.
We refer to this method of utilizing side information as
the {\em adaptive version} of derandomized knockoffs
with side information.

The complete procedure of both versions
for derandomized knockoffs with side information
is summarized in Algorithm~\ref{alg:weighted_aggregate_knockoff}. 
The following theorem establishes 
the validity of both versions. 

\begin{theorem}
\label{thm:weighted_dkn}
For any $\alpha_{\kn} ,\alpha_{\ebh} \in (0, 1)$, any number of knockoff copies $M \ge1$, 
and any fixed side information $u\in\RR_+^p$, the selected set offered by both
versions described in
Algorithm~\ref{alg:weighted_aggregate_knockoff}
satisfies $\fdr\leq \alpha_{\ebh}$.
\end{theorem}   
The proof of Theorem~\ref{thm:weighted_dkn} can be 
found in Appendix~\ref{sec:proof_of_wdkn}.
Intuitively, the weighted version makes explicit 
use of the side information but only applies 
in a special case; on the other hand, the adaptive
version is more general, but need to pay the price
of learning when utilizing the side information.
As such, it remains interesting to investigate the
pros and cons of these two versions in different scenarios.

\begin{algorithm}[htbp]
\caption{Derandomized knockoffs with side information}\label{alg:weighted_aggregate_knockoff}
\begin{algorithmic}[1]
\REQUIRE  Data $(\bX,\bY)$; parameters $\alpha_{\ebh},\alpha_{\kn} \in (0,1)$,
$M \in \mathbb{N}_+$ and $u\in \RR_+^p$. \;\\ 
\vspace{0.05in}
\FOR{$m = 1,\ldots,M$}
\STATE Sample the knockoff copy $\tilde{\bX}^{(m)}$.\;
\STATE Compute the feature importance statistics: $W^{(m)} = \calW\big([\bX,\tilde{\bX}^{(m)}],\bY)$.\;
\STATE \texttt{Weighted version}: compute the stopping time $T^{(m)}$
according to~\eqref{eq:early_stopping_time}, and
the e-values according
to~\eqref{eq:weighted_eval}.\\
\texttt{Adaptive version}: compute the stopping time $T^{(m)}$
according to~\eqref{eq:adaptive_stopping_time},
and the e-values 
according to~\eqref{eq:adaptive_eval}.\;
\ENDFOR
\STATE Compute the average e-value $e_j^{\textnormal{avg}} = \frac{1}{M}\sum^M_{m=1} e_j^{(m)}$ for each $j\in[p]$.\;
\STATE Compute $\hat{k} = \max\big\{k: e^{\textnormal{avg}}_{(k)} \ge {p}/(\alpha_{\ebh} k)\big\}$,
or $\hat k = 0$ if this set is empty.\;
\vspace{0.05in}
\ENSURE The selected set of discoveries 
$\calS_{\kn\textnormal{--derand}} 
\defn \big\{j \in [p]: e_j^{\textnormal{avg}} \ge p / (\alpha_{\ebh} \hat{k})\big\}$.
\end{algorithmic}
\end{algorithm}

\section{Proofs for Section~\ref{sec:extensions_appendix}}

\subsection{Proof of Theorem~\ref{thm:robustness}}
\label{appx:proof_robustness}
Fix $m\in[M]$. For any $j\in\calH_0$, we additionally define
the ``masked'' feature importance statistics 
$W^{(m,j)} = (W^{(m)}_1,\ldots,|W^{(m)}_j|,\ldots,W^{(m)}_p)$,
and the stopping time induced by $W^{(m,j)}$ as 
\$
T^{(m,j)} = \inf\bigg\{t > 0: \frac{1 + 
  \sum_{k \in [p]}\ind\{W^{(m,j)}_k \le -t\}}
  {\sum_{\ell \in [p]} \ind\{W^{(m,j)}_{\ell}\ge t \}} 
\le \alpha_{\kn}\bigg\}.
\$
Observe that when $W^{(m)}_j \ge 0$,  it holds that
$T^{(m)} = T^{(m,j)}$. Consequently,
\$
\EE\big[e_j^{(m)} \cdot \ind\{\hat{\kl}^{\max}_j \le \eps\} \given W^{(m,j)} \big] 
=  p\cdot \EE\bigg[\frac{\ind\{W_j^{(m)} \ge T^{(m)}, \hat{\kl}^{\max}_j \le \eps\}}
{1 + \sum_{k \in[p]} \ind\{W_k^{(m)} \le -T^{(m)}\}}
\Biggiven W^{(m,j)}\bigg]\\
=  p\cdot \EE\bigg[\frac{\ind\{W_j^{(m)} \ge T^{(m,j)}, \hat{\kl}^{\max}_j \le \eps\}}
{1 + \sum_{k \in[p]} \ind\{W_k^{(m)} \le -T^{(m,j)}\}}
\Biggiven W^{(m,j)}\bigg]\\
\stepa{=}  p\cdot \frac{\ind\{|W^{(m)}_j| \ge T^{(m,j)}\}}
{1 + \sum_{k \neq j} \ind\{W^{(m)}_k \le -T^{(m,j)}\}}
\cdot \PP\big(W^{(m)}_j > 0, \hat{\kl}^{\max}_j\le \eps \given W^{(m,j)}\big),
\$
where step (a) follows from the fact that $T^{(m,j)}$
is a function of $W^{(m,j)}$. Define now $\{\bX_j^{(m,+)},\bX_j^{(m,-)}\}$
the unordered set of $(\bX_j,\tilde{\bX}^{(m)}_j)$ such that
when $\bX_j = \bX_j^{(m,+)}$ and $\tilde{\bX}^{(m)}_j = \bX_j^{(m,-)}$,
$W_j > 0$ and vice versa. Let 
\$
\rho^{(m)}_j = \sum_{i=1}^n \log\bigg(\frac{P_j\big(\bX^{(m,+)}_{ij} \given \bX_{i,-j}\big)}
{Q_j\big(\bX^{(m,+)}_{ij} \given \bX_{i,-j} \big)}
\cdot
\frac{Q_j\big(\bX^{(m,-)}_{ij} \given \bX_{i,-j}\big)}
{P_j\big(\bX^{(m,-)}_{ij} \given \bX_{i,-j} \big)} \bigg).
\$
Then
\$
& \PP\Big(W^{(m)}_j > 0, \hat{\kl}^{\max}_j \le \eps \biggiven \big\{\bX_j^{(m,+)}, \bX^{(m,-)}_j \big\},
\bX_{-j}, \tilde{\bX}_{-j}^{(m)},\bY \Big)\\
\le & 
\PP\Big(W^{(m)}_j > 0, \hat{\kl}^{(m)}_j \le \eps \biggiven \big\{\bX_j^{(m,+)}, \bX^{(m,-)}_j \big\},
\bX_{-j}, \tilde{\bX}_{-j}^{(m)},\bY \Big)\\
= & \PP\Big(\bX_j = \bX_j^{(m,+)}, \tilde{\bX}_j = \bX^{(m,-)}, \rho_j^{(m)} \le \eps
\biggiven \big\{\bX_j^{(m,+)}, \tilde{\bX}^{(m,-)}_j\big\},
\bX_{-j}, \tilde{\bX}_{-j}^{(m)},\bY\Big)\\
\stepa{=} &e^{\rho^{(m)}_j} \cdot \ind\{\rho^{(m)}_j \le \eps\}  
\cdot \PP\Big(\bX_j = \bX_j^{(m,-)}, \tilde{\bX}_j = \bX^{(m,+)}
\biggiven \big\{\bX_j^{(m,+)}, \tilde{\bX}^{(m,-)}_j\big\},
\bX_{-j}, \tilde{\bX}_{-j}^{(m)},\bY\Big)\\
\le & e^{\eps}
\cdot \PP\Big(\bX_j = \bX_j^{(m,-)}, \tilde{\bX}_j = \bX^{(m,+)}
\biggiven \big\{\bX_j^{(m,+)}, \tilde{\bX}^{(m,-)}_j\big\},
\bX_{-j}, \tilde{\bX}_{-j}^{(m)},\bY\Big)\\
= & e^{\eps}
\cdot \PP\Big(W^{(m)}_j < 0
\biggiven \big\{\bX_j^{(m,+)}, \tilde{\bX}^{(m,-)}_j\big\},
\bX_{-j}, \tilde{\bX}_{-j}^{(m)},\bY\Big),
\$
where step (a) follows from Lemma 1 of~\citet{barber2020robust}.
Next, we shall make
use of the following lemma that connects $T^{(m,j)}$ to $T^{(m,\ell)}$, whose 
proof is exactly the same as Lemma 6 of~\citet{barber2020robust}, so we 
omit the proof here.
\begin{lemma}\label{lemma:tktl}
For any $j,\ell$, if $W^{(m)}_j \le - \min(T^{(m,j)},T^{(m,\ell)})$ and 
$W^{(m)}_{\ell} \le - \min(T^{(m,j)}, T^{(m,\ell)})$, then
$T^{(m,j)} = T^{(m,\ell)}$.
\end{lemma}
Using the above and the tower property, we have
\$
\EE\big[e_j^{(m)} \cdot \ind \{\hat{\kl}^{\max}_j \le \eps\}\big]
\le & p\cdot e^{\eps} \cdot \EE\bigg[\frac{\ind\{W^{(m)}_j \le -T^{(m)}_j\}}
{1+ \sum_{k\neq j} \ind \{W_k^{(m)} \le -T_j^{(m)}\}}\bigg]\\
= &p\cdot e^{\eps} \cdot \EE\bigg[\frac{\ind\{W^{(m)}_j \le -T^{(m)}_j\}}
{1+ \sum_{k\neq j} \ind \{W_k^{(m)} \le -T_k^{(m)}\}}\bigg].
\$
The last equality is due to Lemma~\ref{lemma:tktl}. Summing over
$ m\in [M]$ and $j\in\calH_0$, we have
\$
\sum_{j\in \calH_0} \EE\big[e_j^{\textnormal{avg}}\cdot \ind\{\hat{\kl}^{\max}_j \le\eps\}\big]
= &\frac{1}{M}\sum_{j\in \calH_0} \sum_{m=1}^M\EE\big[e_j^{(m)} \cdot \ind\{\hat{\kl}^{\max}_j\le\eps\}\big]\\
\le& \frac{p}{M}\sum^M_{m=1} e^{\eps}
\EE\bigg[\frac{\sum_{j\in\calH_0} \ind \{W^{(m)}_j \le -T_j^{(m)}\}}
{\sum_{k\in\calH_0} \ind\{W_k^{(m)} \le -T^{(m)}_k \}}\bigg] \\
\le & p\cdot e^{\eps}.
\$
Finally, 
\$
\EE\bigg[\frac{\big|\big\{j: j\in \calS_{\ebh} \cap \calH_0, \hat{\kl}^{\max}_j \le \eps\big\}\big|}{|\calS_{\ebh}|\vee 1}\bigg]
= & \sum_{j\in\calH_0} \EE\bigg[\frac{\ind\{e_j^{\textnormal{avg}} \geq \frac{p}{\alpha|\calS_{\ebh}|}\}\cdot \ind\{\hat{\kl}^{\max}_j\le \eps\}}
{|\calS_{\ebh}| \vee 1}\bigg]\\
= & \sum_{j\in\calH_0}\sum^p_{k=1}
\EE\bigg[\frac{\ind\{e_j^{\textnormal{avg}}\ge \frac{p}{\alpha k}, |\calS_{\ebh}| = k\}\cdot \ind\{\hat{\kl}^{\max}_j\le \eps\}}
{k}\bigg]\\
\le & \frac{\alpha}{p} \sum_{j\in \calH_0}\sum_{k=1}^p
\EE\big[e_j^{\textnormal{avg}} \cdot \ind\{|\calS_{\ebh}| = k\}\cdot\ind\{ \hat{\kl}^{\max}_j \le \eps\}\big]\\
\le & \frac{\alpha}{p}\sum_{j\in \calH_0}\EE[e_j^{\textnormal{avg}} \cdot \ind\{\hat{\kl}^{\max}_j\le \eps\}]\\ 
\le & e^{\eps}\cdot \alpha,
\$
completing the proof.

\subsection{Proof of Theorem~\ref{thm:weighted_dkn}}
\label{sec:proof_of_wdkn}
By Theorem~\ref{thm:relaxed_ebh}, for both versions of the algorithm it suffices to check that the $e_j$'s satisfy
the relaxed e-value condition~\eqref{eqn:evals_relaxed}.
\paragraph{Weighted version}
We continue using the notation for masked feature importance
statistics and the corresponding stopping time.
For any $m\in[M]$, $j\in \calH_0$ and the e-value
defined in~\eqref{eq:weighted_eval}, we have
\$
e_j^{(m)} = \frac{pu_j \cdot \ind\{W_j^{(m)} \ge T^{(m)}\}}
{u_j + \sum_{\ell \neq j} u_\ell  \cdot \ind\{W_{\ell}^{(m)} \le -T^{(m)}\}} 
= \frac{pu_j \cdot \ind\{W_j^{(m)} \ge T^{(m,j)}\}}
{u_j + \sum_{\ell \neq j} u_{\ell}  \cdot \ind\{W_{\ell}^{(m)} \le -T^{(m,j)}\}} 
\$
Since $T^{(m,j)}$ is fully deterministic given $W^{(m,j)}$,
we have
\@
\label{eq:weighted_eq1}
\EE\big[e_j^{(m)} \biggiven W^{(m,j)}\big]
= & \PP\big(W_j^{(m)} > 0 \given W^{(m,j)}\big) \cdot 
\frac{pu_j \cdot \ind \big\{|W_j^{(m)}| \ge T^{(m,j)}\big\}}
{u_j + \sum_{\ell \neq j} u_{\ell} \cdot
\ind \big\{ W_{\ell}^{(m)} \le -T^{(m,j)}\big\}} \nonumber\\
\stackrel{\rm (a)}{=} &\PP\big(W_j < 0 \given W^{(m,j)}\big) \cdot 
\frac{pu_j \cdot \ind \big\{|W_j^{(m)}| \ge T^{(m,j)}\big\}}
{u_j + \sum_{\ell \neq j} u_{\ell} \cdot
\ind \big\{ W_{\ell}^{(m)} \le -T^{(m,j)}\big\}}\nonumber\\ 
= & \EE \bigg[\frac{p u_j \cdot \ind\big\{W_j^{(m)} \le- T^{(m,j)}\big\}}
{u_j + \sum_{\ell \neq j} u_{\ell} \cdot 
\ind\big\{W^{(m)}_{\ell} \le -T^{(m,j)}\big\}} \bigggiven W^{(m,j)}\bigg],
\@
where step (a) uses the property
of knockoffs that conditional
on the magnitudes $(|W_1|,|W_2|,\ldots,|W_p|)$,
the signs of the null $W_j$'s
are i.i.d.~coin flips (c.f.~\citet[Lemma 3.3]{candes2018panning}).
By Lemma~\ref{lemma:tktl},
\$
\eqref{eq:weighted_eq1} = 
\EE \bigg[\frac{p u_j \cdot \ind\big\{W_j^{(m)} \le- T^{(m,j)}\big\}}
{u_j + \sum_{\ell \neq j} u_{\ell} \cdot 
\ind\big\{W^{(m)}_{\ell} \le -T^{(m,\ell)}\big\}} \bigggiven W^{(m,j)}\bigg]
= \EE \bigg[ \frac{p u_j \cdot \ind\big\{W_j^{(m)} \le- T^{(m,j)}\big\}}
{\sum_{\ell \in [p]} u_{\ell} \cdot 
\ind\big\{W^{(m)}_{\ell} \le -T^{(m,\ell)}\big\}} \bigggiven W^{(m,j)} \bigg] .
\$
Summing over $j \in \calH_0$, we have 
\$
\sum_{j\in\calH_0}\EE\big[e_j^{(m)}\big] = 
p \cdot \EE\Big[\frac{\sum_{j\in\calH_0} u_j \cdot \ind\{W_j^{(m)} \le -T^{(m,j)}\}}
{\sum_{\ell \in [p]} u_{\ell} \cdot \ind\{W_{\ell} \le -T^{(m,j)}\}}\Big] 
\le p \cdot \EE\Big[\frac{\sum_{j\in\calH_0} u_j \cdot \ind\{W_j^{(m)} \le -T^{(m,j)}\}}
{\sum_{\ell \in \calH_0} u_{\ell} \cdot \ind\{W_{\ell} \le -T^{(m,j)}\}}\Big] 
=p.
\$
We then average over $m\in[M]$ and obtain that
\$
\sum_{j\in\calH_0} \EE[e_j^{\textnormal{avg}}] = 
\sum_{j\in\calH_0} \EE\Big[\frac{1}{M} \sum^M_{m=1} e_j^{(m)}\Big]
= \frac{1}{M}\sum^M_{m=1}\sum_{j\in\calH_0} \EE\big[e_j^{(m)}\big]
\le p.
\$

\paragraph{Adaptive version}
For any $m\in[M]$ and 
the e-value defined in~\eqref{eq:adaptive_eval},
\$
\sum_{j\in\calH_0}\EE\big[e_j^{(m)}\big]
= p \cdot \sum_{j \in \calH_0} 
\EE\bigg[\frac{\ind\big\{j \in P_m(T^{(m)})\big\}}
{1+ |N_m(T^{(m)})|}\bigg] 
\le   p \cdot \EE\bigg[\frac{\big|\big\{j\in\calH_0, j\in P_m(T^{(m)})\big\}\big|}
{1+ \big|\big\{j\in\calH_0, j\in N_m(T^{(m)}) \big\}\big| }\bigg]
\le p,
\$
where the last inequality follows from Theorem 1 of~\citet{ren2020knockoffs}.
We then have $\sum_{j\in\calH_0} \EE[e_j^{\textnormal{avg}}] \leq p$ by linearity of expected value.

\section{Additional simulations}
\label{sec:add_simulations}
In this section, we collect results from  additional simulations
to evaluate our proposed method under a variety of scenarios.

\subsection{The effect of $c$ and $\alpha_{\kn}$}
Under the two simulation settings in Section~\ref{sec:simulation},
we additionally investigate the joint effect of $c$ and $\alpha_{\kn}$
on the performance of derandomized knockoffs. Here, we consider 
$\alpha_{\ebh} = 0.1$ and $0.2$; for the former choice, we let 
$\alpha_{\kn}$ range in $\{0.01,0.02,\ldots,0.2\}$ and in $\{0.02,0.04,\ldots,0.4\}$
for the latter. In both cases, $c \in \{0.1,0.2,\ldots,2\}$. 
Figure~\ref{fig:heatmap} and~\ref{fig:heatmap_2} are heatmaps
of the proposed procedure's power resulting from different choices of $(c,\alpha_{\kn})$ 
with $\alpha_{\ebh} =0.1$ and $\alpha_{\ebh} = 0.2$, respectively.
Figure~\ref{fig:alpha_2} further plots the power and FDR as functions
of $\alpha_{\kn}$ when $\alpha_{\ebh}=0.2$.
In all cases, we observe that 
derandomized knockoffs consistently performs well with 
the choice $c=1$ and $\alpha_{\kn} = \alpha_{\ebh}/2$, 
justifying our intuition that $\alpha_{\kn} = \alpha_{\ebh}/2$
is a good default choice.

\begin{figure}[h]
\centering 
\rotatebox{90}{Gaussian}
~
\begin{minipage}{0.3\textwidth}
\includegraphics[width = \textwidth]{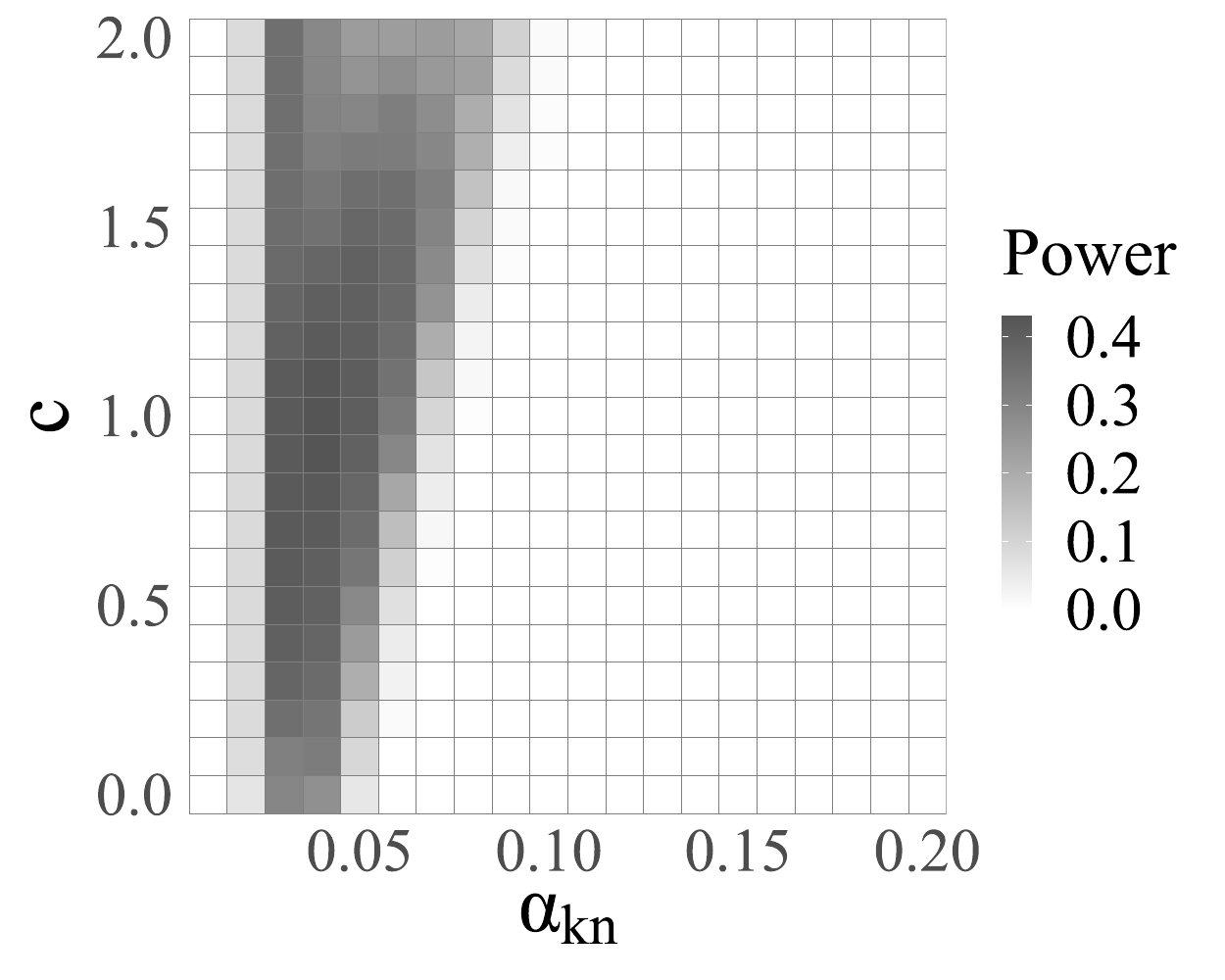}
\end{minipage}
\begin{minipage}{0.3\textwidth}
\includegraphics[width = \textwidth]{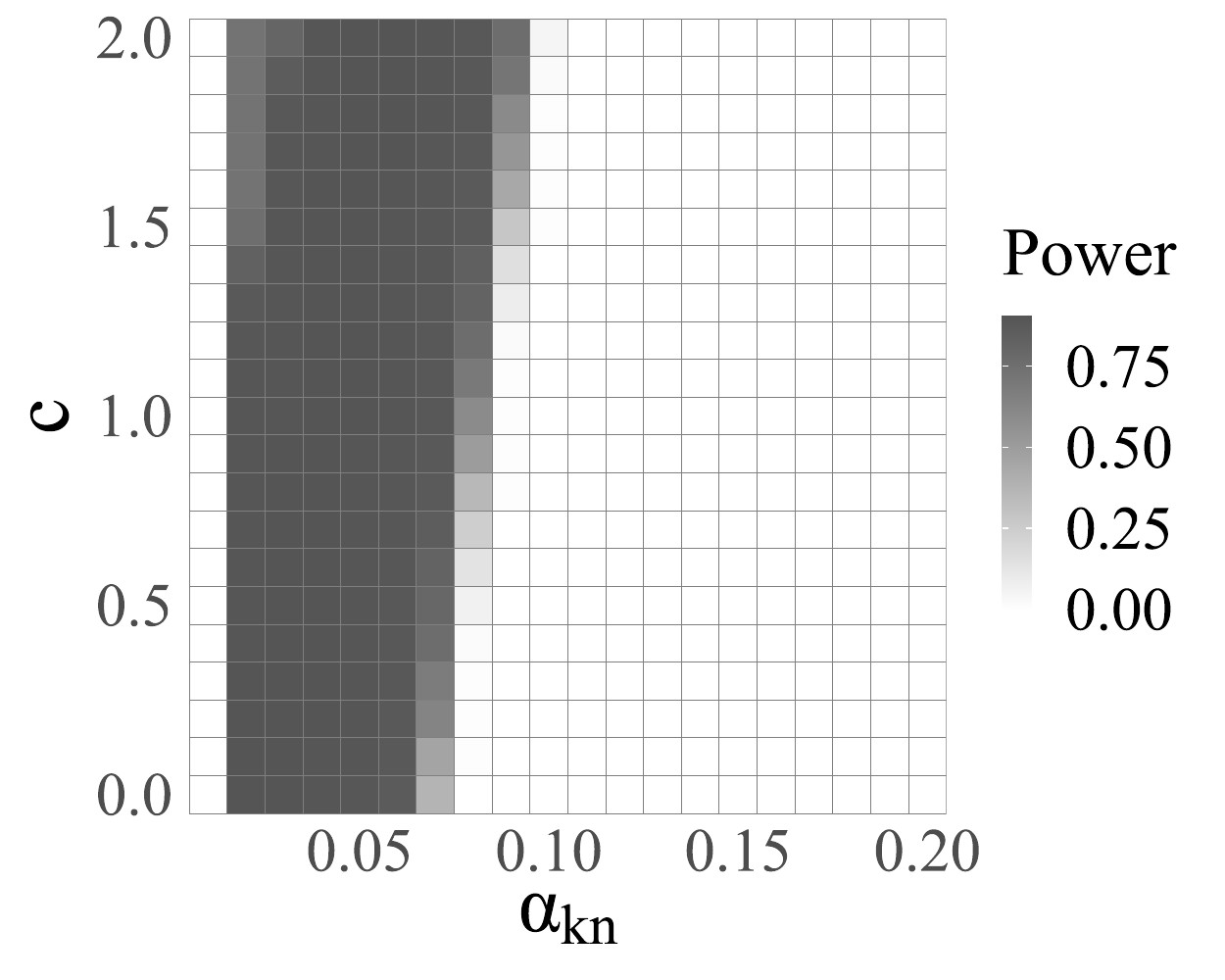}
\end{minipage}
\begin{minipage}{0.3\textwidth}
\includegraphics[width = \textwidth]{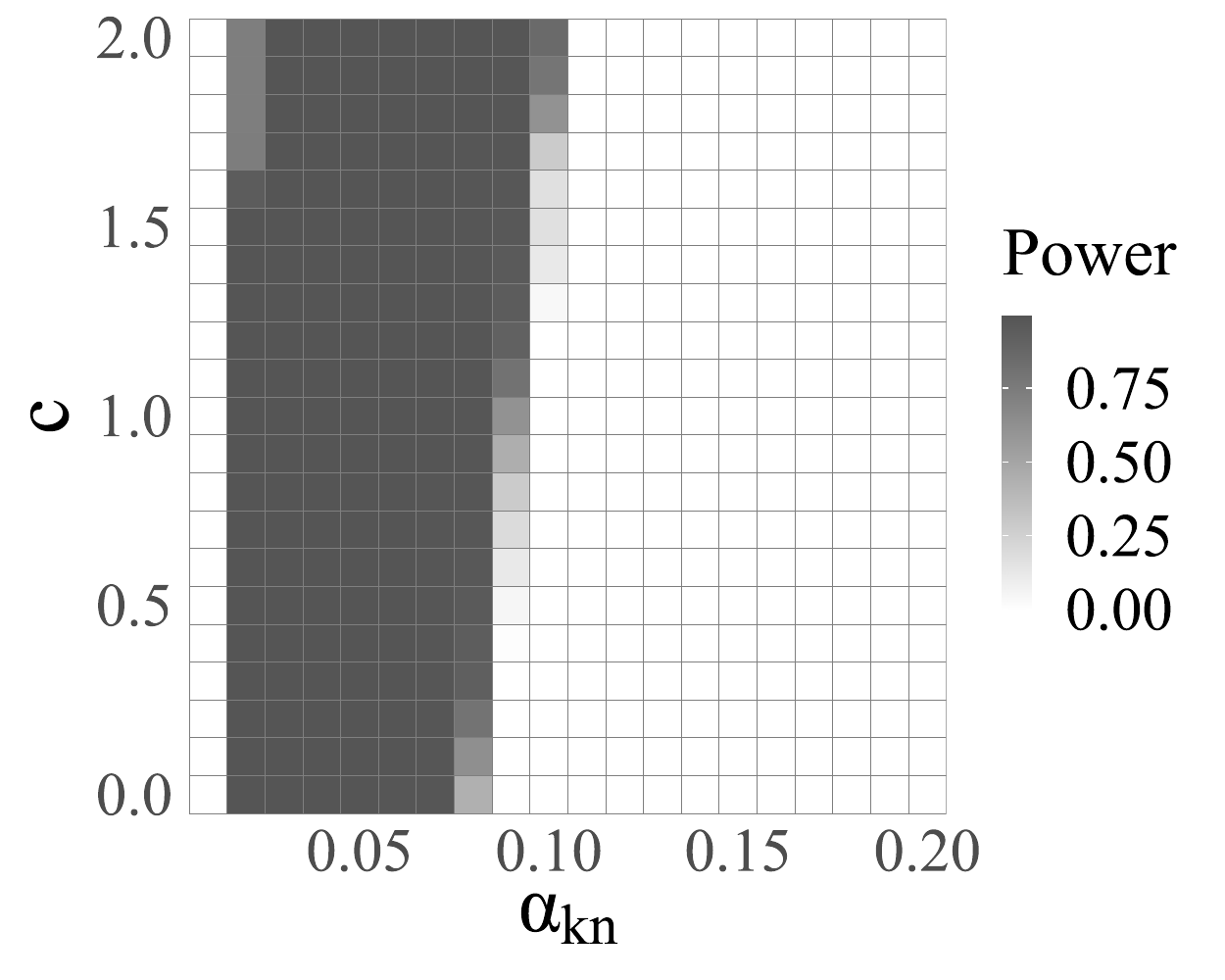}
\end{minipage}\\
\rotatebox{90}{Logistic}
~
\begin{minipage}{0.3\textwidth}
\includegraphics[width = \textwidth]{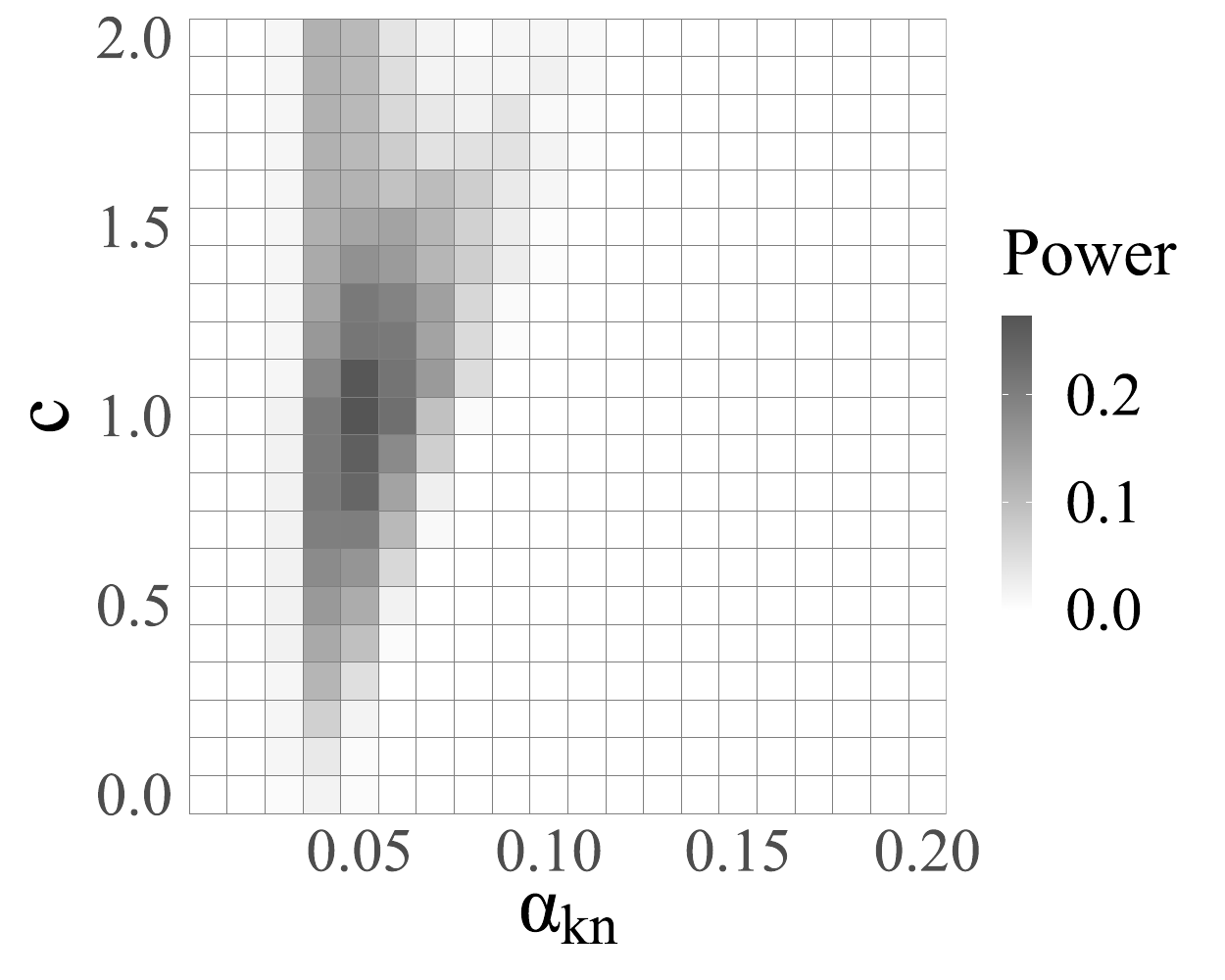}
\end{minipage}
\begin{minipage}{0.3\textwidth}
\includegraphics[width = \textwidth]{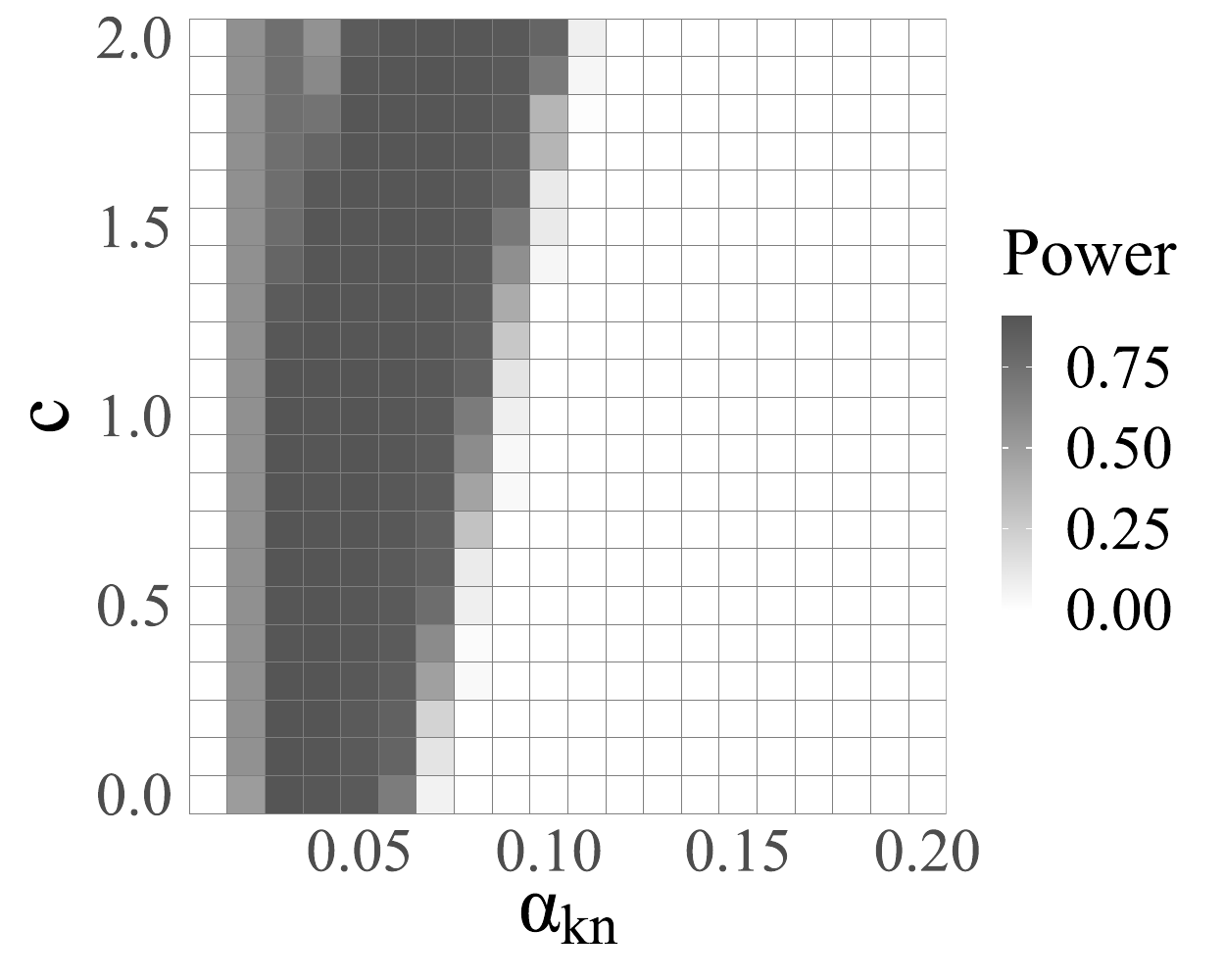}
\end{minipage}
\begin{minipage}{0.3\textwidth}
\includegraphics[width = \textwidth]{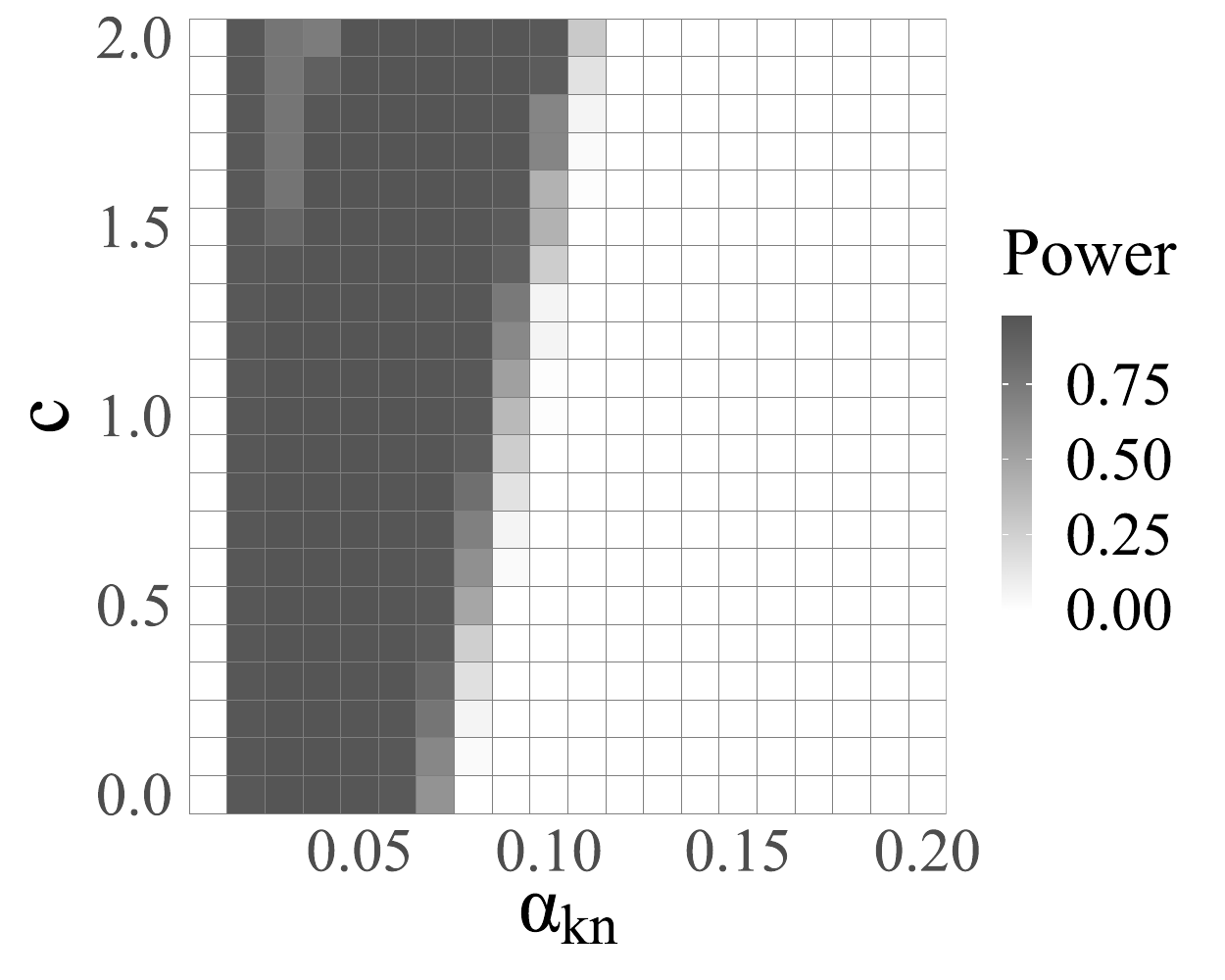}
\end{minipage}
\caption{Heatmaps of the power of derandomized knockoffs with different choices 
of $(c,\alpha_{\kn})$. The left, middle and right columns correspond to low, medium 
and high signal amplitude, respectively. A darker color represents a higher value. 
In these experiments, we set $\alpha_{\ebh} = 0.1$, and we observe that power is consistently 
high around $\alpha_{\kn}=0.05$, justifying our intuition that $\alpha_{\kn} = \alpha_{\ebh}/2$
is a good default setting.}
\label{fig:heatmap}
\end{figure}

\begin{figure}[h]
\centering 
\rotatebox{90}{Gaussian}
~
\begin{minipage}{0.3\textwidth}
\includegraphics[width = \textwidth]{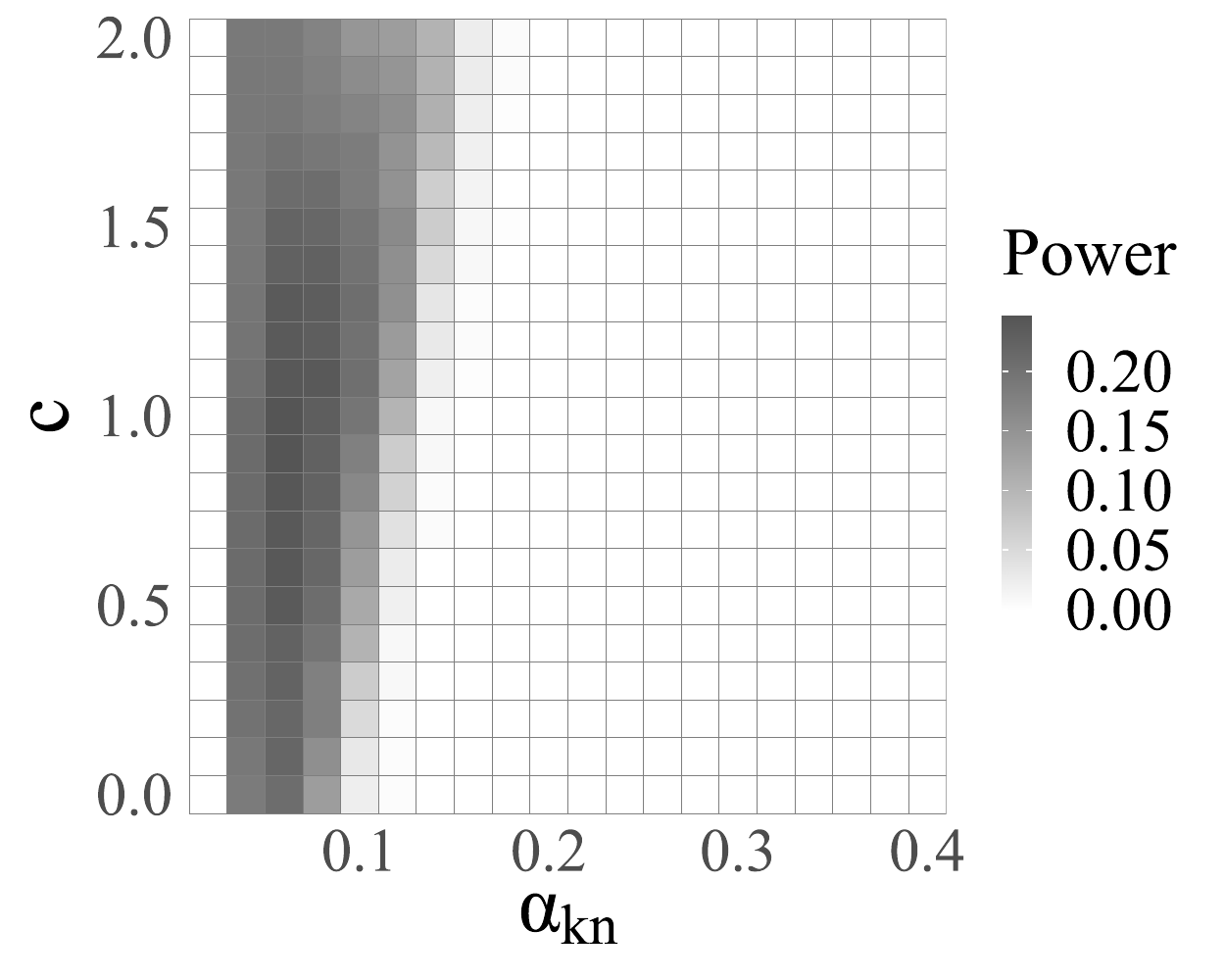}
\end{minipage}
\begin{minipage}{0.3\textwidth}
\includegraphics[width = \textwidth]{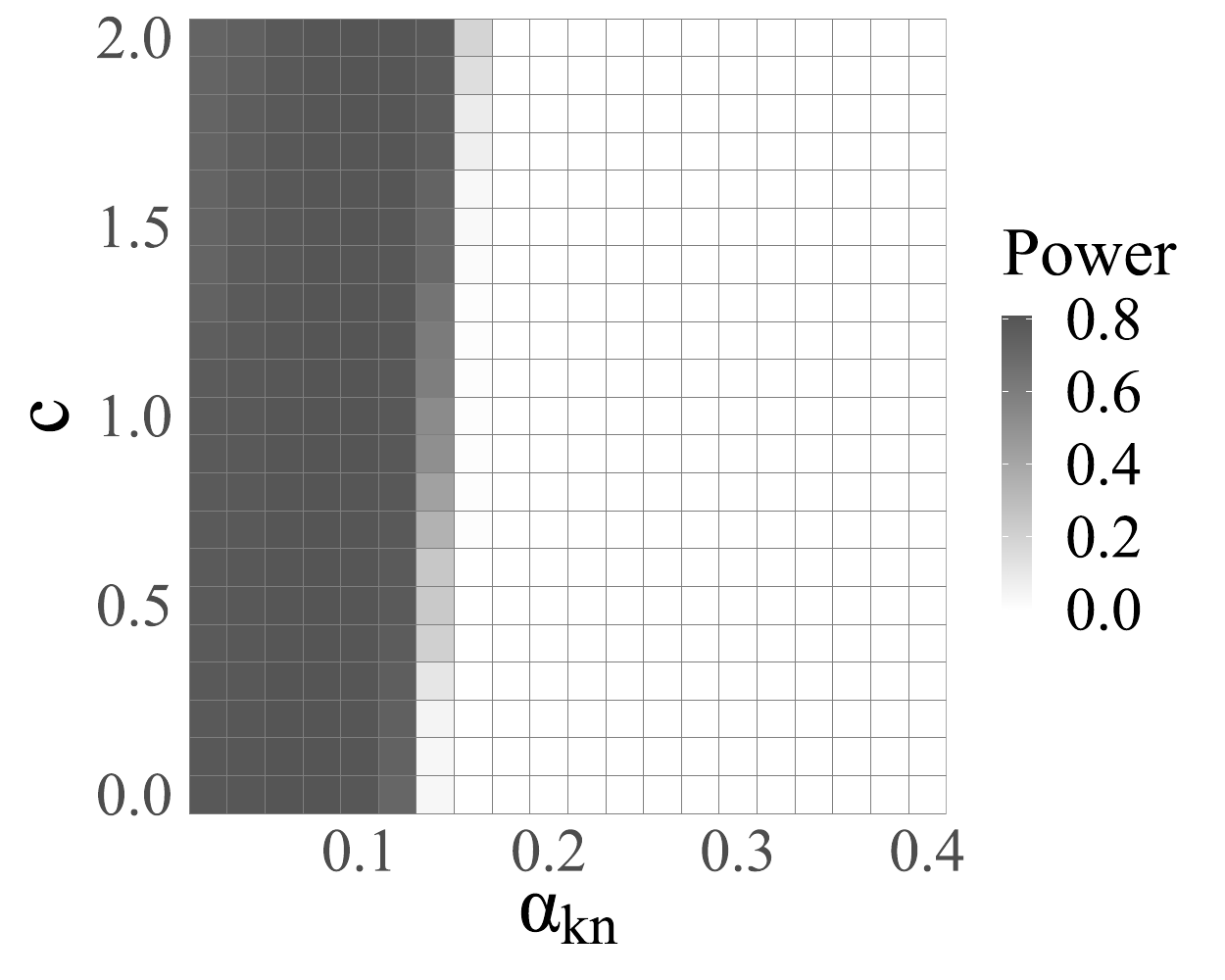}
\end{minipage}
\begin{minipage}{0.3\textwidth}
\includegraphics[width = \textwidth]{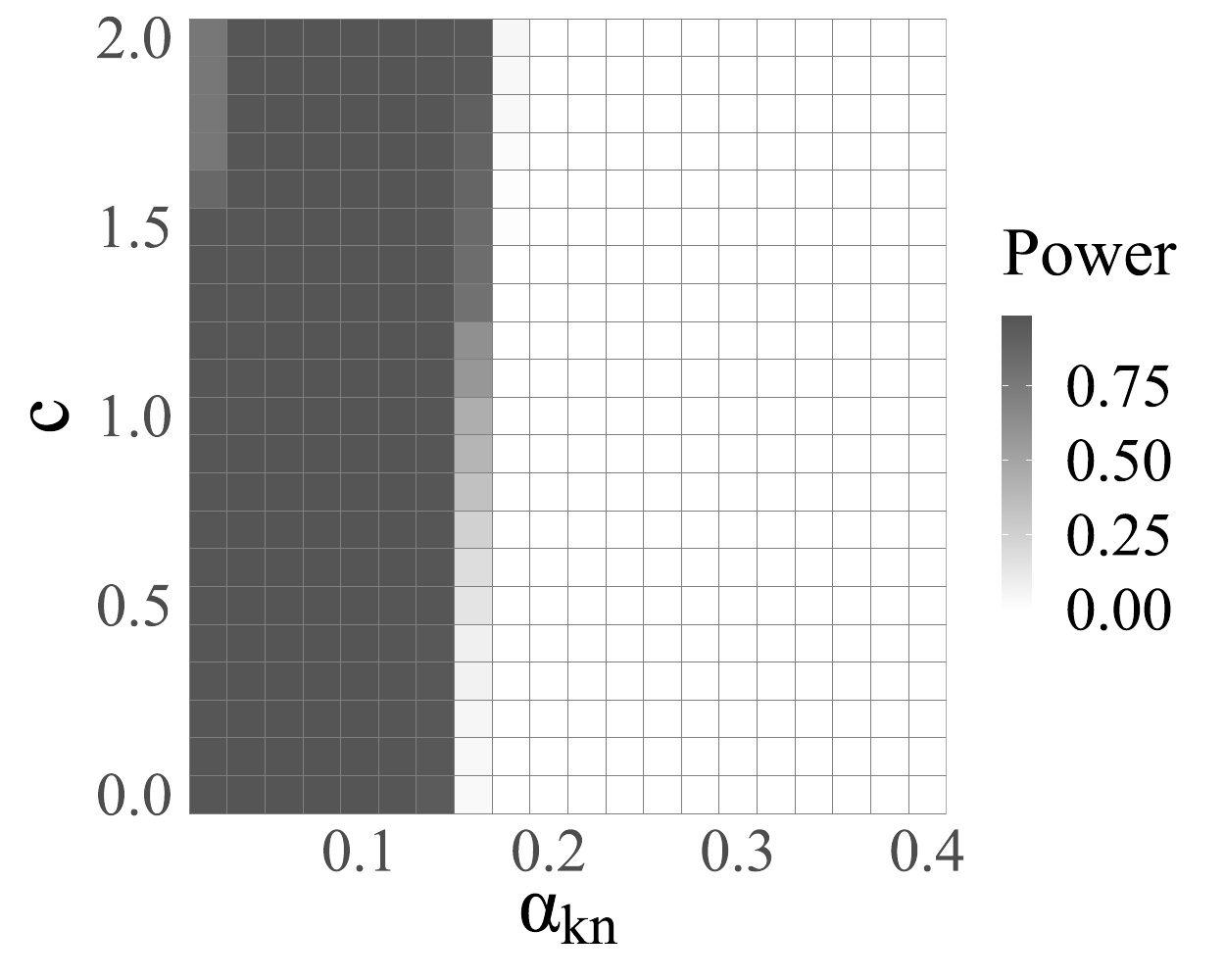}
\end{minipage}\\
\rotatebox{90}{Logistic}
~
\begin{minipage}{0.3\textwidth}
\includegraphics[width = \textwidth]{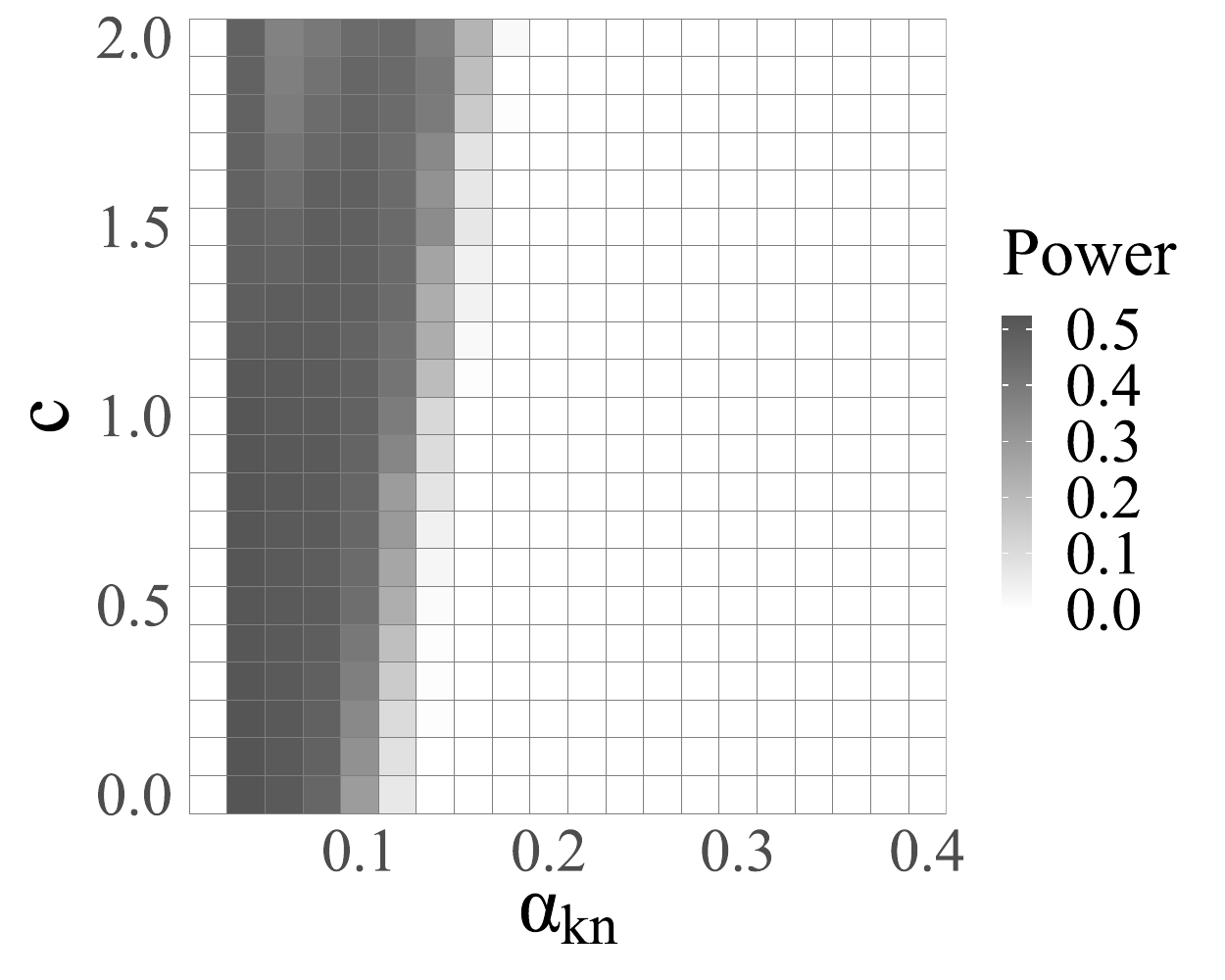}
\end{minipage}
\begin{minipage}{0.3\textwidth}
\includegraphics[width = \textwidth]{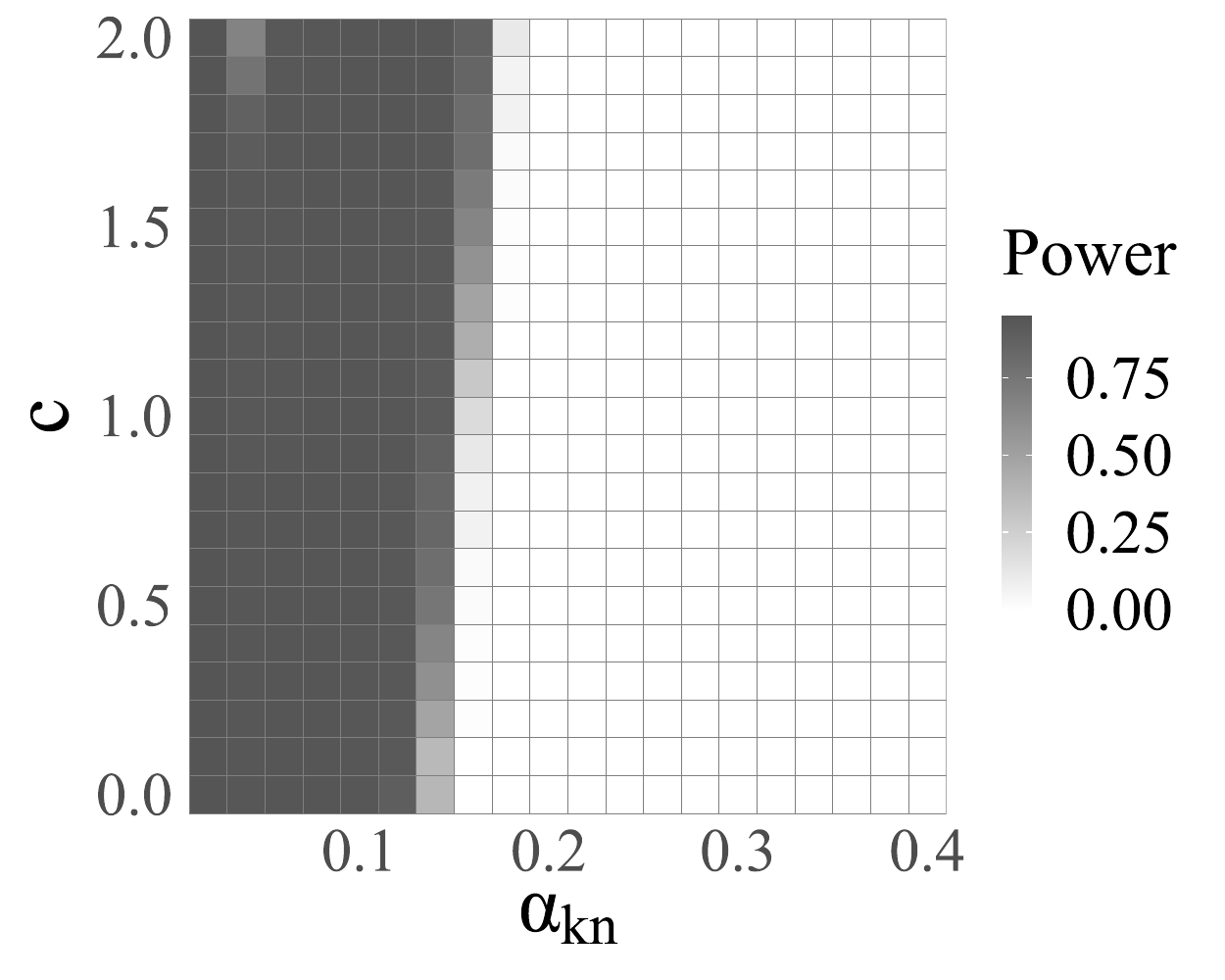}
\end{minipage}
\begin{minipage}{0.3\textwidth}
\includegraphics[width = \textwidth]{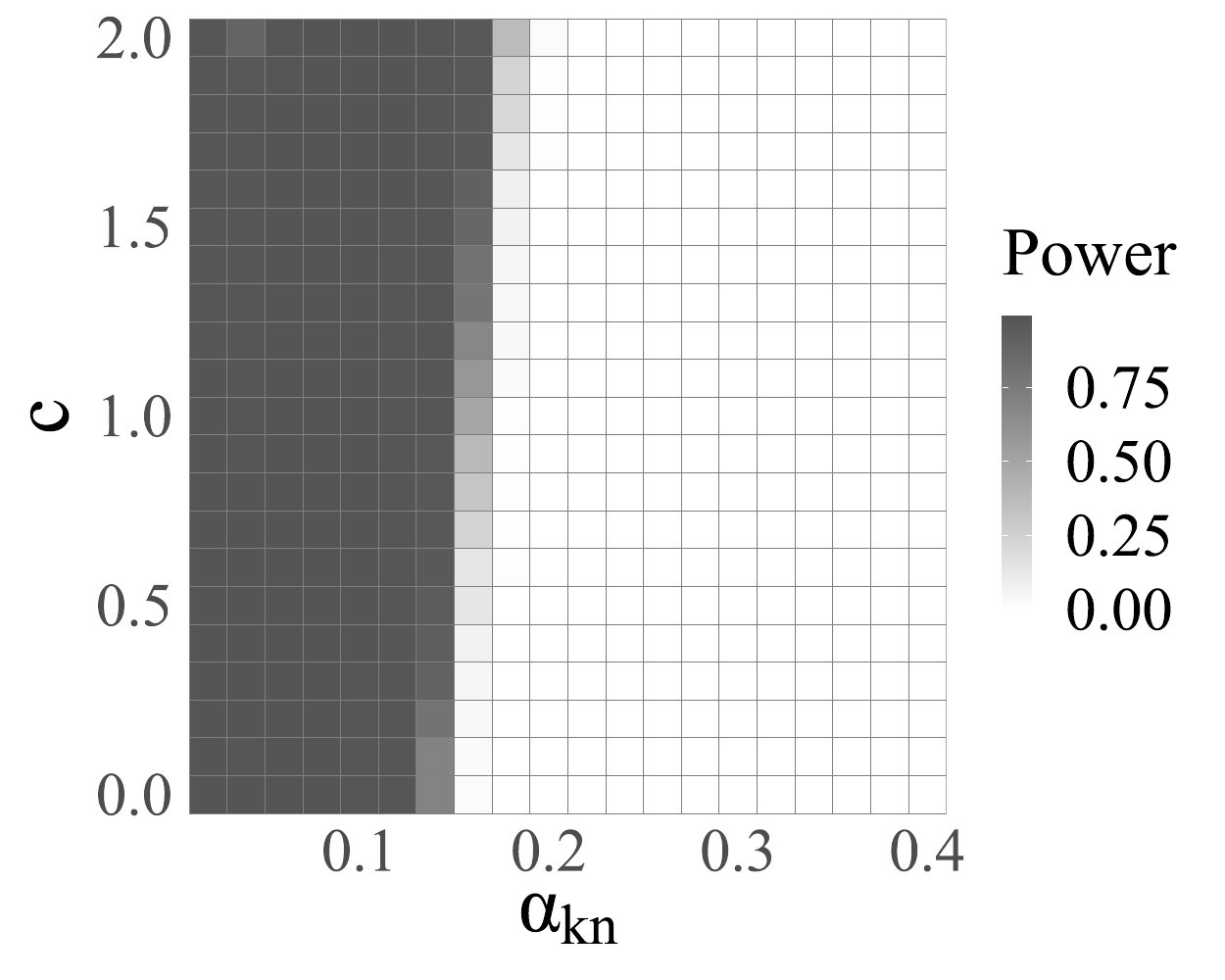}
\end{minipage}
\caption{Heatmaps of the power of derandomized knockoffs with different choices 
of $(c,\alpha_{\kn})$. The target FDR level $\alpha_{\ebh} = 0.2$, and the power is 
consistently high around $\alpha_{\kn} = 0.1$. 
The other details are the same as in Figure~\ref{fig:heatmap}.}
\label{fig:heatmap_2}
\end{figure}

\begin{figure}[h]
\centering 
\rotatebox{90}{Gaussian}
\begin{minipage}{0.45\textwidth}
\centering
\includegraphics[width = 0.8\textwidth]{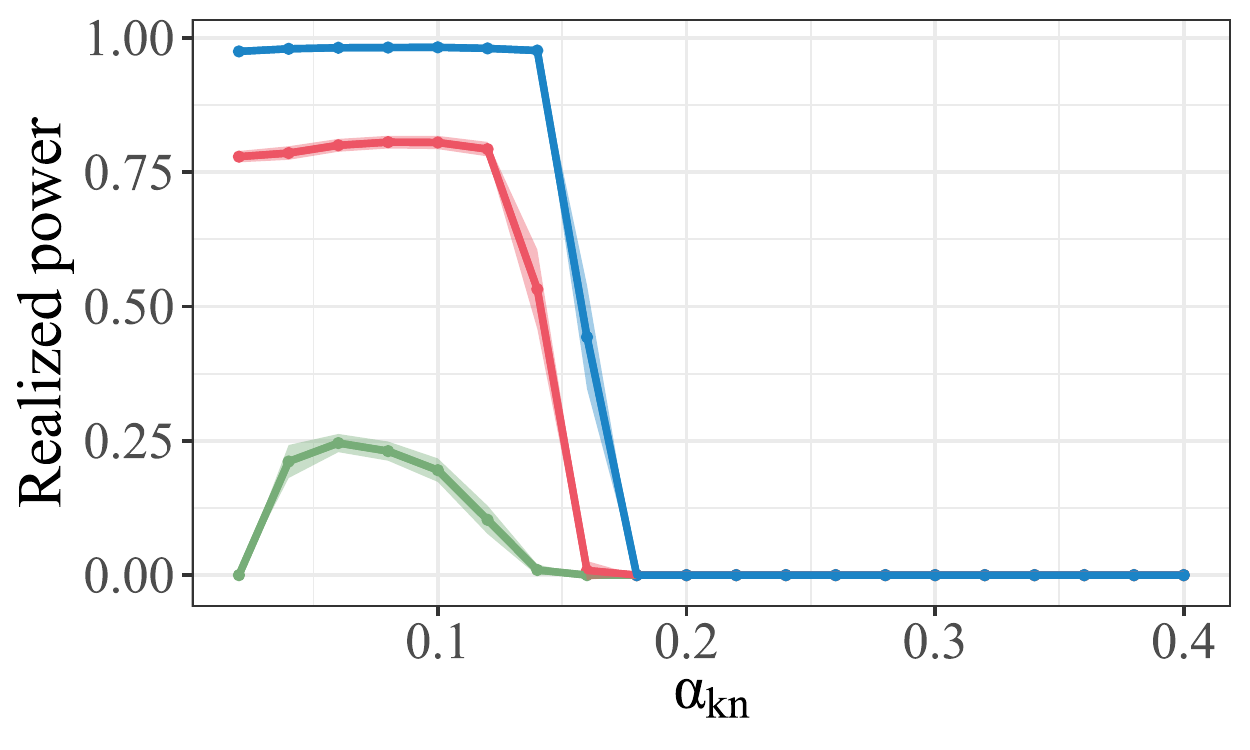}
\end{minipage}
\begin{minipage}{0.45\textwidth}
\centering
\includegraphics[width = 0.8\textwidth]{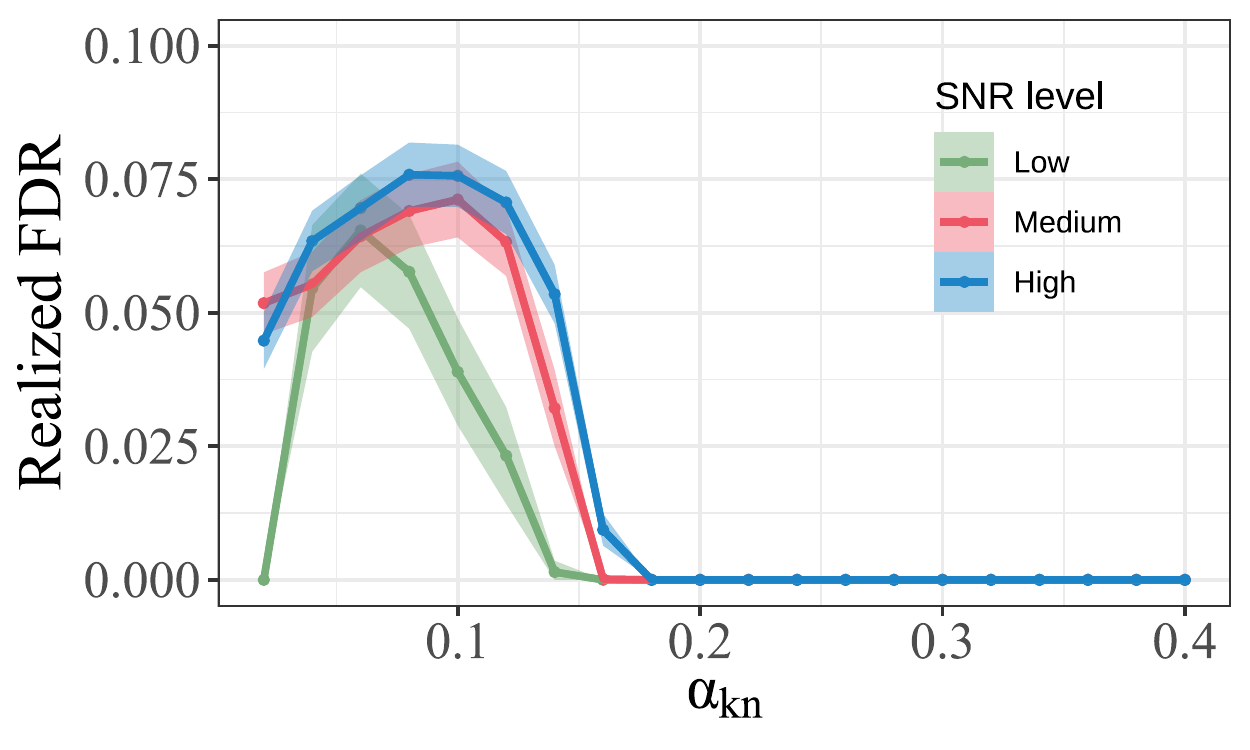}
\end{minipage}\\
\rotatebox{90}{Logistic}
\begin{minipage}{0.45\textwidth}
\centering
\includegraphics[width = 0.8\textwidth]{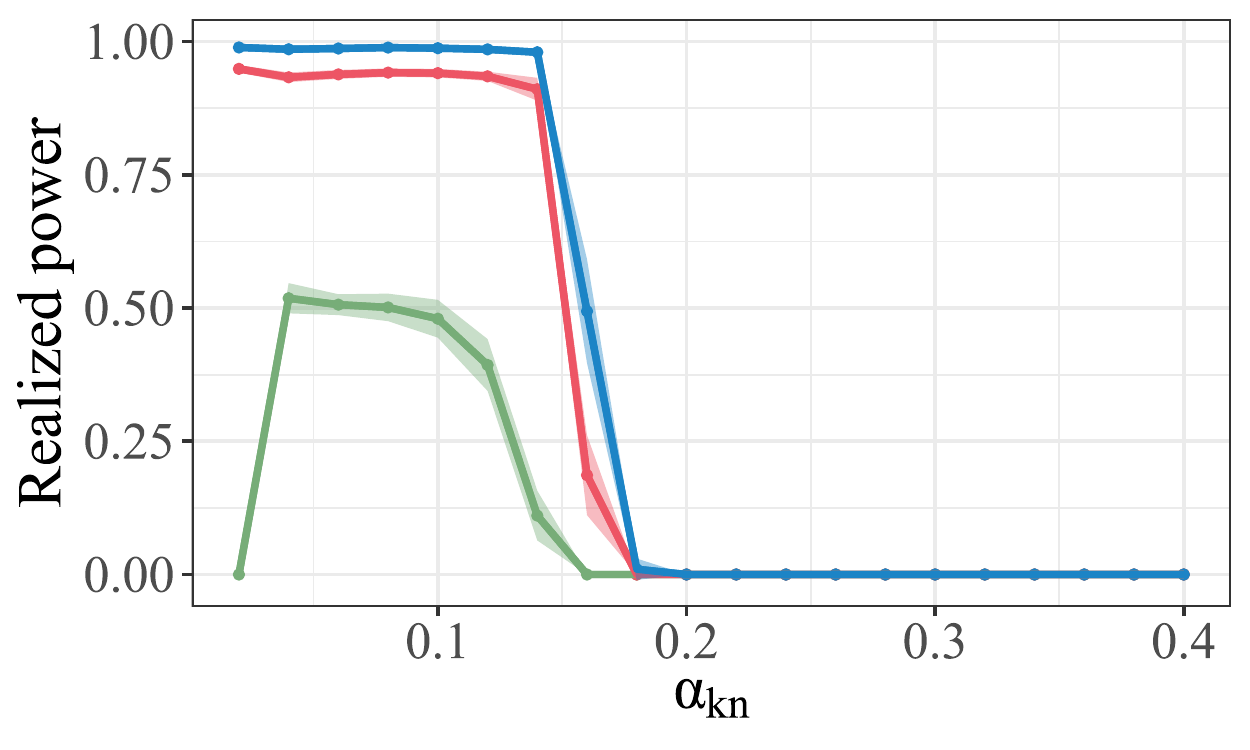}
\end{minipage}
\begin{minipage}{0.45\textwidth}
\centering
\includegraphics[width = 0.8\textwidth]{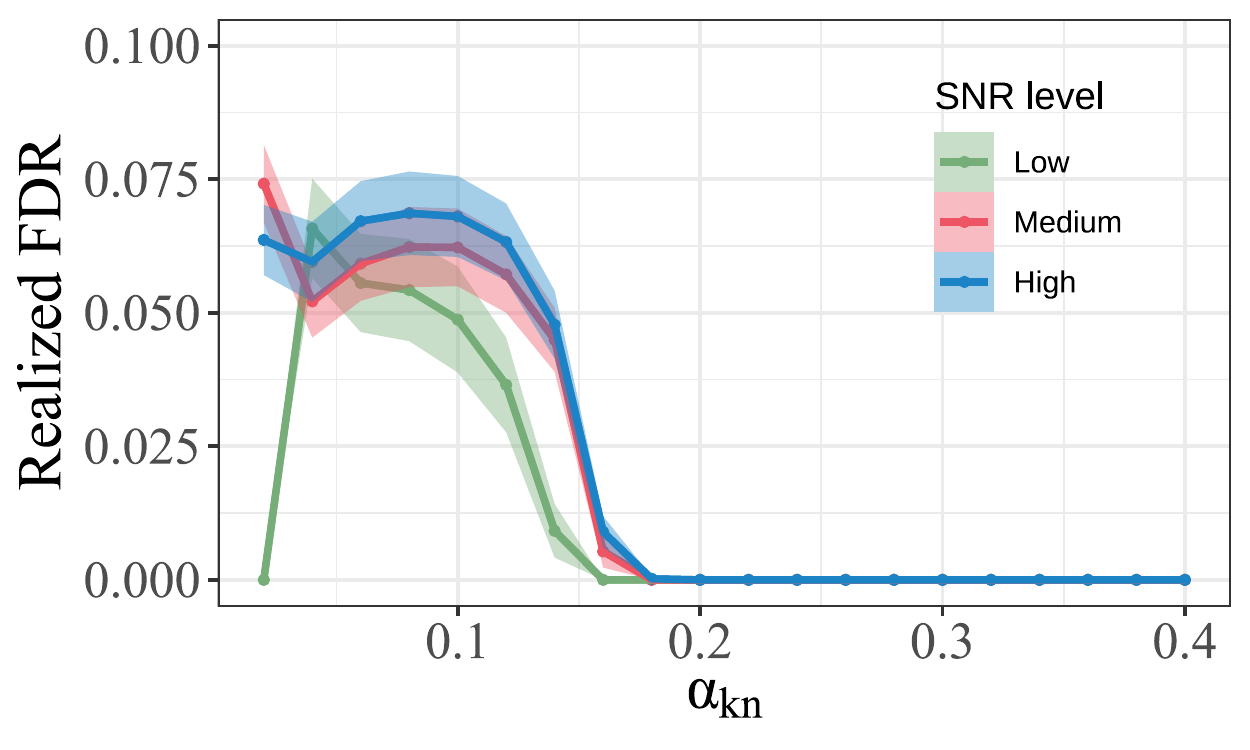}
\end{minipage}
\caption{Realized power (left) and FDR (right) of derandomized 
knockoffs as a function of the parameter $\alpha_{\kn}$ for the simulation 
data experiments.  The target FDR level $\alpha_{\ebh} = 0.2$. The details 
are otherwise the same as in Figure~\ref{fig:alpha}.}
\label{fig:alpha_2}
\end{figure}
\subsection{A high-dimensional setting}
Consider $n=800$ samples and $p=1000$ features, with
$|\calH_0| = 100$. The marginal distribution of the 
features $P_X = \calN(0,\Sigma)$ where 
$\Sigma_{jk} = 0.5^{|j-k|}$ for all $j,k\in[p]$.
The model of $Y\given X$ is the same Gaussian
linear model as in Section~\ref{sec:simulation};
the coefficient vector $\beta$ is given by
\$
\beta = \Big(\underbrace{0,\ldots,0}_{9},
\frac{\bar{\beta}_1}{\sqrt{n}},
\underbrace{0,\ldots,0}_{9}, 
-\frac{\bar{\beta}_2}{\sqrt{n}},
 \ldots, 
\underbrace{0,\ldots,0}_{9},
\frac{\bar{\beta}_{99}}{\sqrt{n}},
\underbrace{0,\ldots,0}_{9},
-\frac{\bar{\beta}_{100}}{\sqrt{n}}),
\$
where as before $\bar{\beta}$ is sampled 
from $\calN(A,I_{100})$ and fixed throughout 
the trials.
The signal amplitude $A$ ranges in $\{4,5,6,7,8,9\}$, and 
the implementation of original 
and derandomized knockoffs is the same as that 
in Section~\ref{sec:simulation}. 
Figure~\ref{fig:simulation_highdim} and~\ref{fig:simulation_highdim_prob}
present the results under the high-dimensional setting,
where we see a similar pattern: derandomized 
knockoffs exhibits comparable (or even slightly higher) power 
when the signal strength is moderately strong, while there 
is some power loss in the weak-signal regime. 
Derandomized knockoffs achieves lower FDR,
and decreased variability (especially for the conditional variability).
When we compare the marginal and conditional selection probability of 
individual hypotheses by both methods, we can also observe that the original knockoffs 
tends to have a higher selection probability for the null features that are rarely selected 
by the original knockoffs.

\begin{figure}[ht]
\centering
\rotatebox{90}{\hspace{-7ex}High-dimensional}
~
\begin{minipage}{0.3\textwidth}
\centering
\includegraphics[width = \textwidth]{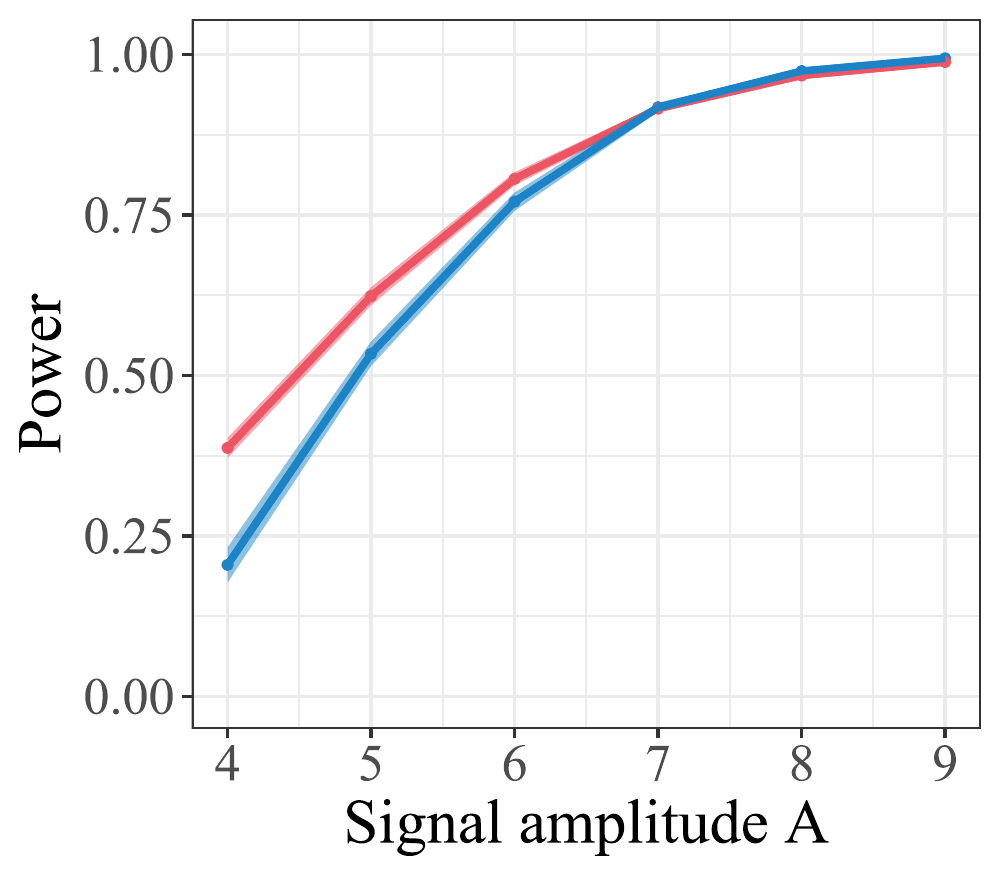}
\end{minipage}
\begin{minipage}{0.3\textwidth}
\centering
\includegraphics[width = \textwidth]{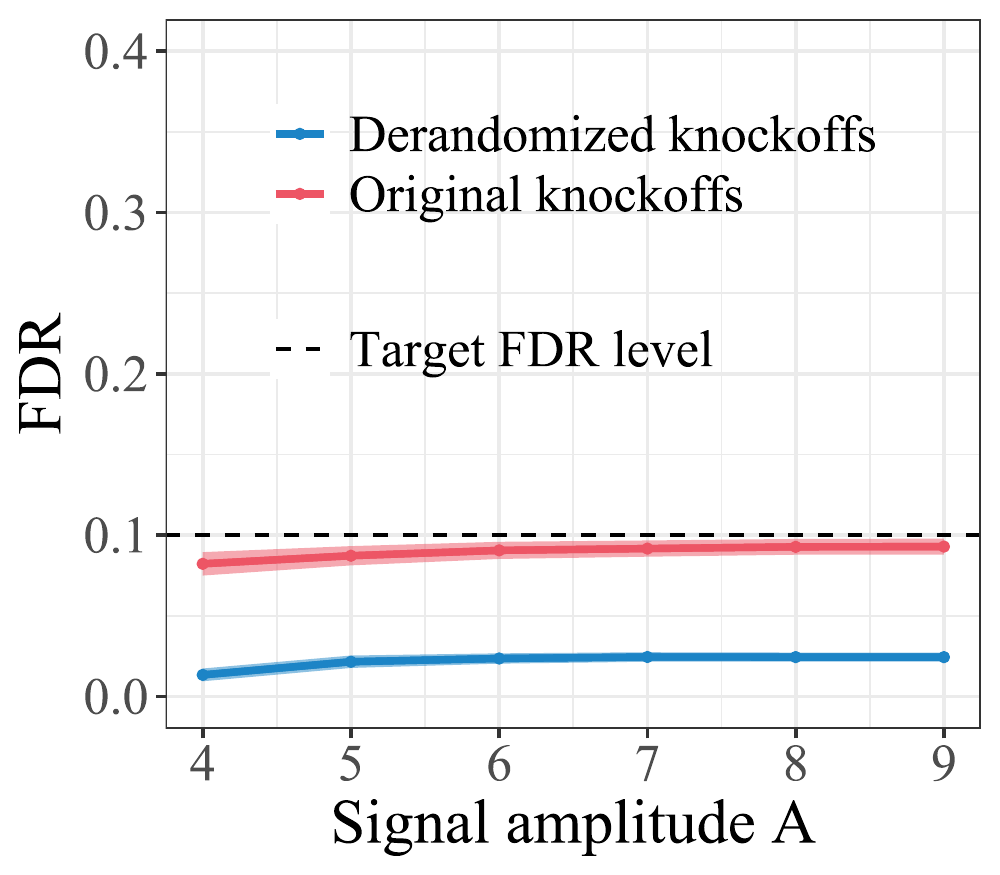}
\end{minipage}
\begin{minipage}{0.3\textwidth}
\centering
\includegraphics[width = \textwidth]{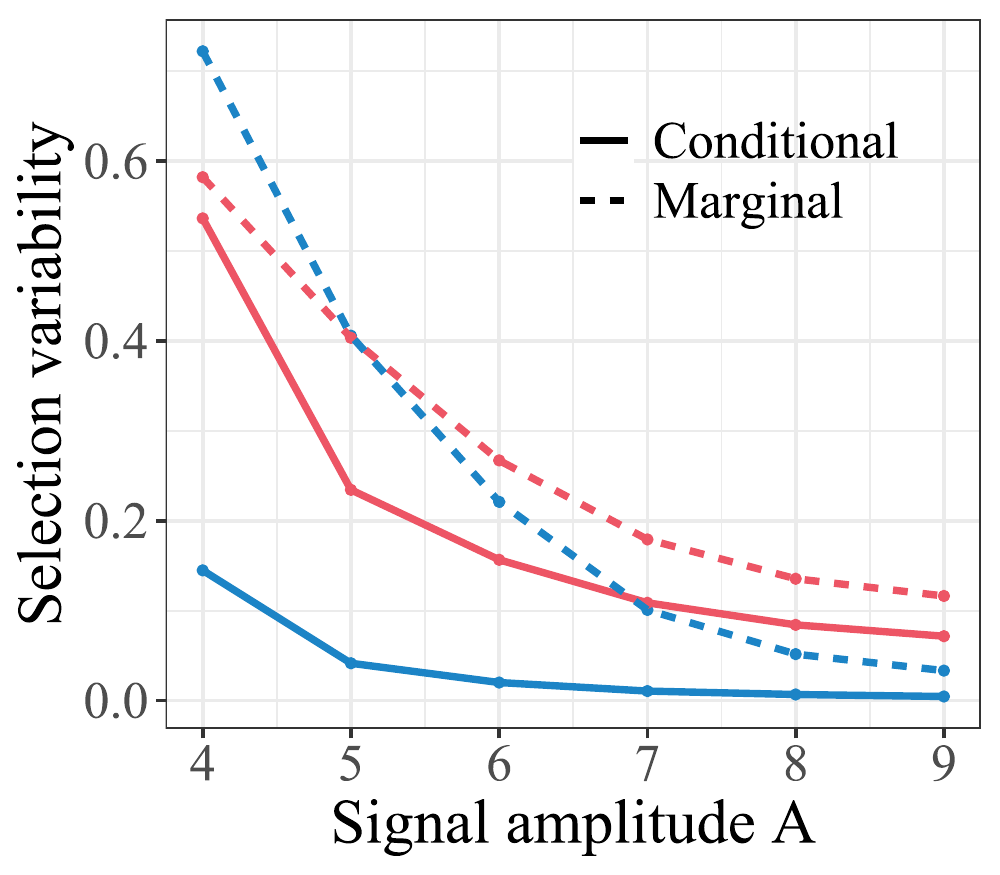}
\end{minipage}\\
\centering
\caption{Power, FDR,
and selection variability, for the high-dimensional simulated data experiment.
The other details are the same as in Figure~\ref{fig:simulation}.}
\label{fig:simulation_highdim}
\end{figure}

\begin{figure}[h] 
\centering 
\includegraphics[width=\textwidth]{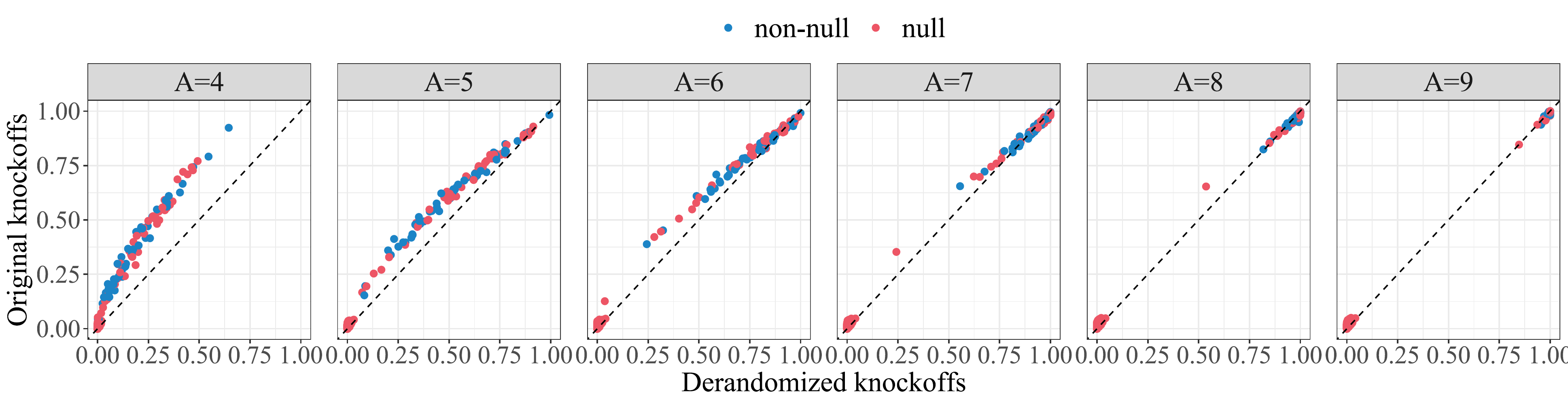}
\includegraphics[width=\textwidth]{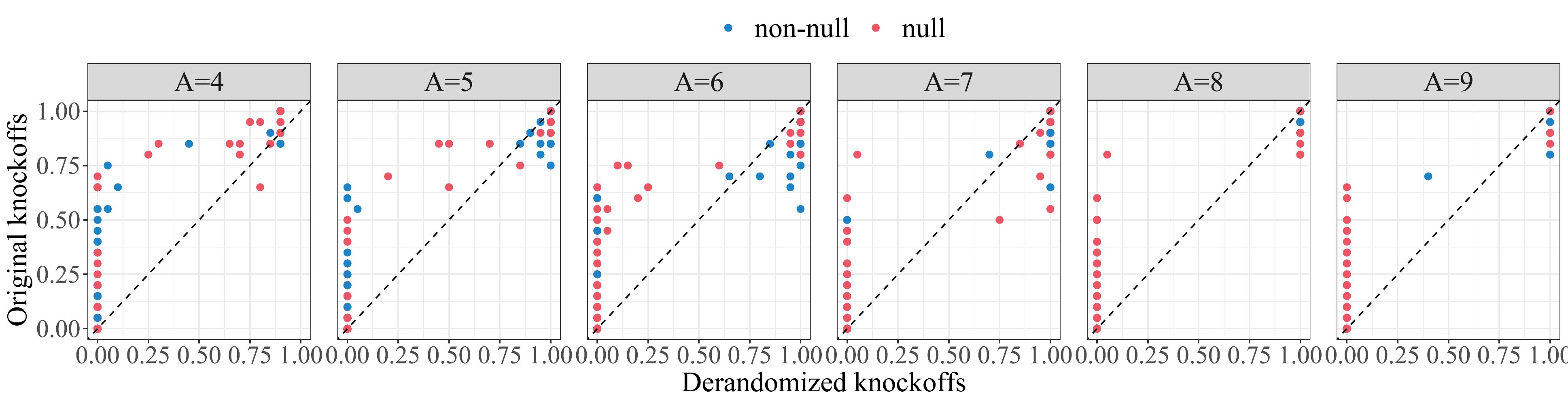}
\caption{Top: the marginal selection probability $\hat{p}_j$
by original knockoffs versus that by derandomized knockoffs.
Bottom: the conditional selection probability $\hat{p}_{j,1}$
by original knockoffs versus that by derandomized knockoffs. 
The results are from simulations under the high-dimensional setting.
The details are otherwise the same as in Figure~\ref{fig:sel_linear}.}
\label{fig:simulation_highdim_prob}
\end{figure}

\subsection{Comparison under simulation settings in~\citet{candes2018panning}}
We now compare our proposed method with the original knockoffs method
under two simulation settings considered in~\citet{candes2018panning}.
In both settings, there are $n=3000$ samples. For the low-dimensional setting, 
there are $p=1000$ features and in the high-dimensional setting, $p=6000$.
The marginal distribution of the features is $P_X = \calN(0,I_p)$ and the 
model $Y\given X$ is the same Gaussian linear model as in Section~\ref{sec:simulation},
where the coefficient vector $\beta$ contains $60$ nonzero entries whose location 
and signs are randomly determined and fixed throughout the trials. The nonzero entries
all have the same magnitude $\frac{A}{\sqrt{n}}$, where $A \in \{2,3,4,5\}$ in the low-dimensional 
setting and $A\in \{3,4,5,6\}$ in the high-dimensional one.
The implementation of original and derandomized knockoffs is the same as 
in Section~\ref{sec:simulation}, and  that of original knockoffs also 
coincides with~\citet{candes2018panning}, except that~\citet{candes2018panning}
use a selection threshold slightly different from~\eqref{eqn:define_knockoff} and 
is designed for controlling a relaxed version of FDR.

Figure~\ref{fig:original} shows the power and FDR of the two methods 
under the two settings. We again see that derandomized knockoffs exhibits 
comparable power as original knockoffs when the signals are reasonably strong,
and there is some power loss in the weak-signal regime. The realized FDR of derandomized knockoffs 
is consistently lower than that of original knockoffs.

\begin{figure}[h]
\centering
\rotatebox{90}{\hspace{-7ex}Low-dimensional}~
\begin{minipage}{0.45\textwidth}
\centering
\includegraphics[width = 0.9\textwidth]{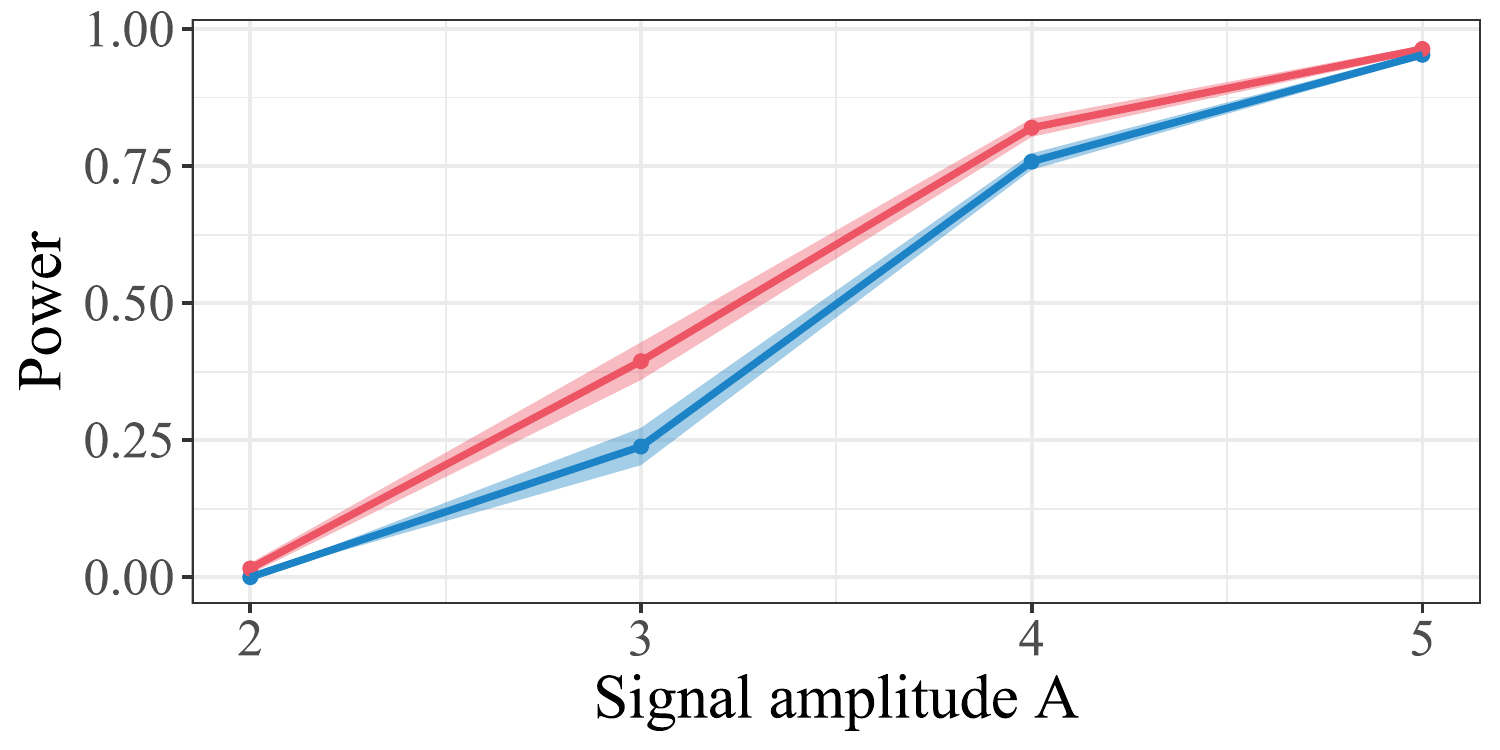}
\end{minipage}
\begin{minipage}{0.45\textwidth}
\centering
\includegraphics[width = 0.9\textwidth]{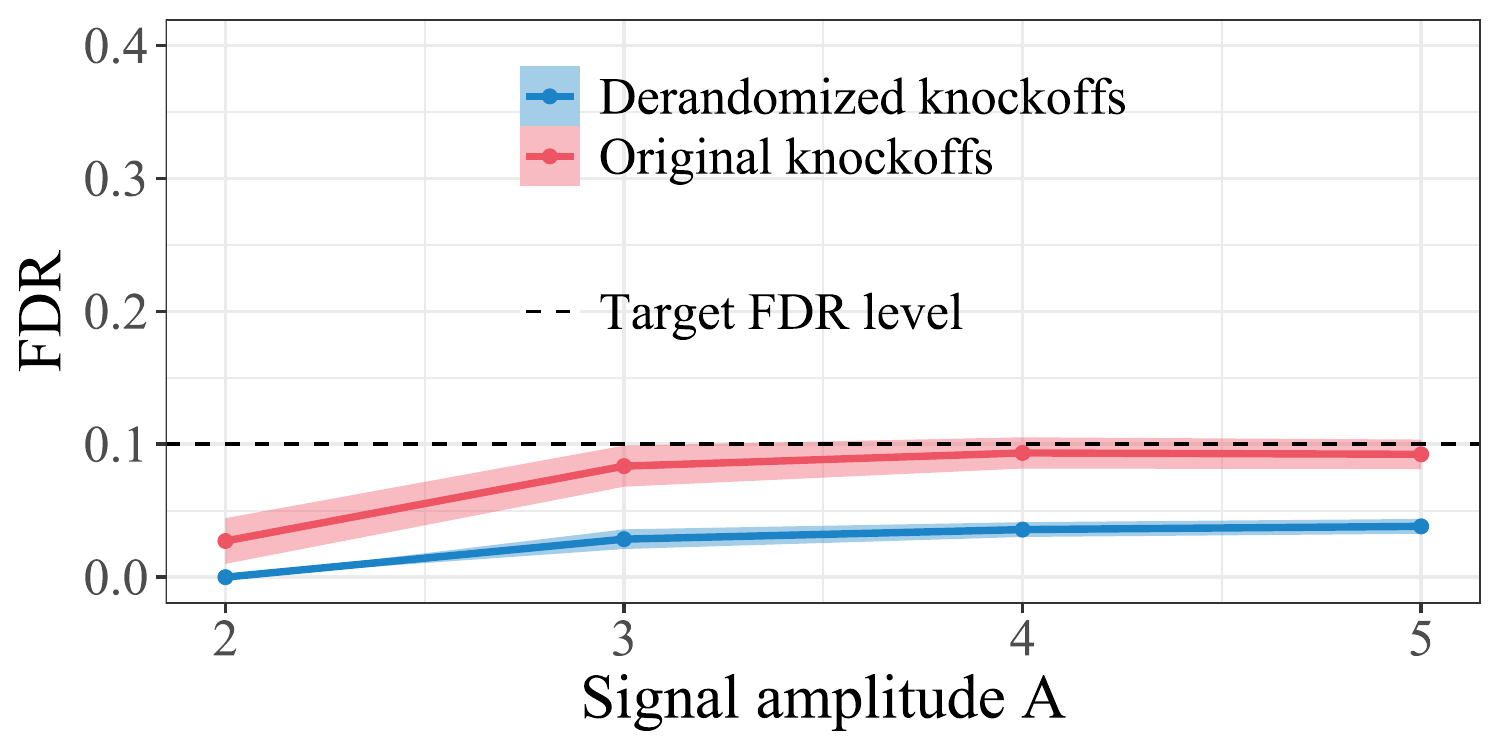}
\end{minipage}\\
\rotatebox{90}{\hspace{-7ex}High-dimensional}~
\begin{minipage}{0.45\textwidth}
\centering
\includegraphics[width = 0.9\textwidth]{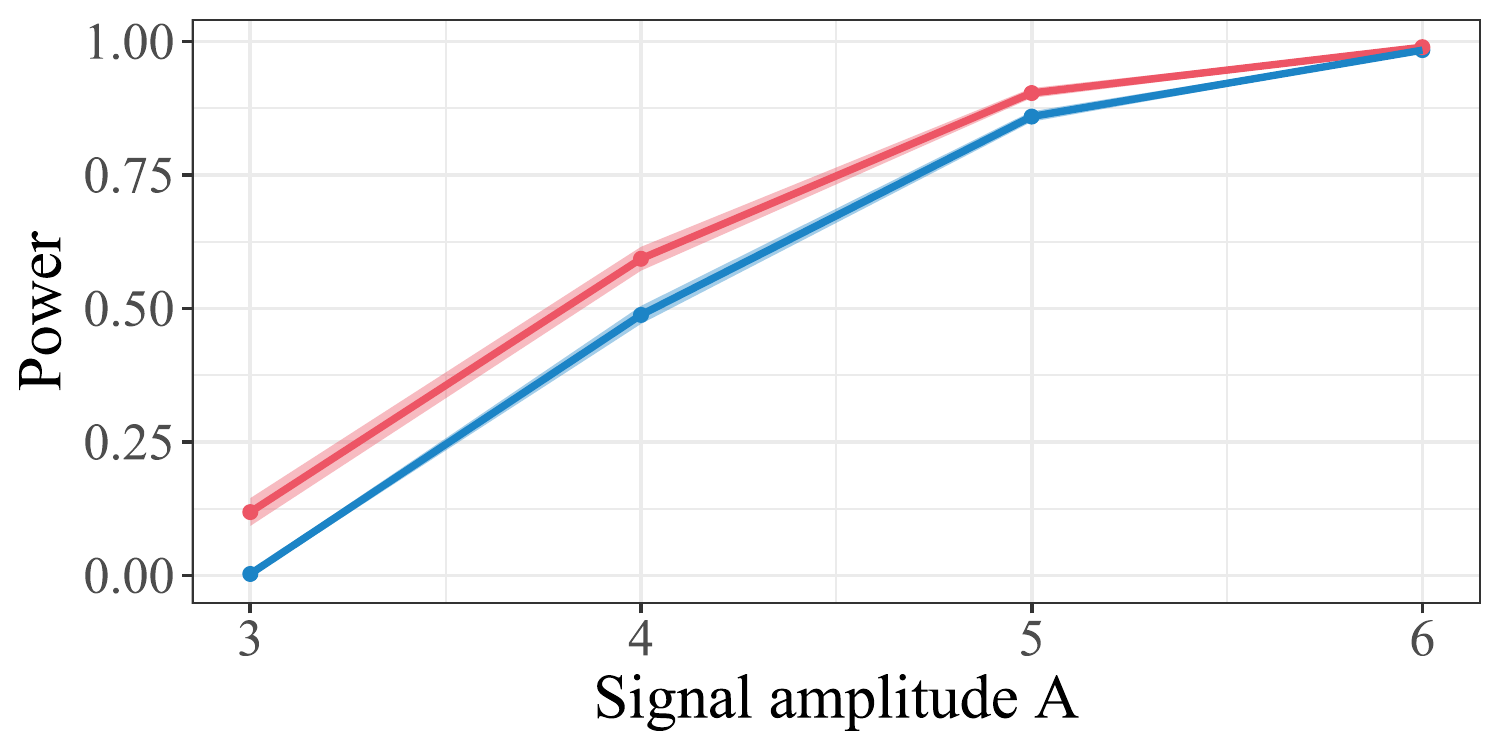}
\end{minipage}
\begin{minipage}{0.45\textwidth}
\centering
\includegraphics[width = 0.9\textwidth]{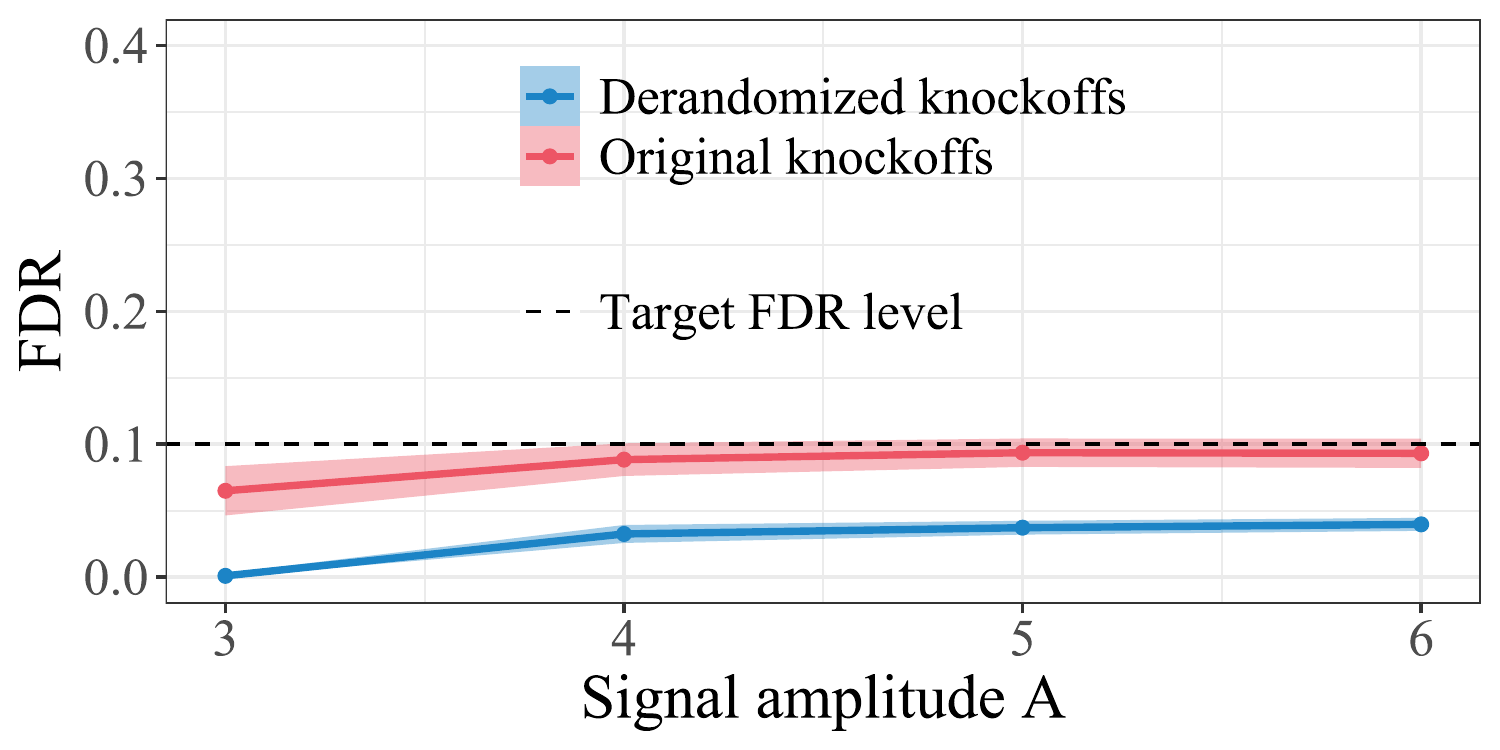}
\end{minipage}
\caption{Power and FDR under the low-dimensional simulation setting (top)
and the high-dimensional simulation setting (bottom)
considered in~\citet{candes2018panning}. Shading indicates
error bars, and the results are averaged over 100 independent 
trials.}
\label{fig:original}
\end{figure}

\subsection{Comparison with the derandomization
scheme of~\citet{ren2021derandomizing}}
As discussed earlier,~\citet{ren2021derandomizing}
propose a scheme for derandomizing the knockoffs
by thresholding the (unweighted) selection frequency. 
The algorithm proposed there is designed for the
purpose of per family-wise error (PFER) control; it
is nevertheless of interest to compare the two 
methods empirically. Hereafter, we refer to the 
method of~\citet{ren2021derandomizing} as the 
{\em PFER version}, and the method proposed in this 
paper the {\em FDR version}. We compare the two methods
under the Gaussian linear model defined in Section~\ref{sec:simulation}.
For both methods, we generate $M=50$ knockoff copies
given the dataset $(\bX,\bY)$. In the $m$-th run of
the PFER version, we generate a knockoff copy $\tilde{\bX}^{(m)}$
and compute the the feature importance statistics
$W^{(m)} = \calW([\bX,\tilde{\bX}^{(m)}],\bY)$ as
in the FDR version. We then apply the filter 
of~\citet{janson2016familywise} (designed for 
the PFER control) to $W^{(m)}$ 
with a sequence of parameters $v \in \calV \defn\{1,2,3,4\}$.
Here, the parameter $v$ suggests the target PFER level 
the filter controls (i.e.~$\EE[|\calS^{(m)} \cap \calH_0|]\le v$),
and with each $v$, we obtain a selected set 
$\hat{\calS}^{(v,m)}$. Next, we compute 
$\Pi_j^{(v)} = \frac{1}{M}\sum_{m=1}^M 
\ind\{j\in\hat{\calS}^{(v,m)}\}$ for all
$v \in \calV$ and $j\in[p]$. Finally, 
the set of discoveries corresponding to a parameter
$v$ is given by $\hat{\calS}^{(v)} = \{j: \Pi_j^{(v)} \ge 0.5\}$.
For each $v$, the aggregated selected set $\hat{\calS}^{(v)}$
has PFER controlled at level $2v$, but is not guaranteed to control
the FDR; in order to compare it with our proposed method,
we compute the realized power and the realized 
FDR of these selected sets,
plotting the power as a function of the 
FDR (i.e., the ROC curve) in Figure~\ref{fig:roc_curve}. 
We implement the derandomized knockoffs in the 
same way as in Section~\ref{sec:simulation},
where we take $\alpha_{\ebh} = 0.1$ and $\alpha_{\kn} = 0.05$.
The realized power and FDR of the derandomized
knockoffs is marked as a rectangle in Figure~\ref{fig:roc_curve}.
We can see from the figures that at the same level of realized
FDR, the FDR version (i.e., the derandomized procedure proposed in our present work) achieves comparable power
with the previous version of \cite{ren2021derandomizing}.

\begin{figure}[ht]
\centering
\includegraphics[width = \textwidth]{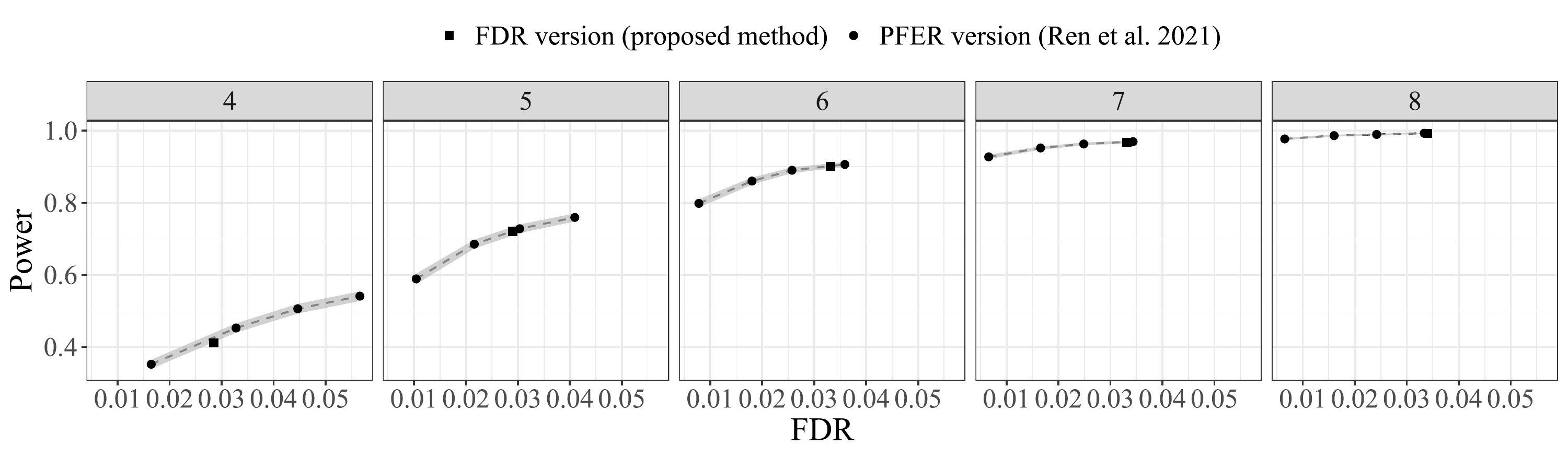}
\caption{The dashed line represents the realized power of 
the PFER version as a function of the realized FDR; the rectangle 
represents the realized power and FDR resulting from the FDR version
(with the target FDR level $\alpha_{\ebh} = 0.1$).}
\label{fig:roc_curve}
\end{figure}
\subsection{Robustness to the estimation error of $P_X$}\label{sec:sim_robust}
We empirically evaluate the robustness 
of the derandomized knockoffs procedure.
In this experiment, $n=600$ and $p=100$; the covariate 
distribution $P_X=\calN(0,\Sigma)$, where $\Sigma_{ij} = 0.5^{|i-j|}$.
$Y\given X$ follows a Gaussian linear model, with the coefficient vector being
\$ 
\beta = (\tfrac{A}{\sqrt{n}},0,-\tfrac{A}{\sqrt{n}},0,\ldots),
\$
and $A \in \{4.5,5,5.5,6,6.5\}$.
The implementation of derandomized knockoffs is the same as in 
Section~\ref{sec:simulation}, except that
now we assume no knowledge of 
covariance matrix $\Sigma$; instead, to estimate $\Sigma$, we are given an additional $n_0$ many
unlabeled samples (i.e., we observe $X$ but do not observe $Y$).
We construct knockoffs with $Q_X = \calN(0,\hat{\Sigma})$,
where $\hat{\Sigma}$ is the sample covariance
matrix computed with the $n_0$ unlabeled 
samples; we then apply the original knockoffs
and derandomized knockoffs, where
the derandomized version is implemented according to 
Algorithm~\ref{alg:aggregate_knockoff}
with the same parameters used in Section~\ref{sec:simulation}.
Figure~\ref{fig:robustness} describes
the realized power and FDR achieved by the original
knockoffs and the proposed method under the 
Gaussian linear model setting. In particular, the covariance
matrix is estimated with $n_0\in\{600, 700, 800, 900, 1000\}$.
Here, we observe the FDR to be controlled at the desired level for both methods,
with lower FDR for derandomized knockoffs as before. Derandomized
knockoffs shows power that is comparable to the original method across most settings,
although the power
is lower in one setting (the smallest value of $n_0$, where the estimated distribution of $X$ is most
unreliable).
\begin{figure}[ht]
\centering
\includegraphics[width = \textwidth]{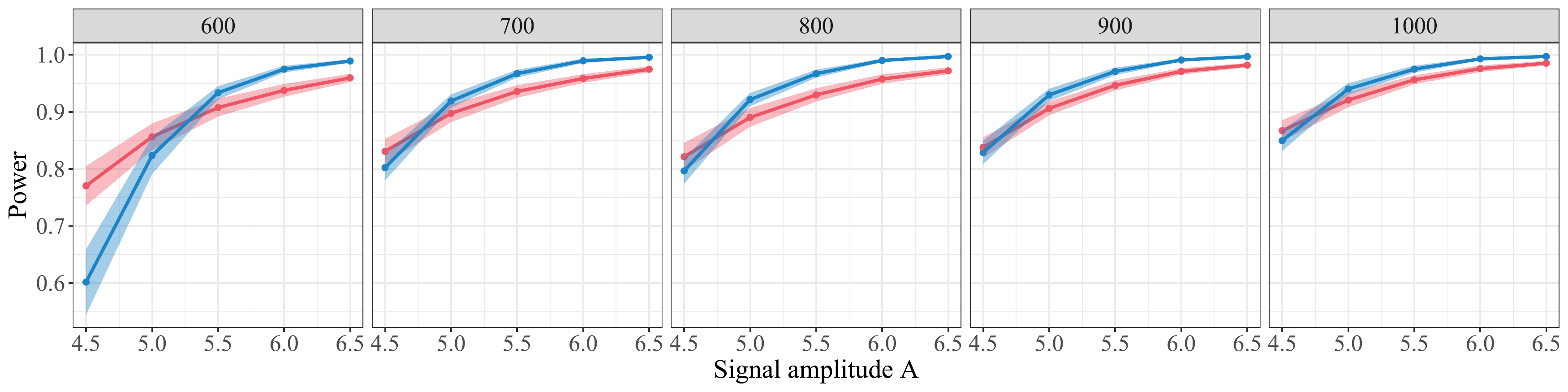}
\includegraphics[width = \textwidth]{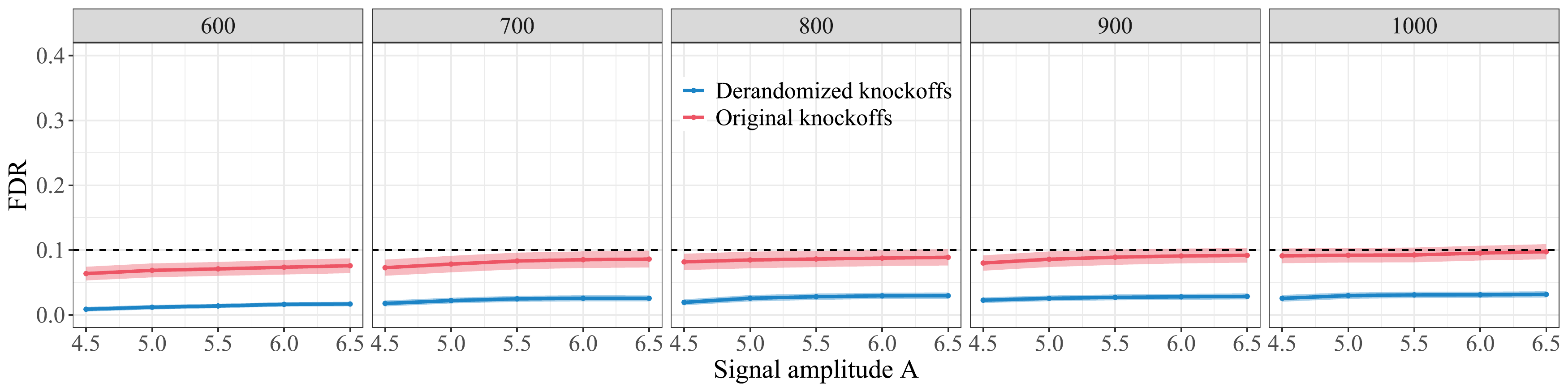}
\caption{Power (top) and FDR (bottom), for the Gaussian linear simulation
data experiments with estimated $P_X$. Each panel corresponds to a size
of unlabeled data for estimating $P_X$.}
\label{fig:robustness}
\end{figure}

\subsection{A multiple-environment setting}
We evaluate the performance of the algorithm
proposed in Section~\ref{sec:multi} in a 
multi-environment simulation setting. Suppose
$p=100$, and there are two environments. In
environment $e \in[2]$, we assume 
$X^e \sim \calN(0,\Sigma)$, where $\Sigma_{jk} = 0.5^{|j-k|}$
for all $j,k\in[p]$;
the response follows from a environment-specific
linear model:
$Y^e \sim \calN((X^e)^\top \beta^{(e)},\Sigma)$,
where $\beta^{(e)}$ is the coefficient vector. We
construct the coefficient vectors as:
\$\beta^{(1)} = (\underbrace{\tfrac{A}{\sqrt{n}},
\tfrac{A}{\sqrt{n}},\ldots,\tfrac{A}{\sqrt{n}}}_{50},0,\ldots,0),\quad
\beta^{(2)} = (\underbrace{\tfrac{A}{\sqrt{n}},\tfrac{A}{\sqrt{n}},\ldots,
\tfrac{A}{\sqrt{n}}}_{60},0,\ldots,0),
\$
and the signal amplitude $A$ ranges in $\{3,3.5,4,4.5,5,5.5\}$.
Here, $\{1,2,\ldots,50\}$ is the set of nonnulls, and 
the FDR target is $0.1$. We implement MEKF with the code
from~\url{https://github.com/lsn235711/MEKF_code}, where 
the empirical prior is used when constructing the multi-environment
feature importance statistics. The derandomized MEKF is
implemented with the same details, and we take $\alpha_{\ebh}
=0.1$ and $\alpha_{\kn} = 0.05$.
Figure~\ref{fig:multi} shows the power and FDR resulting 
from MEKF and the derandomized MEKF in this simulated 
setting---the derandomized version shows higher
power and lower FDR.

\begin{figure}[ht]
\centering
\begin{minipage}{0.35\textwidth}
\centering
\includegraphics[width = \textwidth]{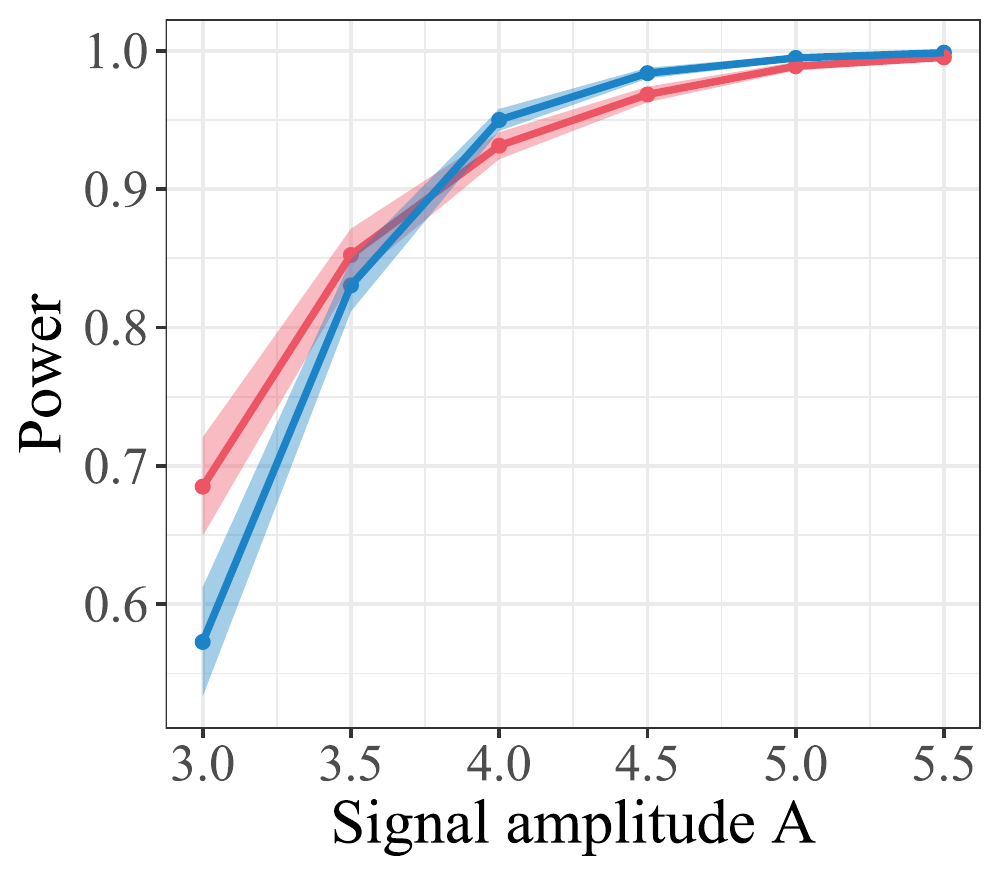}
\end{minipage}
\begin{minipage}{0.35\textwidth}
\centering
\includegraphics[width = \textwidth]{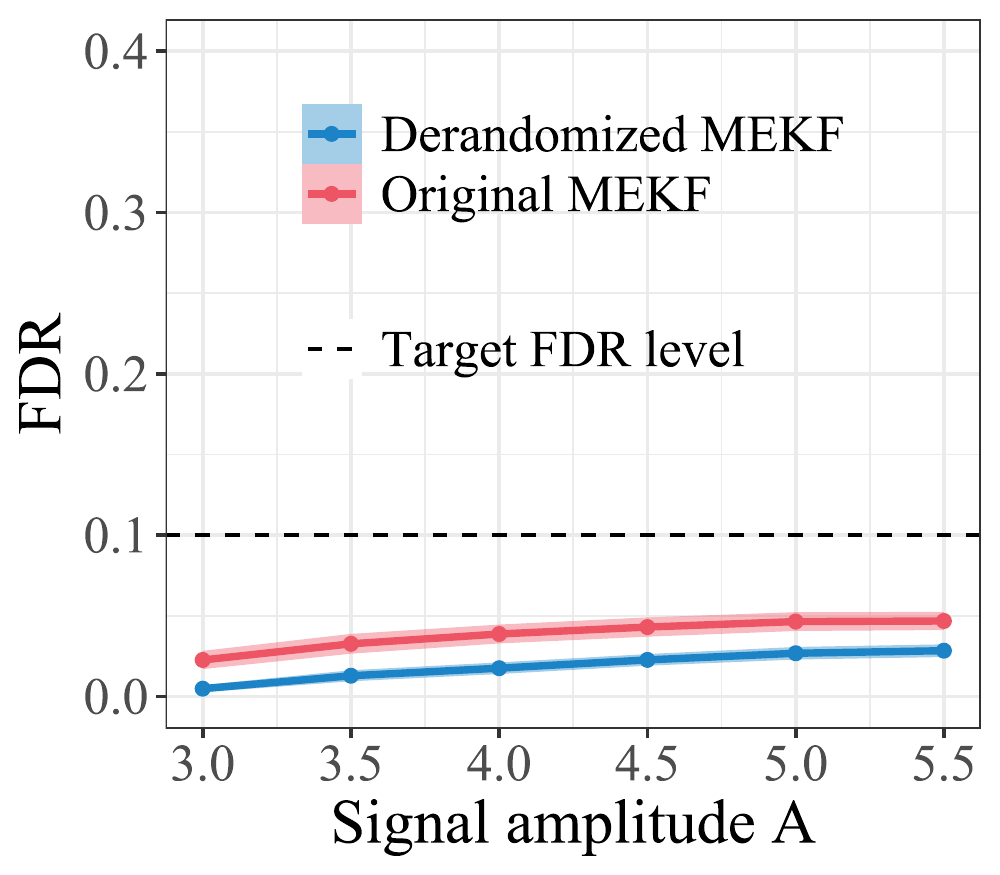}
\end{minipage}
\caption{Power and FDR, for the multi-environment simulated setting.
The other details are the same as in Figure~\ref{fig:simulation}.}
\label{fig:multi}
\end{figure}

\subsection{Side information}
We set up an experiment to illustrate how side
information can help improve the power of 
derandomized knockoffs.
The simulation setting is the same as the 
Gaussian linear model in Section~\ref{sec:sim_robust}
with a few changes. Here, we assume known $\Sigma$, 
and $A \in \{3,3.5,4,4.5,5,5.5\}$; the nonnulls are placed on 
$\{1,2,\ldots,50\}$ while the remaining features $\{51,\dots,100\}$ are null
(the index hence carries information, with lower values of $j$ corresponding to signals).
We apply the weighted version of derandomized knockoffs
with side information, using $u_j = \exp(-j)$.
The other implementation details are exactly the same
as in Section~\ref{sec:simulation}.
Figure~\ref{fig:sideinfo} plots the power and FDR
from the original knockoffs, derandomized knockoffs,
and the weighted version of derandomized knockoffs 
with side information. Clearly, the weighted version
achieves a significantly higher power by incorporating
side information, while still controlling FDR.

\begin{figure}[ht]
\centering
\begin{minipage}{0.35\textwidth}
\centering
\includegraphics[width = \textwidth]{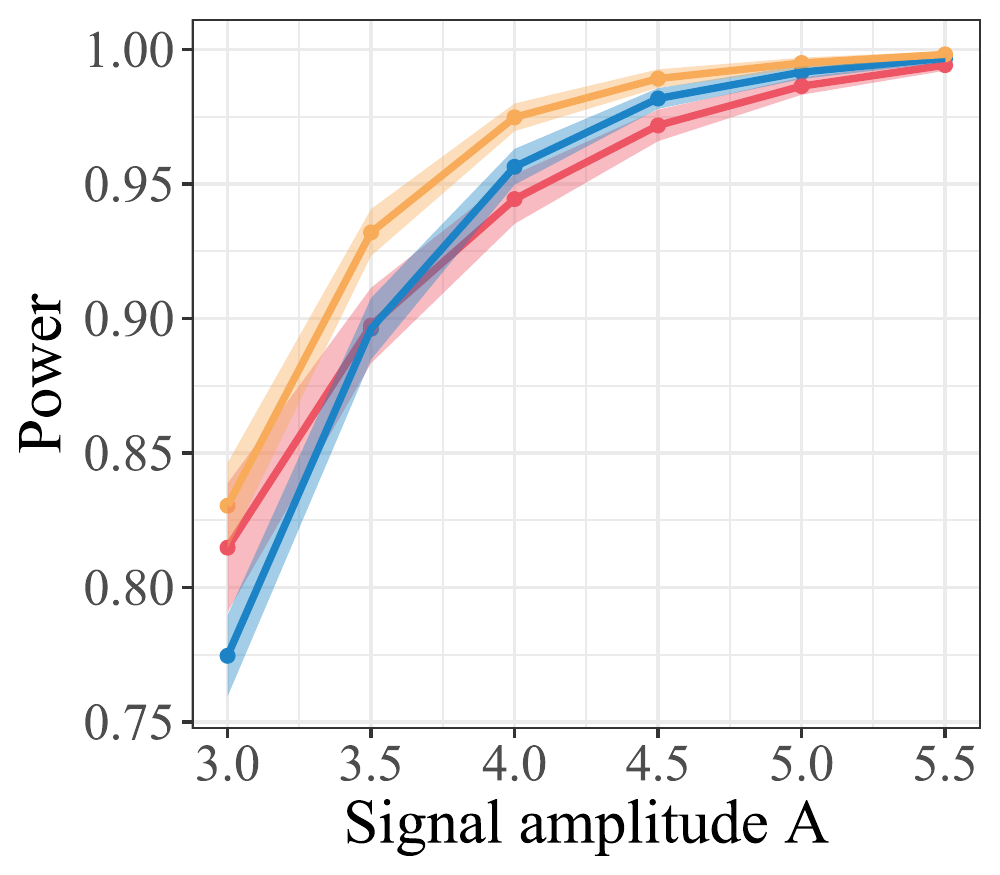}
\end{minipage}
\begin{minipage}{0.35\textwidth}
\centering
\includegraphics[width = \textwidth]{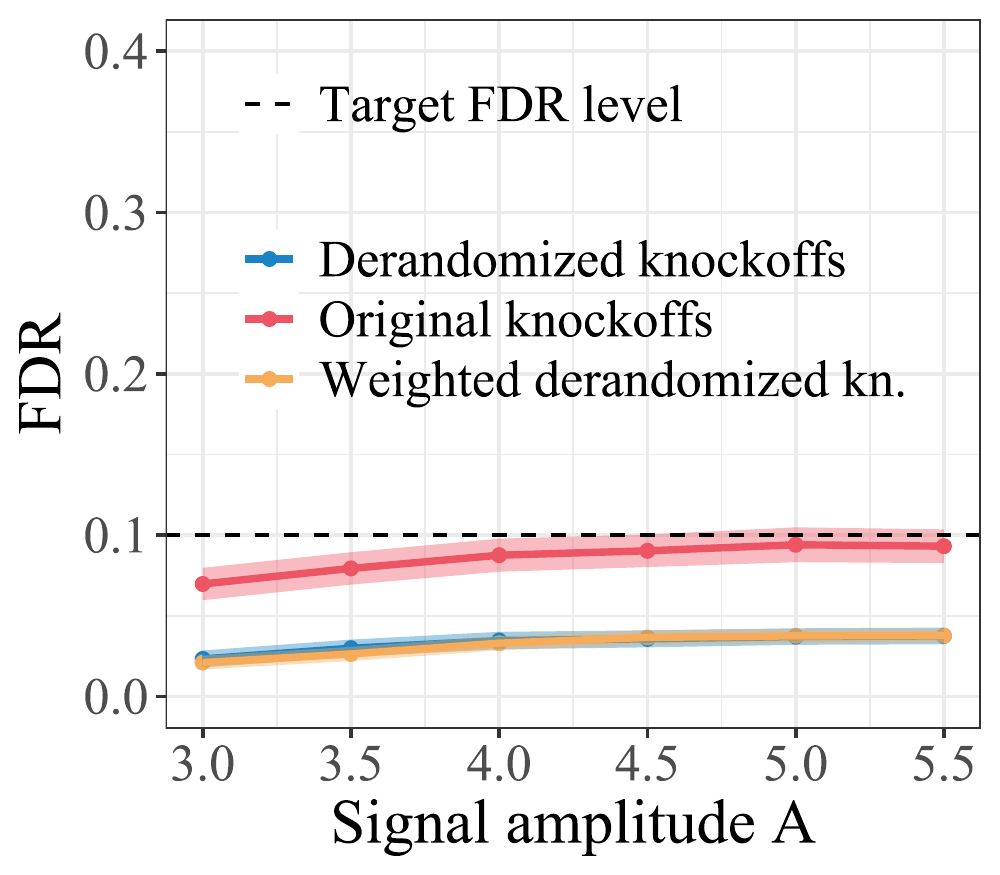}
\end{minipage}
\caption{Power and FDR, for the simulated experiments with side information.
The other details are the same as in Figure~\ref{fig:simulation}.}
\label{fig:sideinfo}
\end{figure}

\section{Additional results from the analysis of the HIV dataset}
\label{sec:discovery_list}

For the HIV dataset, over 50 independent runs of both methods,
the average number of discoveries made by the derandomized knockoffs
procedure  is $63.5$, and the average number of discoveries made by
the original knockoffs is $77.0$. As we have observed in the 
simulations, the derandomized knockoffs procedure often shows
a lower FDR while achieving comparable power with the 
original knockoffs---in other words, the original knockoffs would often
make more (false) discoveries to achieve the same detecting
power as the derandomized knockoffs. To examine whether this may be the case
in our real data example as well, we compare the discoveries
made by both methods with the reference found at
\url{https://hivdb.stanford.edu/dr-summary/comments/PI/}.
To be specific, for each method we take the set of mutations 
selected for more than half of the times; this results in
a set of 63 discoveries for the derandomized knockoffs, and
a set of 74 discoveries for the original version (which includes
all 63 of the discoveries made by derandomized knockoffs).

Table~\ref{tab:list_mut} contains a complete 
list of the 63 mutations selected in at least 50\% of the runs of the derandomized
knockoffs procedure, and Table~\ref{tab:list_mut_orig} shows the list
of the 74 mutations selected
 in at least 50\% of the runs of the original
knockoffs procedure.
In each table, the column ``Annotation'' contains 
the comments on the corresponding mutations from the reference,
where ``Major'', ``Accessory'' and ``Other'' refer to 
major effect, accessory effect and other effect on the drug 
resistance to PI (i.e., the discovery is confirmed); ``NA'' means no comments are found from the 
reference. We say a mutation is ``verified'' by the reference
if the annotation is not ``NA''. For derandomized knockoffs, Out of the 63 discoveries,
38 mutations are verified. For original knockoffs, out of the 74 discoveries, 39 are verified.
We thus see approximately the same number of
verified mutations made by the two methods, even though the original knockoffs method
makes substantially more discoveries overall---this can be viewed as
evidence that the derandomized knockoffs procedure maintains comparable detection
power while decreasing FDR and decreasing variability.

\newcommand{\ac}{Accessory}
\newcommand{\ot}{Other}
\begin{table}
\caption{
\label{tab:list_mut}
The complete list of mutations discovered
by derandomized knockoffs at FDR level $0.1$. 
The annotation is the reported effect 
from~\url{https://hivdb.stanford.edu/dr-summary/comments/PI/}.}
\centering
\fbox{\begin{tabular}{cc}
  Mutation & Annotation\\
\hline
6W & NA\\
10F & \ac\\
10I & \ot\\
10L & NA\\
10V & \ot\\
12P & NA\\
14R & NA\\
16A & NA\\
20I & \ot\\
20R & \ot\\
20T & \ac\\
20V & \ot\\
22V & NA\\
24I & \ac\\
32I & Major\\
33F & \ac\\
33L & NA\\
34Q & NA\\
36I & NA\\
36L & NA\\
36M & NA\\
37Q & NA\\
43T & \ac\\
46I & Major\\
46L & Major\\
47V & Major\\
48M & Major\\
48V & Major\\
50L & Major\\
50V & Major\\
53L & \ac\\
54L & \ac\\
\end{tabular}
\hspace{2ex}
\begin{tabular}{cc}
Mutation & Annotation\\
\hline
54M & Major\\
54S & Major\\
54T & Major\\
54V & Major\\
58E & \ac\\
61D & NA\\
62V & NA\\
63P & NA\\
64L & NA\\
69K & NA\\
71A & NA\\
71T & \ot\\
71V & \ot\\
72L & NA\\
73T & \ac\\
76V & Major\\
77I & NA\\
82A & Major\\
82F & Major\\
82S & Major\\
82T & Major\\
84A & Major\\
84V & Major\\
88D & \ac\\
89I & NA\\
89V & \ac\\
90M & Major\\
92K & NA\\
92Q & NA\\
93I & NA\\
95F & NA\\
&\\ 
\end{tabular}
}
\vspace{1em}
\end{table}

\begin{table}
\caption{
\label{tab:list_mut_orig}
The complete list of mutations selected for 
more than $50\%$ times (out of $50$ independent runs) by the
original knockoffs procedure
at FDR level $0.1$. The annotation
is the reported effect from~\url{https://hivdb.stanford.edu/dr-summary/comments/PI/}.}
\centering
\fbox{%
\begin{tabular}{cc}
Mutation & Annotation\\
\hline
6W & NA\\
10F & \ac\\
10I & \ot\\
10L & NA\\
10V & \ot\\
12P & NA\\
14R & NA\\
16A & NA\\
20I & \ot\\
20K & NA\\
20R & \ot\\
20T & \ac\\
20V & \ot\\
22V & NA\\
24I & \ac\\
32I & Major\\
33F & \ac\\
33L & NA\\
34Q & NA\\
35D & NA\\
36I & NA\\
36L & NA\\
36M & NA\\
37Q & NA\\
39Q & NA\\
43T & \ac\\
46I & Major\\
46L & Major\\
47V & Major\\
48M & Major\\
48V & Major\\
50L & Major\\
50V & Major\\
53L & \ac\\
54L & Major\\
54M & Major\\
54S & Major\\
\end{tabular}
\hspace{2ex}
\begin{tabular}{cc}
  Mutation & Annotation\\
  \hline
54T & Major\\
54V & Major\\
58E & \ac\\
61D & NA\\
61E & NA\\
62V & NA\\
63P & NA\\
64L & NA\\
67C & NA\\
67E & NA\\
67Y & NA\\
69K & NA\\
71A & NA\\
71I & \ot\\
71T & \ot\\
71V & \ot\\
72L & NA\\
72M & NA\\
73T & \ac\\
76V & Major\\
77I & NA\\
82A & Major\\
82F & Major\\
82S & Major\\
82T & Major\\
84A & Major\\
84V & Major\\
88D & \ac\\
89I & NA\\
89V & \ac\\
90M & Major\\
91S & NA\\
92K & NA\\
92Q & NA\\
93I & NA\\
93L & NA\\
95F & NA\\
\end{tabular}
}
\vspace{1em}
\end{table}

\clearpage
\bibliographystyle{rss}
\bibliography{ref}

\begin{thebibliography}{29}
\expandafter\ifx\csname natexlab\endcsname\relax\def\natexlab#1{#1}\fi
\expandafter\ifx\csname url\endcsname\relax
  \def\url#1{\texttt{#1}}\fi
\expandafter\ifx\csname urlprefix\endcsname\relax\def\urlprefix{URL: }\fi

\bibitem[{Barber and Cand{\`e}s(2015)}]{barber2015controlling}
Barber, R.~F. and Cand{\`e}s, E.~J. (2015) Controlling the false discovery rate via knockoffs.
\newblock \textit{The Annals of Statistics}, \textbf{43}, 2055--2085.

\bibitem[{Barber et~al.(2020)Barber, Cand{\`e}s and Samworth}]{barber2020robust}
Barber, R.~F., Cand{\`e}s, E.~J. and Samworth, R.~J. (2020) Robust inference with knockoffs.
\newblock \textit{The Annals of Statistics}, \textbf{48}, 1409--1431.

\bibitem[{Bates et~al.(2021)Bates, Cand{\`e}s, Janson and Wang}]{bates2021metropolized}
Bates, S., Cand{\`e}s, E., Janson, L. and Wang, W. (2021) Metropolized knockoff sampling.
\newblock \textit{Journal of the American Statistical Association}, \textbf{116}, 1413--1427.

\bibitem[{Benjamini and Hochberg(1995)}]{benjamini1995controlling}
Benjamini, Y. and Hochberg, Y. (1995) Controlling the false discovery rate: a practical and powerful approach to multiple testing.
\newblock \textit{Journal of the Royal Statistical Society: Series B (Methodological)}, \textbf{57}, 289--300.

\bibitem[{Cand{\`e}s et~al.(2018)Cand{\`e}s, Fan, Janson and Lv}]{candes2018panning}
Cand{\`e}s, E., Fan, Y., Janson, L. and Lv, J. (2018) Panning for gold: ‘model-x’ knockoffs for high dimensional controlled variable selection.
\newblock \textit{Journal of the Royal Statistical Society: Series B (Statistical Methodology)}, \textbf{80}, 551--577.

\bibitem[{Dai et~al.(2022)Dai, Lin, Xing and Liu}]{dai2022false}
Dai, C., Lin, B., Xing, X. and Liu, J.~S. (2022) False discovery rate control via data splitting.
\newblock \textit{Journal of the American Statistical Association}, 1--18.

\bibitem[{Dai et~al.(2023)Dai, Lin, Xing and Liu}]{dai2023scale}
--- (2023) A scale-free approach for false discovery rate control in generalized linear models.
\newblock \textit{Journal of the American Statistical Association}, 1--31.

\bibitem[{Dai and Barber(2016)}]{dai2016knockoff}
Dai, R. and Barber, R. (2016) The knockoff filter for fdr control in group-sparse and multitask regression.
\newblock In \textit{International conference on machine learning}, 1851--1859. PMLR.

\bibitem[{Emery and Keich(2019)}]{emery2019controlling}
Emery, K. and Keich, U. (2019) Controlling the fdr in variable selection via multiple knockoffs.
\newblock \textit{arXiv preprint arXiv:1911.09442}.

\bibitem[{Gimenez and Zou(2019)}]{gimenez2019improving}
Gimenez, J.~R. and Zou, J. (2019) Improving the stability of the knockoff procedure: Multiple simultaneous knockoffs and entropy maximization.
\newblock In \textit{The 22nd International Conference on Artificial Intelligence and Statistics}, 2184--2192. PMLR.

\bibitem[{Janson and Su(2016)}]{janson2016familywise}
Janson, L. and Su, W. (2016) Familywise error rate control via knockoffs.
\newblock \textit{Electronic Journal of Statistics}, \textbf{10}, 960--975.

\bibitem[{Koyuncu and Yener(2022)}]{koyuncu2022missing}
Koyuncu, D. and Yener, B. (2022) Missing value knockoffs.
\newblock \textit{arXiv preprint arXiv:2202.13054}.

\bibitem[{Li et~al.(2021)Li, Sesia, Romano, Cand{\`e}s and Sabatti}]{li2021searching}
Li, S., Sesia, M., Romano, Y., Cand{\`e}s, E. and Sabatti, C. (2021) Searching for robust associations with a multi-environment knockoff filter.
\newblock \textit{Biometrika}.

\bibitem[{Liu et~al.(2010)Liu, Roeder and Wasserman}]{liu2010stability}
Liu, H., Roeder, K. and Wasserman, L. (2010) Stability approach to regularization selection (stars) for high dimensional graphical models.
\newblock \textit{Advances in neural information processing systems}, \textbf{23}.

\bibitem[{Luo et~al.(2022)Luo, Fithian and Lei}]{luo2022improving}
Luo, Y., Fithian, W. and Lei, L. (2022) Improving knockoffs with conditional calibration.
\newblock \textit{arXiv preprint arXiv:2208.09542}.

\bibitem[{Meinshausen and B{\"u}hlmann(2010)}]{meinshausen2010stability}
Meinshausen, N. and B{\"u}hlmann, P. (2010) Stability selection.
\newblock \textit{Journal of the Royal Statistical Society: Series B (Statistical Methodology)}, \textbf{72}, 417--473.

\bibitem[{Nguyen et~al.(2020)Nguyen, Chevalier, Thirion and Arlot}]{nguyen2020aggregation}
Nguyen, T.-B., Chevalier, J.-A., Thirion, B. and Arlot, S. (2020) Aggregation of multiple knockoffs.
\newblock In \textit{International Conference on Machine Learning}, 7283--7293. PMLR.

\bibitem[{Patterson and Sesia(2018)}]{patterson2018knockoff}
Patterson, E. and Sesia, M. (2018) knockoff: The knockoff filter for controlled variable selection.
\newblock \textit{R package version 0.3}, \textbf{2}.

\bibitem[{Ren and Cand{\`e}s(2020)}]{ren2020knockoffs}
Ren, Z. and Cand{\`e}s, E. (2020) Knockoffs with side information.
\newblock \textit{arXiv preprint arXiv:2001.07835}.

\bibitem[{Ren et~al.(2021)Ren, Wei and Cand{\`e}s}]{ren2021derandomizing}
Ren, Z., Wei, Y. and Cand{\`e}s, E. (2021) Derandomizing knockoffs.
\newblock \textit{Journal of the American Statistical Association}, 1--11.

\bibitem[{Rhee et~al.(2006)Rhee, Taylor, Wadhera, Ben-Hur, Brutlag and Shafer}]{rhee2006genotypic}
Rhee, S.-Y., Taylor, J., Wadhera, G., Ben-Hur, A., Brutlag, D.~L. and Shafer, R.~W. (2006) Genotypic predictors of human immunodeficiency virus type 1 drug resistance.
\newblock \textit{Proceedings of the National Academy of Sciences}, \textbf{103}, 17355--17360.

\bibitem[{Romano et~al.(2020)Romano, Sesia and Cand{\`e}s}]{romano2020deep}
Romano, Y., Sesia, M. and Cand{\`e}s, E. (2020) Deep knockoffs.
\newblock \textit{Journal of the American Statistical Association}, \textbf{115}, 1861--1872.

\bibitem[{Sesia et~al.(2019)Sesia, Sabatti and Cand{\`e}s}]{sesia2019gene}
Sesia, M., Sabatti, C. and Cand{\`e}s, E.~J. (2019) Gene hunting with hidden markov model knockoffs.
\newblock \textit{Biometrika}, \textbf{106}, 1--18.

\bibitem[{Shah and Samworth(2013)}]{shah2013variable}
Shah, R.~D. and Samworth, R.~J. (2013) Variable selection with error control: another look at stability selection.
\newblock \textit{Journal of the Royal Statistical Society: Series B (Statistical Methodology)}, \textbf{75}, 55--80.

\bibitem[{Spector and Janson(2022)}]{spector2022powerful}
Spector, A. and Janson, L. (2022) Powerful knockoffs via minimizing reconstructability.
\newblock \textit{The Annals of Statistics}, \textbf{50}, 252--276.

\bibitem[{Vovk(2020)}]{vovk2020note}
Vovk, V. (2020) A note on data splitting with e-values: online appendix to my comment on glenn shafer's" testing by betting".
\newblock \textit{arXiv preprint arXiv:2008.11474}.

\bibitem[{Vovk and Wang(2021)}]{vovk2021values}
Vovk, V. and Wang, R. (2021) E-values: Calibration, combination and applications.
\newblock \textit{The Annals of Statistics}, \textbf{49}, 1736--1754.

\bibitem[{Wang and Ramdas(2022)}]{wang2022false}
Wang, R. and Ramdas, A. (2022) False discovery rate control with e-values.
\newblock \textit{Journal of the Royal Statistical Society: Series B (Statistical Methodology)}.
\newblock \urlprefix\url{https://rss.onlinelibrary.wiley.com/doi/abs/10.1111/rssb.12489}.

\bibitem[{Wasserman et~al.(2020)Wasserman, Ramdas and Balakrishnan}]{wasserman2020universal}
Wasserman, L., Ramdas, A. and Balakrishnan, S. (2020) Universal inference.
\newblock \textit{Proceedings of the National Academy of Sciences}, \textbf{117}, 16880--16890.

\end{thebibliography}

\section*{List of figure legends}
\begin{itemize}
\item Figure~\ref{fig:alpha}:  Realized power (left) and FDR (right) of derandomized 
knockoffs as a function of the parameter $\alpha_{\kn}$ for the simulation 
data experiments. The offset parameter $c=1$. 
Shading for the power and FDR plots indicates error bars.
The target FDR level $\alpha_{\ebh} = 0.1$.
Results are averaged over $100$ independent trials.
\item Figure~\ref{fig:simulation}: Power, FDR,
and selection variability, for the simulated data experiments.
Shading for the power and FDR plots indicates error bars.
Results are averaged over $100$ independent trials. 
\item Figure~\ref{fig:sel_linear}: Top: the marginal selection probability $\hat{p}_j$
by original knockoffs versus that by derandomized knockoffs.
Bottom: the conditional selection probability $\hat{p}_{j,1}$
by original knockoffs versus that by derandomized knockoffs. 
The results are from simulations under the Gaussian linear model.
Each point corresponds to a feature, where the blue ones are 
non-nulls and the red ones are nulls. For the given dataset, 
many null features never selected by derandomized knockoffs have
large selection probability by the original knockoffs.
\item Figure~\ref{fig:sel_logistic}: The results are from simulations under the logistic model.
The other details are the same as in Figure~\ref{fig:sel_logistic}. 
\item Figure~\ref{fig:realdata}: Results for the HIV data experiment. Left: boxplot of the number of discoveries for original and derandomized knockoffs,
over 50 independent trials. Right: histogram of the selection probability for each feature $j\in[p]$, over 50 independent
trials, for original and derandomized knockoffs.
\item Figure~\ref{fig:heatmap}: Heatmaps of the power of derandomized knockoffs with different choices 
of $(c,\alpha_{\kn})$. The left, middle and right columns correspond to low, medium 
and high signal amplitude, respectively. A darker color represents a higher value. 
In these experiments, we set $\alpha_{\ebh} = 0.1$, and we observe that power is consistently 
high around $\alpha_{\kn}=0.05$, justifying our intuition that $\alpha_{\kn} = \alpha_{\ebh}/2$
is a good default setting.
\item Figure~\ref{fig:heatmap_2}: Heatmaps of the power of derandomized knockoffs with different choices 
of $(c,\alpha_{\kn})$. The target FDR level $\alpha_{\ebh} = 0.2$, and the power is 
consistently high around $\alpha_{\kn} = 0.1$. 
The other details are the same as in Figure~\ref{fig:heatmap}.
\item Figure~\ref{fig:alpha_2}: Realized power (left) and FDR (right) of derandomized 
knockoffs as a function of the parameter $\alpha_{\kn}$ for the simulation 
data experiments.  The target FDR level $\alpha_{\ebh} = 0.2$. The details 
are otherwise the same as in Figure~\ref{fig:alpha}. 
\item Figure~\ref{fig:simulation_highdim}: Power, FDR,
and selection variability, for the high-dimensional simulated data experiment.
The other details are the same as in Figure~\ref{fig:simulation}.
\item Figure~\ref{fig:simulation_highdim_prob}: Top: the marginal selection probability $\hat{p}_j$
by original knockoffs versus that by derandomized knockoffs.
Bottom: the conditional selection probability $\hat{p}_{j,1}$
by original knockoffs versus that by derandomized knockoffs. 
The results are from simulations under the high-dimensional setting.
The details are otherwise the same as in Figure~\ref{fig:sel_linear}.
\item Figure~\ref{fig:original}: Power and FDR under the low-dimensional simulation setting (top)
and the high-dimensional simulation setting (bottom)
considered in~\citet{candes2018panning}. Shading indicates
error bars, and the results are averaged over 100 independent 
trials.
\item Figure~\ref{fig:roc_curve}: The dashed line represents the realized power of 
the PFER version as a function of the realized FDR; the rectangle 
represents the realized power and FDR resulting from the FDR version
(with the target FDR level $\alpha_{\ebh} = 0.1$).
\item Figure~\ref{fig:robustness}: Power (top) and FDR (bottom), for the Gaussian linear simulation
data experiments with estimated $P_X$. Each panel corresponds to a size
of unlabeled data for estimating $P_X$.
\item Figure~\ref{fig:multi}: Power and FDR, for the multi-environment simulated setting.
The other details are the same as in Figure~\ref{fig:simulation}. 
\item Figure~\ref{fig:sideinfo}: Power and FDR, for the simulated experiments with side information.
The other details are the same as in Figure~\ref{fig:simulation}.
\end{itemize}

\end{document}